\documentclass[a4paper,onecolumn,11pt,accepted=2024-10-05]{quantumarticle}

\def\papershowtoc{}
\pdfoutput=1
\usepackage[final]{hyperref}
\hypersetup{
           breaklinks=true,   
           colorlinks=true,   
           citecolor=red,
           allcolors={red}
        }

\usepackage{amsmath} 
\usepackage{amsfonts}
\usepackage{amssymb}
\usepackage{amsthm}
\usepackage{thmtools}
\usepackage{thm-restate}
\usepackage{braket}
\usepackage{algorithm2e}
\usepackage[]{nicefrac}
\usepackage{bbm}
\usepackage{multirow}
\usepackage{color}
\usepackage[normalem]{ulem}
\usepackage{commath,mathtools}
\usepackage{tikz}
\usetikzlibrary{quantikz}
\usepackage{caption}
\usepackage{fancyhdr}

\usepackage[utf8]{inputenc}

\usepackage{cancel}
\usepackage{graphicx}
\usepackage{varwidth}

\def \eps {\varepsilon}

\def\Pr{\mathrm{Pr}}

\newcommand{\nrm}[1]{\left\lVert #1 \right\rVert}
\newcommand{\Oo}{\ensuremath{\mathcal{O}}}
\newcommand{\paren}[1]{\left( #1 \right)}
\def\Tr{\mathrm{Tr}}

\newcommand{\polylog}[1]{\mathrm{polylog}\paren{#1}}

\newtheorem{theorem}{Theorem} 
\newtheorem{definition}[theorem]{Definition}
\newtheorem{lemma}[theorem]{Lemma}

\begin{document}

\title{Implementing any Linear Combination of Unitaries on Intermediate-term Quantum Computers} 

\author{Shantanav Chakraborty}  
\email{shchakra@iiit.ac.in}
\affiliation{CQST and CSTAR, International Institute of Information Technology Hyderabad, Telangana, India} 
\begin{abstract}

We develop three new methods to implement any Linear Combination of Unitaries (LCU), a powerful quantum algorithmic tool with diverse applications. While the standard LCU procedure requires several ancilla qubits and sophisticated multi-qubit controlled operations, our methods consume significantly fewer quantum resources. The first method (\textit{Single-Ancilla LCU}) estimates expectation values of observables with respect to any quantum state prepared by an LCU procedure while requiring only a single ancilla qubit, and no multi-qubit controlled operations. The second approach (\textit{Analog LCU}) is a simple, physically motivated, continuous-time analogue of LCU, tailored to hybrid qubit-qumode systems. The third method (\textit{Ancilla-free LCU}) requires no ancilla qubit at all and is useful when we are interested in the projection of a quantum state (prepared by the LCU procedure) in some subspace of interest. We apply the first two techniques to develop new quantum algorithms for a wide range of practical problems, ranging from Hamiltonian simulation, ground state preparation and property estimation, and quantum linear systems. Remarkably, despite consuming fewer quantum resources they retain a provable quantum advantage. The third technique allows us to connect discrete and continuous-time quantum walks with their classical counterparts. It also unifies the recently developed optimal quantum spatial search algorithms in both these frameworks, and leads to the development of new ones that require fewer ancilla qubits. Overall, our results are quite generic and can be readily applied to other problems, even beyond those considered here.
\end{abstract}
\newpage
\ifdefined\papershowtoc
\newpage
\tableofcontents
\newpage
\fi
\section{Introduction}
\label{sec:introduction}
We are currently in an era of quantum computing where theoretical advancements have been accompanied by drastic improvements in experimental capabilities \cite{arute2019quantum,zhong2020quantum,campagne2020quantum,madsen2022quantum,acharya2023supressing}. With rapid progress being made, it is reasonable to envision a stage in the near future, where quantum computing will transition away from the NISQ era \cite{preskill2018quantum,bharti2022noisy}. Quantum devices available immediately after the current NISQ stage, will most likely not have the capabilities of a large-scale, fully-programmable, fault-tolerant quantum computer. These devices, known as \textit{early fault-tolerant quantum computers} \cite{campbell2021early,Faehrmann2022randomizingmulti,dong2022ground,
zhang2022computing,lin2022heisenberg,wang2022quantum}, would have a limited number of logical qubits (restricting the availability of ancilla qubits), and short depth circuits with little to no multi-qubit controlled gates. On the other hand, for particular quantum technological platforms, it might be possible to engineer certain specific interactions more precisely, and for longer time-scales than others. For instance, it might be easier to engineer hybrid qubit-qumode systems in the intermediate-term  \cite{wallraff2004strong,pirkkalainen2013hybrid,kurizki2015quantum, andersen2015hybrid,campagne2020quantum,gan2020hybrid,sawaya2020resource} as many of the most promising quantum technological platforms such as superconducting systems \cite{wallraff2004strong}, ion-traps \cite{leibfried2003quantum}, and photonic systems \cite{gottesman2001encoding}, naturally have access to continuous variables. We refer to such devices, which will become available shortly after the current stage, as ``intermediate-term quantum computers''.  

It is thus crucial to develop quantum algorithms of practical interest that are implementable on intermediate-term quantum computers. Indeed, quantum algorithms tailored to early fault-tolerant quantum computers are already being developed \cite{campbell2021early,Faehrmann2022randomizingmulti,dong2022ground,
zhang2022computing,lin2022heisenberg,wang2022quantum}. With many quantum technological platforms vying for supremacy, it is  also essential to develop physically motivated quantum algorithms that can exploit the degrees of freedom that are naturally available to such platforms.  

There are only a handful of quantum algorithmic frameworks that can be applied to a diverse range of problems. However, most of these are only implementable on fully-fault tolerant quantum computers, which might be decades away. Linear Combination of Unitaries (LCU) is one such paradigm. Over the years, it has been widely applied and has been central to the development of a plethora of useful quantum algorithms ranging from Hamiltonian simulation \cite{childs2012hamiltonian, berry2014exponential,berry2015hamiltonian,berry2015simulating}, quantum linear systems \cite{childs2017quantum,chakraborty2019power} and differential equations \cite{berry2017quantum,liu2021efficient,childs2021high}, quantum walks \cite{apers2018FF,ambainis2019quadratic,apers2022quadratic}, ground state preparation \cite{ge2019faster,keen2021quantum,he2022quantum} and a large-class of optimization problems \cite{chowdhury2017quantum,van2020quantum}. 

The wide applicability of this procedure stems from the generic settings in which it can be applied. Given a Hermitian matrix $H$, and an initial state $\rho_0$, the LCU procedure implements any function $f(H)$ that can be well-approximated by a linear combination of unitaries, i.e.\ $f(H)\approx \sum_{j} c_j U_j$. Specifically, it prepares the quantum state
\begin{equation}
\label{eq:lcu-state}
\rho=\dfrac{f(H)\rho_0 f(H)^{\dag}}{\Tr[f(H)\rho_0 f(H)^{\dag}]},
\end{equation}
using $U_j$, controlled over multiple ancilla qubits \footnote{In order to prepare $\rho$, several rounds of amplitude amplification is also required}. In fact, for most of the applications mentioned previously, the problem boils down to applying a specific $f(H)$, to some initial state. Despite its broad applicability, LCU has its drawbacks when it comes to being implementable in the intermediate-term. First, for many problems of interest, there is a significant overhead in terms of the number of ancilla qubits needed. Second, the procedure requires implementing a sequence of sophisticated multi-qubit controlled-unitary operations, which is challenging for intermediate-term quantum computers. Furthermore, simply preparing the quantum state $\rho$ is often not useful. In most practical scenarios, we are interested in estimating some property of $\rho$, such as the expectation value of some observable $O$, i.e.\ $\Tr[O\rho]$. Extracting useful information about $\rho$ either requires additional runs of the LCU procedure or leads to even deeper quantum circuits.

In this work, we significantly enhance the applicability of the LCU framework to intermediate-term quantum devices. We develop three approaches that seek to reduce the resource required to implement any LCU. These methods either prepare the state $\rho$, or help extract useful information from it. Being considerably simpler than the standard LCU procedure, they are suitable for implementation using intermediate-term devices such as early fault-tolerant quantum computers and hybrid qubit-qumode systems. We apply each of them to develop quantum algorithms of practical interest.

Firstly, we develop a randomized quantum algorithm that estimates properties of the state $\rho$, prepared by any LCU procedure. More precisely, for any observable $O$, our algorithm estimates the quantity $\Tr[O\rho]$, to arbitrary accuracy. This technique, which we refer to as \textit{Single-Ancilla LCU}, requires only one ancilla qubit that acts as a control, and implements two (controlled) unitaries sampled according to the distribution of the LCU coefficients, followed by a single-shot measurement. By repeatedly running this simple short-depth quantum circuit, one obtains samples whose average converges to the expectation value we seek to estimate. Our procedure is suitable for early fault-tolerant quantum devices as it can implement any LCU (a) using only a single ancilla qubit, and (b) no multi-qubit controlled gates. In contrast, as mentioned previously, the standard LCU method requires several ancilla qubits and a series of sophisticated (multi-qubit) controlled operations. We rigorously compare the cost of implementing \textit{Single-Ancilla LCU} with the generic LCU procedure, and show that each coherent run of our method costs less, while requiring more classical runs. Furthermore, we apply our method to develop novel quantum algorithms for Hamiltonian simulation, estimating the properties of ground states of Hamiltonians, as well as quantum linear systems. 

Secondly, we develop \textit{Analog LCU}, a physically motivated, continuous-time analogue for implementing a linear combination of unitaries. This technique requires coupling the system Hamiltonian $H$ to a continuous-variable ancilla system (such as a one-dimensional quantum Harmonic oscillator), initialized in some easy-to-prepare continuous-variable quantum state (such as a Gaussian). The overall system is then evolved according to the resulting interaction Hamiltonian. Although this approach requires a continuous-variable ancilla register, the overall algorithm is considerably simpler than the standard LCU procedure. Moreover, this technique might be particularly useful for intermediate-term quantum computers (e.g. hybrid qubit-qumode systems) as such interactions can already be engineered on several quantum technological platforms. Examples of discrete systems coupled to continuous-variable ones include ion traps and superconducting systems \cite{wallraff2004strong,pirkkalainen2013hybrid,andersen2015hybrid, kurizki2015quantum,gan2020hybrid}. We show that this approach can be used to develop novel analog quantum algorithms for ground state preparation and solving quantum linear systems.

Suppose for a specific problem, we are interested in the projection of the LCU state $f(H)\ket{\psi_0}$ in some subspace, and it suffices to ensure that the measurement is successful, on average. In such scenarios, we show that the ancilla registers can be dropped entirely. We call this the \textit{Ancilla-free LCU} technique. This approach involves randomly sampling the unitaries $U_j$ according to the distribution of the LCU coefficients $c_j/\norm{c}_1$ without any ancilla registers. On average, this prepares some density matrix for which the projection in this subspace can be proven to be at least as large. This scenario arises in the context of quantum spatial search algorithms: the problem of finding an element in a marked subset of nodes of any ergodic, reversible, Markov chain. Indeed, the goal is to prove a generic quadratic speedup over classical random walks, for which, the expected number of steps to solve this problem is known as the hitting time ($HT$). This problem has only recently been resolved using generic LCU \cite{ambainis2019quadratic}, following a long line of works that provided speedups in particular cases \cite{szegedy2004quantum, magniez2007search, krovi2016quantum}. Consequently, we use \textit{Ancilla-free LCU} to design optimal spatial search algorithms, also placing recent results \cite{ambainis2019quadratic, apers2019unified, apers2022quadratic} within this framework, along with developing new ones. As compared to quantum spatial search algorithms using \textit{Standard LCU}, our methods achieve the same generic quadratic speedup while requiring $O(\log HT)$ fewer ancilla qubits. 

In addition to providing a unified framework for quantum spatial search, \textit{Ancilla-free LCU} also allows us to establish a relationship between the different frameworks of classical random walks and quantum walks. Finally, in order to complete the picture, we also establish a connection between discrete and continuous-time quantum walks by using the frameworks of block encoding \cite{low2019hamiltonian, chakraborty2019power} and quantum singular value transformation (QSVT) \cite{gilyen2018quantum, gilyen2019quantum,martyn2021grand}.

The paper is organized as follows. In the rest of this section, we provide a brief overview of the main results in Sec.~\ref{sec:contributions}, and also relate our work to prior results in Sec.~\ref{subsec:prior-work}. In Sec.~\ref{sec:preliminaries}, we review some basic definitions and techniques that we will be using in this article. We formally describe the three different approaches to implementing LCU in Sec.~\ref{sec:new-approaches-LCU}. The rest of the article involves applying these techniques to develop new quantum algorithms. In Sec.~\ref{sec:ham-sim}, we apply the \textit{Single-Ancilla LCU} method to develop a novel quantum algorithm for Hamiltonian simulation. In Sec.~\ref{sec:gsp}, we make use of our techniques to develop new quantum algorithms for ground state preparation of Hamiltonians (Sec.~\ref{subsec:analog-gsp}) and also for ground state property estimation (Sec.~\ref{subsec:gsp-single-ancilla}). In Sec.~\ref{sec:qls}, we develop novel analog quantum linear systems algorithms, tailored to hybrid qubit-qumode systems (Sec.~\ref{subsec:analog-qls}) and also use \textit{Single-Ancilla LCU} to estimate expectation values with respect to the solution of quantum linear systems (Sec.~\ref{subsec:single-ancilla-qls}). In Sec.~\ref{sec:quantum-random-walks}, we apply \textit{Ancilla-free LCU} we establish a relationship between different frameworks of classical and quantum walks by developing optimal quantum spatial search algorithms that reduce the number of ancilla qubits needed, also placing recently developed algorithms within this framework. Finally, we conclude and discuss possible future research directions in Sec.~\ref{sec:discussion}.
\subsection{Summary of our results}
\label{sec:contributions}
In this section, we state the main results of this article. We begin by briefly outlining each of the three variants of implementing LCU and the applications we consider. We summarize them in Fig.~\ref{fig:summary-relationship}.

\begin{figure}[h!]
\centering
\includegraphics[scale=0.47]{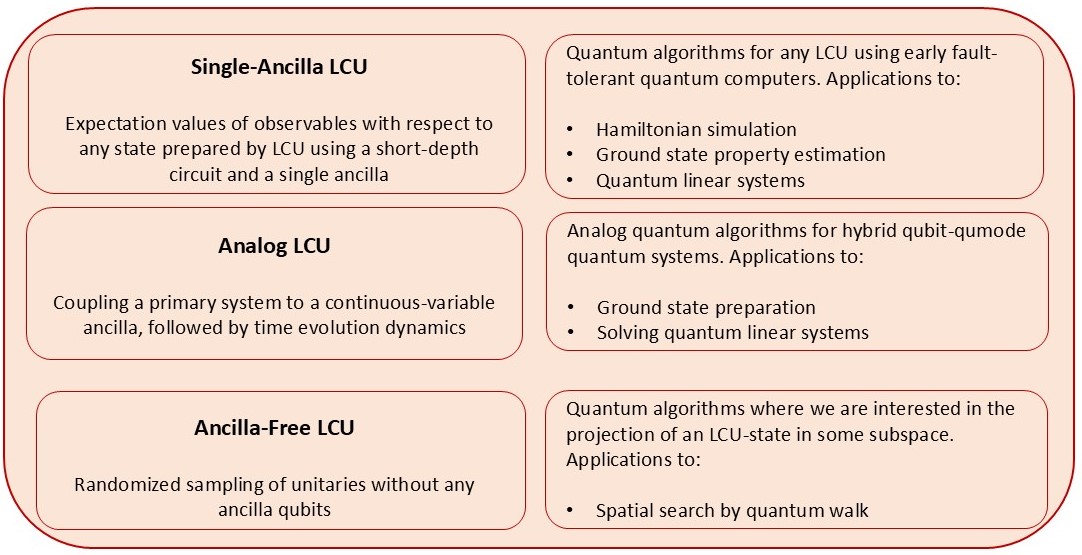}
\caption{\textit{Summary of the main results -- The three approaches to LCU and their applications}}
\label{fig:summary-relationship}
\end{figure}
\subsubsection{Single-Ancilla LCU: Estimating expectation values of observables}
\label{subsec:intro-single-ancilla}
Given any initial state $\rho_0$, and Hamiltonian $H$, we develop a randomized quantum algorithm that estimates the expectation value
\begin{equation}
\label{eq:expectation-lcu-state}
\Tr[O\rho]=\dfrac{\Tr[O~f(H)\rho_0f(H)^{\dag}]}{\Tr[f(H)\rho_0 f(H)^{\dag}]},
\end{equation} 
to arbitrary accuracy, for any function $f$ that can be approximated by a linear combination of unitaries, i.e.\ $f(H)\approx \sum_j c_j U_j$. For this task, the standard LCU procedure requires several ancilla qubits, and sophisticated (multi-qubit) controlled unitaries. In contrast, for quantum algorithms tailored to early fault-tolerant quantum devices, both the number of ancilla qubits, as well as the number of multi-qubit controlled operations should be as low as possible. Given these restrictions, the priority of such algorithms is to reduce the cost of each coherent run, even if this results in an increase in the number of classical repetitions required (and hence, the overall cost). As quantum coherence need not be maintained across multiple runs, this leads to having multiple repetitions of a low-cost quantum circuit (in terms circuit/gate depth), which is preferred in the intermediate regime.

The \textit{Single-Ancilla LCU} technique estimates the expectation value in Eq.~\eqref{eq:expectation-lcu-state} while satisfying the aforementioned features: (a) it uses only a single ancilla qubit, (b) requires no multi-qubit controlled operations, and (c) despite restrictions (a) and (b), the cost of each coherent run is lower than the \textit{Standard LCU} procedure. However, it requires more classical repetitions as compared to the generic LCU technique, and hence has a higher overall cost. To this end, we develop a randomized quantum algorithm makes use of the quantum circuit of Faerhmann et al. \cite{Faehrmann2022randomizingmulti} (shown in Fig.~\ref{fig:single-ancilla-circuit}), wherein the authors used it to generate randomized multi-product formulas. For any $f(H)$ that can be approximated by an LCU, our method repeatedly samples from this circuit to estimate the numerator of Eq.~\eqref{eq:expectation-lcu-state}. Note that the denominator is apriori unknown. We show that the knowledge of any rudimentary lower bound $\ell_*$ of this quantity allows us to leverage the same algorithm to estimate it \footnote{Apriori knowledge of $\ell_*$ is application specific. For instance, the minimum eigenvalue of $f(H)$ can be such a lower bound. Note that this information is also needed in case of \textit{Standard LCU}}. Overall, our algorithm separately estimates both the numerator as well as the denominator and we rigorously calculate the accuracy with both these quantities need to be estimated so that their ratio is $\varepsilon$-close the expectation value we seek to estimate. Our overall procedure, and its and its correctness has been proven in detail in Sec.~\ref{subsec:one-ancilla-LCU}. Here, we state the general result informally. 
~\\
\begin{figure}[h!]
\centering
\begin{quantikz}
\lstick{$\ket{+}$} & \qw &\ctrl{1} & \octrl{1} & \gate[2]{X\otimes O} \\
\lstick{$\rho_0$}    & \qwbundle{} &\gate{V_1} & \gate{V_2} &
\end{quantikz}
\caption{\small\textit{Quantum circuit corresponding for the \textit{Single-Ancilla LCU} procedure. For $f(H)\approx \sum_{j}c_j U_j$, repeated runs of this short-depth quantum circuit can estimate $\Tr[O~f(H)\rho_0 f(H)^{\dag}]/\Tr[f(H)\rho_0 f(H)^{\dag}]$, to arbitrary accuracy. For this, $V_1$ and $V_2$ are sampled at random according to $\mathcal{D}\sim \{c_j/\norm{c}_1, U_j\}$. Each run of the circuit outputs a random variable corresponding to the outcome of the measurement of the observable $X\otimes O$. Overall we need to repeat this circuit $T$ times, such that the sample mean of the $T$ observations to converge to the desired estimate.}}
\label{fig:single-ancilla-circuit}
\end{figure}
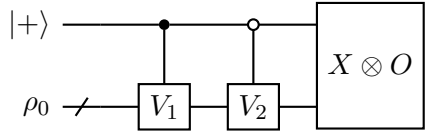
~\\

\begin{theorem}[Informal]
\label{thm:informal-randomized-time-evolution}
Let $O$ be some observable and $\rho_0$ be some initial state, prepared in cost $\tau_{\rho_0}$. Suppose there exists a Hermitian matrix $H\in\mathbb{C}^{N\times N}$, and a function $f:[-1,1]\mapsto\mathbb{R}$ such that $f(H)\approx \sum_{j} c_j U_j$, where $\norm{c}_1=\sum_{j}|c_j|$, and each $U_j$ is implemented with cost $\tau_j$. Define $\langle\tau\rangle=\sum_{j}c_j\tau_j/\norm{c}_1$. Also, suppose we know some $\ell_*$ such that $\ell^2=\Tr[f(H)\rho_0 f(H)^{\dag}]\geq \ell_*$. Then there exists a procedure that outputs $\mu$ and $\tilde{\ell}$ such that
$$
\left|\mu/\tilde{\ell} - \Tr[O\rho]\right| \leq \varepsilon,
$$
with a constant probability, using only one ancilla qubit and 
$$
T=O\left(\dfrac{\norm{O}^2 \norm{c}_1^4}{\varepsilon^2\ell_*^2}\right) 
$$  
repetitions of the quantum circuit in Fig.~\ref{fig:single-ancilla-circuit}, where the average cost of each such run is $2\langle \tau\rangle+\tau_{\rho_0}$. 
\end{theorem}
~\\~\\
We compare the performance of \textit{Single-Ancilla LCU} with \textit{Standard LCU} in detail (See Sec.~\ref{subsec:one-ancilla-LCU} and Table \ref{table:comparison-standard-LCU}). We show that if implementing each $U_j$ costs $\tau_j$, and preparing the initial state $\rho_0$ costs $\tau_{\rho_0}$, the average cost of each coherent run of our algorithm is $2\langle \tau \rangle + \tau_{\rho_0}$. On the other hand, \textit{Standard LCU} requires implementing a \textit{prepare} gate $R$ and a multi-qubit controlled \textit{select} unitary $Q$, requiring cost $\tau_R$ and $\tau_Q$ respectively. So the total cost is $O(\tau_R+\tau_Q+\tau_{\rho_0})$. Then just $\tau_Q$ is at least as high as $\langle \tau \rangle$. Thus, despite requiring a solitary ancilla qubit and no multi-qubit controlled operations, the cost of each coherent run of our algorithm is lower than \textit{Standard LCU}, given both procedures implement $U_j$ with the same cost. However, the number of classical repetitions (and hence, the overall cost) required for \textit{Standard LCU} is lower. Moreover, being suitable for fully fault-tolerant quantum computers, \textit{Standard LCU} can also leverage procedures such as quantum amplitude amplification and estimation \cite{brassard2002quantum} to estimate $\Tr[O\rho]$ coherently, which further reduces the overall cost by requiring fewer classical repetitions, while increasing the cost of each coherent run substantially. However, procedures such as quantum amplitude amplification and estimation are too involved to be implemented in the intermediate-term. Next, we discuss the applications of our method.
~\\~\\
\textbf{Applications:~} We apply \textit{Single-Ancilla LCU} to several problems of practical interest such as Hamiltonian simulation, ground state property estimation and quantum linear systems. Throughout the article, for each of the applications of \textit{Single-Ancilla LCU}, we calculate $T$, the number
 of repetitions needed of the circuit in Fig.~\ref{fig:single-ancilla-circuit} as well as an upper bound on the cost of each coherent run, given by $\tau_{\max}$. The overall cost would then be the product of these quantities that is $O(\tau_{\max}\cdot T)$. 

For all our applications, we consider some Hamiltonian $H$, and implement a specific $f(H)\approx \sum_{j}c_j U_j$. However, depending on the application, how we can access $H$ (input model) varies. This also determines what the costs, $\tau_{\max}$ and $T$ characterize. Before moving on to the specific applications, we discuss them briefly:
~\\
\begin{itemize}
\item For Hamiltonian simulation, we assume $H$ is a linear combination of strings of Pauli matrices, which can be accessed and implemented directly. We characterize the cost in terms of the gate depth per coherent run ($\tau_{\max}$) as well as the overall gate depth $O(T\cdot \tau_{\max})$.

\item For both ground state property estimation as well as quantum linear systems, we assume that we can access the Hamiltonian $H$ through the time evolution operator $U_{t}=\exp[-itH]$. Furthermore, given access to $U_{t}$, we can perform the time evolution controlled on a single ancilla qubit. This is referred to as the Hamiltonian evolution (HE) model as has been used in prior results specific to ground energy estimation using early fault-tolerant quantum computers \cite{lin2022heisenberg,zhang2022computing,wang2022quantum,dong2022ground}. Much like these works, we calculate: (a) the maximal time of evolution of $H$ (controlled by a single ancilla qubit) required in each coherent run, which will be denoted by $\tau_{\max}$, and (b) the total number of repetitions of the circuit $T$. The total evolution time is then $O(\tau_{\max}\cdot T)$. Note that $\tau_{\max}$ is different from the actual circuit depth. In fact, this relationship depends on how the time evolution is performed. If $U_{t}$ can be performed exactly in time $O(t)$, then the circuit depth scales linearly with the maximal evolution time. However, if $U_t$ is implemented by a Hamiltonian simulation algorithm, then the circuit depth depends on the particular choice of the algorithm. Recall that in the early fault-tolerant regime, we are limited by a small ancilla qubit space and the inability to perform multi-qubit controlled operations. This restricts the choice for the underlying simulation algorithm. In Appendix \ref{sec-appendix:single-ancilla-LCU}, we indeed characterize the complexities of both our algorithms in terms of the circuit depth (more precisely, the gate depth) by choosing particular Hamiltonian simulation algorithms that are suited to this regime.
\end{itemize}

Having discussed the various access models we consider, we move on to the specific applications to which we apply \textit{Single-Ancilla LCU}:

\begin{itemize}
\item[(a)] \textbf{Hamiltonian simulation:~} Consider any Hamiltonian $H$ such that it is expressed as a linear combination of strings of Pauli operators $P_k$. That is, $H=\sum_{k=1}^{L} \lambda_k P_k$, such that $\beta=\sum_{k=1}^{L}|\lambda_k|$. If $\rho_0$ is some initial state prepared in cost $\tau_{\rho_0}$, then our randomized quantum algorithm outputs a parameter $\mu$, with probability at least $1-\delta$, such that 
$$\left|\mu-\Tr[O~e^{-iHt}\rho_0 e^{iHt}]\right|\leq \varepsilon,$$ 
for any measurable observable $O$. Each coherent run of our algorithm makes use of the quantum circuit in Fig.~\ref{fig:single-ancilla-circuit} which uses only a single ancilla qubit. The gate depth is at most $\tau_{\rho_0}+2\tau_{\max}$, where
$$
\tau_{\max}=O\left(\beta^2t^2\dfrac{\log(\beta t\norm{O}/\varepsilon)}{\log\log(\beta t\norm{O}/\varepsilon)}\right).
$$ 
Overall,
$$
T=O\left(\dfrac{\norm{O}^2\ln(1/\delta)}{\varepsilon^2}\right)
$$ 
classical repetitions of this quantum circuit is needed. For this, we decompose $f(H)=e^{-itH}$ as a linear combination of unitaries using ideas from the truncated Taylor series approach \cite{berry2015simulating}, as well as from \cite{wang2022quantum}. We describe our method in detail in Sec.~\ref{sec:ham-sim}. The overall gate depth is given by $T\cdot\tau_{\max}=\widetilde{O}\left(\beta^2t^2\|O\|^2/\varepsilon^2\right)$.

Our method outperforms all first order Trotter methods \cite{childs2021theory} and their randomized variants \cite{childs2018toward, zhao2022hamiltonian}, requiring an exponentially shorter gate depth per coherent run (in terms of the precision $1/\varepsilon$), as well as a shorter overall gate depth. Our algorithm also uses fewer ancilla qubits than the Truncated Taylor series method \cite{berry2015simulating}, which additionally makes multiple uses of complicated subroutines such as oblivious amplitude amplification. The gate depth per coherent run of this algorithm is quadratically better than our method in terms of $\beta$ and $t$, but has a linear dependence on $O(L)$, which implies that there are Hamiltonians (satisfying $\beta\ll L$) for which our method provides an advantage. The state-of-the-art Hamiltonian simulation procedure by Low and Chuang \cite{low2019hamiltonian} is optimal in terms of the number of queries to a block encoding of $H$. However, the construction of the block encoding requires $O(\log L)$ ancilla qubits and multi-qubit controlled operations, and hence is not implementable in the early fault-tolerant regime. This also adds an overhead of $O(L)$ to the gate depth per coherent run. We compare our algorithm rigorously with other methods in Sec.~\ref{sec:ham-sim}. A comparison of the complexities are summarized in Table \ref{table:comparison-ham-sim}.
~\\
\item[(b)] \textbf{Ground state property estimation:~} For any Hamiltonian $H$ with unknown ground state $\ket{v_0}$, and any measurable observable $O$, we provide a randomized quantum algorithm that outputs an $\epsilon$-additive accurate estimate of the expectation value of $O$ with respect to $\ket{v_0}$, i.e.\ $\braket{v_0|O|v_0}$. For this, we make the following standard assumptions: (i) a lower bound on spectral gap of $H$ is known (say $\Delta$), (ii) an initial state $\ket{\psi_0}$ having overlap of at least $\eta$ with $\ket{v_0}$, can be prepared in cost $\tau_{\psi_0}$, (iii) the ground energy of $H$ is known to precision $O(\Delta\cdot(\log (\varepsilon^{-1}\eta^{-1}))^{-1/2})$. 
Our algorithm involves expressing the function $f(H)=e^{-tH^2}$ as a LCU: we show $f(H)=\sum_{j} c_j e^{-ij\sqrt{2t}H}$. For a judiciously chosen value of $t$, $f(H)\ket{\psi_0}$, has the effect of preserving the component of $\ket{\psi_0}$ in the direction of $\ket{v_0}$ while exponentially suppressing all other components that are orthogonal to it, resulting in a state that is close to $\ket{v_0}$. In the Hamiltonian evolution model ($H$ is accessed via the time evolution operator), we show that the \textit{Single-Ancilla LCU} algorithm estimates $\braket{v_0|O|v_0}$ with additive accuracy $\varepsilon$, with probability at least $(1-\delta)^2$, using
$$
T= O\left(\dfrac{\norm{O}^2\ln(1/\delta)}{\varepsilon^2\eta^4}\right),
$$ 
repetitions of the quantum circuit in Fig.~\ref{fig:single-ancilla-circuit} and only a single ancilla qubit. The maximal time evolution of $H$ in any coherent run is at most $2\tau_{\max}+\tau_{\psi_0}$, where
$$
\tau_{\max}=O\left(\dfrac{1}{\Delta}\log\left(\dfrac{\norm{O}}{\varepsilon\eta}\right)\right).
$$ 
The total evolution time is then $O(T.\tau_{\max})=\widetilde{O}(\Delta^{-1}\eta^{-4}\|O\|^2/\varepsilon^2)$. The overall method and its correctness has been formally stated via Theorem \ref{thm:ground-state-estimation}. We also compare our method with other algorithms in detail (See Table \ref{table:comparison-gsp}). The \textit{Standard LCU} procedure \cite{ge2019faster,keen2021quantum} requires $O(\log(1/\Delta)+\log(~\norm{O}\eta^{-1}\varepsilon^{-1}))$ ancilla qubits
and sophisticated multi-qubit controlled operations. Despite this, in the Hamiltonian Evolution access model, the maximal time evolution of $H$, for \textit{Standard LCU} is never better than our method. However, \textit{Standard LCU} requires fewer classical runs (by a factor of $1/\eta^2$) and hence has a lower total evolution time. It can also make use of involved subroutines such as quantum amplitude amplification and estimation to further lower the total evolution time. However these methods substantially increase the maximal time of evolution of $H$. Furthermore, it also increases the number of ancilla qubits needed. Our algorithm also compares favourably to other early fault-tolerant quantum algorithms \cite{zhang2022computing, dong2022ground}. 
~\\
\item[(c)]~ \textbf{Quantum linear systems:~} Suppose we have a Hermitian matrix $H$ such that its eigenvalues in $[-1, -1/\kappa]\cup [1/\kappa, 1]$, such that $\kappa$ is an upper bound on the ratio between the maximum and minimum eigenvalue of $H$ (condition number). Let us assume that the initial quantum state $\ket{b}$ can be prepared in cost $\tau_b$. Then, we show use \textit{Single-Ancilla LCU} to estimate $\braket{x|O|x}$ up to $\varepsilon$ additive accuracy, with probability at least $(1-\delta)^2$, using
$$
T=O\left(\dfrac{\norm{O}^2\kappa^4\log^2\left(\dfrac{\norm{O}\kappa}{\varepsilon}\right)\ln(1/\delta)}{\varepsilon^2}\right),
$$ 
repetitions of the quantum circuit in Fig.~\ref{fig:single-ancilla-circuit} and only one ancilla qubit. The maximal time evolution of $H$ is at most $2\tau_{\max}+\tau_b$, where
$$
\tau_{\max}=O\left(\kappa\log\left(\dfrac{\norm{O}\kappa}{\varepsilon}\right)\right).
$$
The total evolution time is given by $\widetilde{O}(\kappa^5\|O\|^2/\varepsilon^2)$. This approach makes use of the LCU decomposition of $f(H)=H^{-1}$ in Ref.~\cite{childs2017quantum}. We analyze the correctness of our method via Theorem \ref{thm:inversion-estimation-2}. The Childs, Kothari and Somma algorithm \cite{childs2017quantum}, which makes use of \textit{Standard LCU} requires $O(\kappa\log(\kappa\norm{O})/\varepsilon)$ ancilla qubits and sophisticated multi-qubit controlled operations. Despite this, the maximal time evolution of $H$ is never better than our method. The number of classical repetitions (and hence the total evolution time), however scales better than our method. The quantum linear systems algorithm using QSVT \cite{gilyen2019quantum} requires access to a block encoding of $H$. This has near optimal complexity in terms of the number of queries to the block encoding. However, constructing a block encoding of $H$ is resource consuming and is not implementable in the intermediate-term. This is also the case with the state-of-the-art algorithm of Costa et al. \cite{costa2022optimal} which has optimal query complexity on all parameters (the query depth per coherent run is $O(\log \kappa)$ better than the maximal time evolution of $H$ in our case). In Sec.~\ref{subsec:single-ancilla-qls}, we compare our method in detail with other algorithms for solving quantum linear systems (See Table \ref{table:comparison-qls}).  
\end{itemize}

As discussed previously, we analyzed the complexity of both the ground state property estimation algorithm as well as the quantum linear systems algorithm in the Hamiltonian Evolution input model, which is indeed the case for other early fault-tolerant quantum algorithms. In this case, $\tau_{\max}$ measures the maximal time evolution of $H$, while $O(\tau_{\max}.T)$ is the total evolution time. However, this is different from the actual circuit depth of the algorithm, for which one needs to specify how the time-evolution operator $U_t=\exp[-itH]$ is implemented. If this is performed by a Hamiltonian simulation algorithm, both the circuit depth per coherent run, as well as the overall circuit depth depends on the choice of the underlying simulation algorithm. In the early fault-tolerant regime, we intend to leverage algorithms that do not add any overhead in terms of the number of ancilla qubits or multi-qubit controlled gates. This limits the Hamiltonian simulation algorithms that can be implemented in the early fault-tolerant era. For instance, state-of-the-art Hamiltonian simulation algorithms \cite{low2017hamiltonian,low2019hamiltonian} require access to a block encoding of $H$: a unitary where $H$ is in the top left block subnormalized by some factor (See Sec.~\ref{subsec:prelim-block-encoding-qsvt}). This is resource demanding as it increases the number of ancilla qubits, as well as multi-qubit controlled operations and hence are unsuitable for early fault-tolerant quantum devices.

On the other hand, if $H$ can be expressed as a linear combination of Pauli operators, both Trotter-based methods \cite{childs2021theory} as well as the Hamiltonian simulation algorithm by \textit{Single-Ancilla LCU} can be incorporated into both our algorithms. This does not require any additional ancilla qubits or multi-qubit controlled gates. However, both these methods have a suboptimal dependence on $t$, which increases the circuit depth of both our ground state property estimation and quantum linear systems algorithms.    

In Appendix \ref{subsec-app:gsp-qls}, we analyze the circuit depth (in terms of the gate depth per coherent run, as well as the overall gate depth) of both these algorithms while using (a) Hamiltonian simulation by \textit{Single-Ancilla LCU} and (b) $2k$-order Trotter \cite{childs2021theory}. For this we assume that:

\begin{itemize}
\item [(i)] $H$ is a linear combination of strings of Pauli operators, i.e. $H=\sum_{j=1}^{L}\lambda_j P_j$, with $\beta=\sum_{k=1}^{L}|\lambda_k|$.
\item [(ii)] The observable $O$ to be measured is also a linear combination of easy-to-implement unitary observables, i.e.\ $O=\sum_{j=1}^{L_O} h_j O_j$ such that $\norm{O_j}=1$ and $\norm{h}_1=\sum_{j}|h_j|$. This is motivated by the fact that such observables are ubiquitous across condensed matter physics \cite{sachdev2011quantum} and quantum chemistry \cite{ mcardle2020quantum}.
\end{itemize}

We demonstrate that under these assumptions both algorithms can still be performed using only a single ancilla qubit and no multi-qubit controlled operations (See Table \ref{table:circuit-depth-single-ancilla-gsp-qls}). We comprehensively compare our methods with other quantum algorithms all of which require multiple ancilla qubits and multi-qubit controlled operations.  Furthermore, we show that despite this, there are regimes where both our algorithms require a shorter gate depth per coherent run, as compared to even state-of-the-art quantum algorithms (See Table \ref{table:comparison-gsp-depth} for ground state property estimation and Table \ref{table:comparison-qls-depth} for quantum linear systems). Next, we discuss our results in the \textit{Analog LCU} framework.

\subsubsection{Analog LCU: Coupling discrete systems with continuous variable systems} 
We develop a more physical model for LCU in continuous-time. Consider any Hamiltonian $H$, consider any $f(H)$ that can be well approximated by a truncated Fourier transform, i.e.,
$$
\norm{f(H)-\int_{a}^{b}~dz~c(z) \cdot e^{-iHzt}}\leq \varepsilon,
$$ 
where $c:\mathbb{R}\mapsto \mathbb{R}\backslash \{0\}$. Then by a purely continuous-time procedure, for any initial state $\ket{\psi_0}$, we can prepare a state that $O\left(\varepsilon/\norm{c}_1\right)$-close to $f(H)\ket{\psi_0}/\norm{f(H)\ket{\psi_0}}$,  
where $\norm{c}_1=\int_{a}^{b} dz~|c(z)|$. 

This requires coupling the primary system to a continuous variable ancilla (such as a one-dimensional quantum Harmonic oscillator), prepared in a continuous variable state. We show in this work that for several applications, this state is easy to prepare (such as a Gaussian). The overall system is then evolved according to the interaction Hamiltonian $H'=H\otimes \hat{z}$ for an appropriate time $T$, which prepares the desired LCU state in the first register. 

The technique is considerably simpler than its discrete-time counterpart.  Furthermore such hybrid qubit-qumode interactions can be implemented in a number of quantum technological platforms such as trapped ions, Cavity (or Circuit QED), photonic systems, and superconducting qubits \cite{wallraff2004strong,pirkkalainen2013hybrid,andersen2015hybrid, kurizki2015quantum,gan2020hybrid}. Our motivation is not only to provide an alternative approach to implementing LCU but make this paradigm more implementable for intermediate-term quantum devices.
\\~\\
\textbf{Applications:~} We apply these techniques to develop analog quantum algorithms for ground state preparation and for solving quantum linear systems. As mentioned previously, our aim is to couple the system Hamiltonian with an ancillary continuous-variable system.

\begin{itemize}
\item[(a)]~\textbf{Ground state preparation:~} Given a Hamiltonian $H$, we couple this system with a continuous variable ancillary system via the interaction Hamiltonian $H'=H\otimes \hat{z}$. The ancilla system is prepared in an easy-to-prepare continuous variable state, namely a Gaussian. We show that given an initial state $\ket{\psi_0}$ that has an overlap of at least $\eta$ with the ground state, simply evolving the system according to $H'$ results in a state proportional to $f(H)\ket{\psi_0}$ in the first register, where $f(H)=e^{-tH^2}$. We show that, with probability $\eta^2$, this state is $\varepsilon$-close to the ground state of $H$ (provided its ground energy is known up to some precision). The overall time required is
$$
T=O\left(\dfrac{1}{\Delta}\sqrt{\log\left(\dfrac{1}{\eta\varepsilon}\right)}\right),
$$ 
where $\Delta$ is the spectral gap of $H$ (Lemma \ref{lem:analog-gsp-ham}). This quantum algorithm appeared in \cite{apers2022quadratic} and also independently in Ref.~\cite{he2022quantum}. Here we place this in the context of \textit{Analog LCU}: it provides useful intuition for (i) the quantum linear systems algorithms we develop using similar techniques and (ii) the \textit{Single-Ancilla LCU} method for ground state property estimation.
\\
\item[(b)]~\textbf{Quantum linear systems:~} We provide two quantum algorithms for this problem. For both these problems, we couple $H$ to two ancillary continuous variable systems (Harmonic oscillators), i.e. $H'=H\otimes \hat{y}\otimes \hat{z}$.  The first approach works for any Hermitian matrix $H$ with eigenvalues in the domain $[-1,-1/\kappa]\cup [1/\kappa,1]$, where $\kappa$ is an upper bound on the condition number (ratio between the maximum and the minimum non-zero eigenvalue) of $H$. The first register is prepared in the initial state $\ket{b}$, the second register is prepared in the first excited state of the quantum Harmonic oscillator, while the third register is prepared in the ground state of a ``particle in a ring'' of unit radius \cite{griffiths2018introduction}. 

This algorithm (see Sec.~\ref{subsec:analog-qls}) can be seen as an analog variant of the quantum linear systems algorithm of Childs, Kothari and Somma \cite{childs2017quantum}. In order to obtain a quantum state that is $\varepsilon/\kappa$-close to $\ket{x}=H^{-1}\ket{b}/\|H^{-1}\ket{b}\|$ in the first register, with overlap at least $1/\kappa$, we require evolving the system according to $H'$ for a time 
$$
T=O\left(\kappa\sqrt{\log\left(\dfrac{\kappa}{\varepsilon}\right)}\right).
$$

Typically in continuous variable systems, Gaussian states are easier to prepare, and engineer \cite{weedbrook2012gaussian}. Thus, we also provide an analog quantum algorithm for solving quantum linear systems (for positive semidefinite Hamiltonians) in which both the ancilla registers are now prepared in Gaussian states. Evolving this system according to $H'$ prepares a state that is $(\varepsilon/\kappa)^{3/2}$ - close to $\ket{x}$, with overlap $\Omega(1/T)$ in time
$$
T=O\left(\dfrac{\kappa^{3/2}}{\sqrt{\varepsilon}}\right).
$$
Although the complexity is worse than the first approach, this quantum algorithm requires preparing only Gaussian states, which we expect to be easier for intermediate-term quantum computers to implement.
\end{itemize}
\subsubsection{Ancilla - free LCU:~Randomized sampling of unitaries}
Suppose for some Hermitian matrix $H\in \mathbb{C}^{N\times N}$, we intend to implement $f(H)$ such that
$$
\norm{f(H)-\sum_{j=1}^{M}c_j U_j}\leq \gamma,
$$
where $\gamma\in [0,1)$, $c_j\in \mathbb{R}$ and $U_j$ is some unitary. Furthermore, we are interested in the projection of $f(H)\ket{\psi_0}$ in some subspace of interest and for the underlying problem, it suffices to ensure that the projective measurement is successful, on average. Then, dropping the ancilla register altogether and simply sampling $V$ according to $\mathcal{D}\sim \{c_j/\norm{c}_1, U_j\}$, results in the following mixed state, on average
$$
\bar{\rho} = \mathbb{E}[V\rho_0 V^{\dag}]=\dfrac{1}{\norm{c}_1} \sum_{j=1}^{M} c_j U_j \rho_0 U^{\dag}_j,
$$
where $\norm{c}_1=\sum_{j=1}^{M} |c_j|$. Then, if $\norm{c}_1\leq 1$, the average projection of this density matrix $\rho=V\rho_0 V^{\dag}$ on the subspace of interest is at least as large as $f(H)\rho_0 f(H)^{\dag}$. That is, if $\Pi$ is the projector on to this subspace, 
$$\mathbb{E}[\Tr(\Pi \rho)]=\Tr[\Pi \bar{\rho}]\geq \Tr[\Pi f(H)\rho_0 f(H)^{\dag}]. 
$$
We call this technique \textit{Ancilla-free LCU} as it does not require any ancilla qubits, by avoiding the need to prepare $f(H)\rho_0 f(H)^{\dag}$, by \textit{Standard LCU}. Formally, we prove the following theorem:
\begin{restatable}[Randomized unitary sampling]{theorem}{thmbodyunisamp}
  \label{thm:ancilla-free-LCU-theorem}
  Let $H\in \mathbb{C}^{N\times N}$ is a Hermitian matrix. Also let $\varepsilon\in (0,1)$ and suppose $f:[-1,1]\mapsto\mathbb{R}$ be some function such that
  $$
  \norm{f(H)-\sum_{j=1}^{M} c_j U_j}\leq \dfrac{\varepsilon}{3\norm{f(H)}},
  $$
 for some unitaries $U_j$ and $c_j\in \mathbb{R}\backslash\{0\}$ such that $\norm{c}_1=\sum_{j}|c_j|\leq 1$. Suppose $V$ is a unitary sampled from the ensemble $\{c_j/\norm{c}_1, U_j\}$, and applied to some initial state $\rho_0$. Then, the average density matrix, defined as
$$
\bar{\rho} = \mathbb{E}\left[V\rho_0V^{\dag}\right]=\dfrac{1}{\norm{c}_1} \sum_{j=1}^{M} c_j U_j \rho_0 U^{\dag}_j,
$$
satisfies,
$$
\Tr\left[\Pi\bar{\rho}\right]\geq \Tr[\Pi f(H)\rho_0 f(H)^{\dag}] -\varepsilon,
$$
for any projector $\Pi$ acting on the state space of $\bar{\rho}$.
\end{restatable}  
~\\~\\
We prove Theorem \ref{thm:ancilla-free-LCU-theorem} in Sec.~\ref{subsec:ancilla-free-LCU}. The result of Theorem \ref{thm:ancilla-free-LCU-theorem} can be interpreted as follows: Suppose we are interested in the projection of $f(H)\ket{\psi_0}$ in some subspace of interest, and in the average projection of the resulting state in this subspace. Then, instead of implementing the standard LCU procedure, we can simply sample a $V$ as described above and apply it to the initial state $\ket{\psi_0}$. On average, this leads to the density matrix $\bar{\rho}$. Theorem \ref{thm:ancilla-free-LCU-theorem} states that the projection of $\bar{\rho}$ in this subspace is guaranteed to be is at least as large. In other words, the average success probability of the projective measurement on $\rho=V\rho_0 V^{\dag}$ would be at least as high even when \textit{Standard LCU} is replaced with just importance sampling. This also extends to the continuous-time setting. 

Thus while \textit{Standard LCU} requires $\lceil\log M\rceil$ ancilla qubits, our method requires none. However, it is important to note that Theorem \ref{thm:ancilla-free-LCU-theorem} does not always guarantee a high success probability: it only does so, on average. However, this is useful, for instance, if for the underlying problem, when we care about the average success probability. This makes \textit{Ancilla-free LCU} well suited to tackle the spatial search problem, defined as follows for random walks: Consider any reversible Markov chain $P$ such that a certain subset of its nodes (say $M$ out of $N$) are marked. Then, the expected number steps required by a classical random walk to reach the marked nodes is known as the hitting time ($HT$). Analogously, quantum spatial search involves determining the expected number of steps after which the projection of the quantum walk on to the marked subspace is high. Indeed, demonstrating that quantum walks require $O(\sqrt{HT})$ steps on average, for any $P$ and any $M$, had been a long standing open problem. Following a series of works \cite{szegedy2004quantum, childs2004spatial, magniez2007search, krovi2016quantum} that only partially resolved this problem (such as for specific graphs or a single marked element), a generic quadratic speedup (up to a log factor) was finally proven in Ref.~ \cite{ambainis2019quadratic}. This algorithm relies on \textit{Standard LCU}, requiring $O(\log(HT))$ ancilla qubits, in addition to the walk space (edges of $P$). Here we show that the same generic quadratic speedup can be obtained without using any of the $O(\log(HT))$ ancilla qubits by using \textit{Ancilla-free LCU}.
\\~\\
\textbf{Applications:~} We discuss two separate quantum algorithms by discrete-time quantum walks that solve the spatial search problem quadratically faster than classical random walks, while requiring fewer ancilla qubits than the prior art. The first one relies on fast-forwarding discrete-time random walks and formalizes an observation in \cite{apers2019unified}. The second quantum algorithm relies on the fast-forwarding of continuous-time random walks, allowing us to relate both frameworks of classical random walks with discrete-time quantum walks. The running time of both these algorithms scales as the square root of the hitting time of classical random walks (up to log factors) even in the presence of multiple marked elements, which is optimal. Overall, we demonstrate that the \textit{Ancilla-free LCU} provides a unified framework for optimal quantum spatial algorithms. It also allows us to connect discrete and continuous-time random walks with discrete and continuous-time quantum walks. We briefly describe the results we obtain:
~\\
\begin{itemize}
\item[(a)]~\textbf{Quantum spatial search by discrete-time quantum walks:~} We use the fact that if any Hamiltonian $H$ is encoded in the top-left block of a unitary $U_H$, we can obtain a block encoding \cite{low2019hamiltonian, chakraborty2019power} of $H^t$ or $e^{-t(I-H)}$ by implementing a linear combination of Chebyshev polynomials of $H$ \cite{childs2017quantum, ambainis2019quadratic}. Both these procedures require roughly $O(\sqrt{t})$ cost to be implemented. When $H=D$ (the discriminant matrix of $P$), $D^t$ results in a fast-forwarding of discrete-time random walks \cite{apers2018FF}, which is the key subroutine of the optimal spatial search algorithm of \cite{ambainis2019quadratic}. One the other hand, when $H=e^{-t(I-D)}$, a fast-forwarding of continuous-time random walks can be achieved. However, quantum fast forwarding makes use of \textit{Standard LCU}, requiring $O(\log t)$ ancilla qubits. However, for quantum spatial search, we can invoke Theorem \ref{thm:ancilla-free-LCU-theorem} instead, and make use of \textit{Ancilla-free LCU}. This is possible as for this problem we are interested in proving that quantum walks require $\widetilde{O}(\sqrt{HT})$ steps, on average. Thus, our methods do not require any ancilla qubits (other than the walk space).

More precisely, using the framework of interpolated Markov chains (see Sec.~\ref{subsec:rw-qw} for details of these terms), and for a specific initial state $\ket{\sqrt{\pi_U}}$ (related to the stationary distribution of the interpolated random walk), we can show for the first spatial search algorithm, the \textit{Ancilla-free LCU} procedure, on average, prepares a mixed state $\bar{\rho}$ such that
$$
\Tr[(I\otimes\Pi_M)\bar{\rho}]\geq  \norm{\Pi_M D(s)^T\ket{\sqrt{\pi_U}}}^2-\varepsilon,
$$
where $\Pi_M$ is the projector on to the marked subspace. In Ref.~\cite{ambainis2019quadratic}, the authors proved that, on average, the RHS of the aforementioned inequality is $\tilde{\Omega}(1)$, for $T=\widetilde{O}\left(\sqrt{HT}\right)$, for some randomly chosen value of $s\in [0,1)$ (See Algorithm \ref{algo:search-dtqw-1}). Thus our algorithm achieves the same quadratic speedup as \cite{ambainis2019quadratic}, while requiring $O(\log HT)$ fewer ancilla qubits. This also formalizes the observation of Ref.~\cite{apers2019unified}.

For our second spatial search algorithm by discrete-time quantum walk (See Algorithm \ref{algo:search-dtqw-2}), Theorem \ref{thm:ancilla-free-LCU-theorem}, prepares, on average, the mixed state $\bar{\rho}$ such that 
$$
\Tr[(I\otimes \Pi_M)\bar{\rho}]\geq \norm{\Pi_M e^{T(D(s)-I)}\ket{\sqrt{\pi_U}}}^2-\varepsilon,
$$
where again, we prove that the expected value of the RHS is in $\widetilde{\Omega}(1)$, for $T=\widetilde{O}\left(\sqrt{HT}\right)$. For this we prove that the probability of finding a marked vertex is at least as large as the probability of a certain event occurring in the continuous-time interpolated Markov chain $P(s)$, which in turn can be lower bounded by the probability of the same event occurring for the corresponding discrete-time Markov chain. These reductions allow us to leverage the results of \cite{ambainis2019quadratic}. Both these algorithms require no ancilla qubits (other than the walk space), while solving this problem via quantum fast-forwarding by \textit{Standard LCU} would require $O(\log HT)$ ancilla qubits.
\end{itemize}
\begin{figure}[h!]
\centering
\captionsetup{justification=centering}
\includegraphics[scale=0.45]{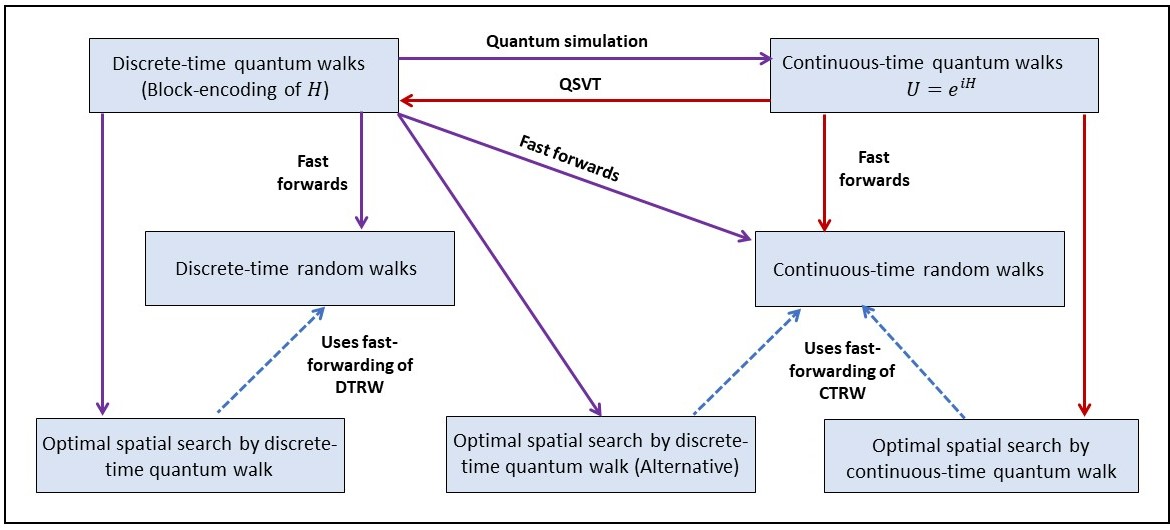}
\caption{\textit{Relationship between discrete and continuous-time quantum walks, and their classical counterparts:} Given the block encoding of the discriminant matrix $D$ of any reversible Markov chain $P$, we can generate a continuous-time quantum walk on both the vertices and edges of $P$, following Ref.~\cite{childs2010relationship}. Conversely, given access to the continuous-time quantum time-evolution operator $U=e^{-iH/2}$, we can implement a discrete-time quantum walk on the edges of $H$. For this, we make use of the framework of Quantum Singular Value Transformation (QSVT). Apers and Sarlette demonstrated that discrete-time random walks can be fast forwarded by discrete-time quantum walks \cite{apers2018FF} which is the cornerstone of the recently developed optimal quantum spatial search algorithm \cite{ambainis2019quadratic}. In this work, we show that for quantum spatial search, the fast-forwarding can be done without needing any ancilla qubits (other than the quantum walk space) through \textit{Ancilla-free LCU}. Additionally, we also demonstrate that discrete-time quantum walks can fast forward continuous-time random walks. This fact, can also be leveraged to develop new optimal quantum spatial search algorithms, which do not require any ancilla qubits (other than the quantum walk space). Finally continuous-time random walks can also be fast-forwarded by continuous-time quantum walks which is central to the optimal quantum spatial search algorithm of \cite{apers2022quadratic}. Thus overall, through this work we connect all different frameworks random and quantum walks.}
\label{fig:walk-rel}
\end{figure}

It is worth mentioning that optimal spatial search algorithm by continuous-time quantum walk \cite{apers2022quadratic} is yet another demonstration of (a continuous-time variant of)\textit{Ancilla-free LCU}, where a quadratic speedup for this problem could be obtained by the fast forwarding of a continuous-time random walk. This algorithm involved using randomized time-evolution through which the full LCU procedure could be bypassed. Additionally, these optimal quantum spatial search algorithms allow us to relate discrete-time random walks and continuous-time random walks with their quantum counterparts through quantum spatial search. In the Appendix (Sec.~\ref{sec:relationship-ctqw-dtqw}), we also establish a relationship between discrete and continuous-time quantum walks. In a seminal work, Childs showed that any dynamics generated by a continuous-time quantum walk can be simulated by a discrete-time quantum walk \cite{childs2010relationship}. We show that given access to a discrete-time quantum walk, one can also generate a continuous-time quantum walk (on the edges of the underlying Markov chain). In the other direction, whether discrete-time quantum walks can be obtained from the continuous-time quantum walk evolution operator, has been unknown. In fact, using the frameworks of block encoding and quantum singular value transformation, we show that given access to a quantum walk time evolution operator, one can obtain a discrete-time quantum walk. We also discuss the subtleties of this approach with regards to quantum fast forwarding and also suggest possible improvements. Overall, this helps us complete the picture (See Fig.~\ref{fig:walk-rel}) by relating both frameworks of quantum and random walks.
\subsection{Prior work}
\label{subsec:prior-work}
In this section, we briefly sketch relevant prior work and relate them to the results we obtain. The linear combination of unitaries technique was first introduced by Childs and Wiebe \cite{childs2012hamiltonian} to develop quantum Hamiltonian simulation algorithms based on multi-product formulas. Since then, LCU has been extensively used to develop improved quantum algorithms for Hamiltonian simulation \cite{berry2014exponential, berry2015simulating, berry2015hamiltonian}. Subsequently, it has been used to develop a wide variety of quantum algorithms for linear algebra, such as for solving quantum linear systems, and linear regression \cite{childs2017quantum, chakraborty2019power}, preparing ground states of Hamiltonians \cite{ge2019faster} and solving optimization problems \cite{chowdhury2017quantum, van2020quantum}. Many of these quantum algorithms have been unified by the more recent framework of quantum singular value transformation (QSVT) \cite{gilyen2019quantum}, which implements polynomial transformations to the singular values of a matrix. Given access to a unitary that block encodes \cite{chakraborty2019power} the matrix, QSVT provides near-optimal query complexities for these problems (in terms of the number of queries made to the block encoding) and requires fewer ancilla qubits than LCU. However, constructing the block encoding itself can be resource consuming and may require a large number of ancilla qubits and multi-qubit controlled gates which increases the overall circuit depth. Thus, it is not likely to be implemented on early fault-tolerant quantum computers.   

The main contribution of this article is to demonstrate that the framework of implementing LCU can be modified so that this framework, which is also applicable to a wide variety of problems, is implementable on intermediate-term quantum computers. As discussed in the previous sections, we introduce three main variants of LCU.

The first technique, \textit{Single-Ancilla LCU}, implements any LCU using only a single ancilla qubit and no multi-qubit controlled operations. Despite this, the cost per coherent run is lower than the generic LCU procedure. Consequently, it is useful for early fault-tolerant quantum computers. A few quantum algorithms, tailored to early fault-tolerant quantum computers have recently been developed. However, these algorithms tackle a single problem, namely, the estimation of the ground state energy of a given Hamiltonian \cite{kyriienko2020quantum,lin2022heisenberg,zhang2022computing,wang2022quantum,dong2022ground}. Our method makes use of the quantum circuit of Faerhmann et al.~\cite{Faehrmann2022randomizingmulti}, wherein it was used for multi-product Hamiltonian simulation. In a way, it is a non-trivial generalization of their technique to implement any LCU. This opens up the possibility of developing several novel quantum algorithms using early fault-tolerant quantum devices. In this work, we apply our method for Hamiltonian simulation, ground state property estimation, and quantum linear systems.

For each of these problems, several quantum algorithms have been developed over the years. Hamiltonian simulation has been widely considered as one of the potential applications of the first useful quantum computer. Algorithms more suited to near-term applicability include Trotter-based approaches \cite{childs2021theory} and their randomized variants such as qDRIFT \cite{campbell2019random, chen2021concentration} and others \cite{childs2018toward, zhao2022hamiltonian}. The standard LCU procedure has also been quite useful for developing near-optimal quantum simulation algorithms \cite{berry2015hamiltonian,berry2015simulating,berry2017quantum}. For our Hamiltonian simulation algorithm, we make use of the LCU decomposition of the time evolution operator \cite{wang2022quantum}, along with the Truncated Taylor series method of Berry et al. \cite{berry2015simulating}. State-of-the-art quantum simulation algorithms use the framework of quantum signal processing \cite{low2017optimal}, and QSVT \cite{gilyen2019quantum}, but require access to a block encoding of the Hamiltonian, which, as discussed previously, is resource consuming. We extensively compare our Hamiltonian simulation procedure with other algorithms (Sec.~\ref{sec:ham-sim}).

The first quantum algorithm for ground state preparation involved using Hamiltonian simulation along with quantum phase estimation \cite{abrams1999quantum}. Subsequently, Refs.~\cite{ge2019faster, keen2021quantum} took advantage of the fact that functions of Hamiltonians can be expressed as a linear combination of unitaries to develop fast quantum algorithms for ground state preparation and estimation. A QSVT-based quantum algorithm has also been developed recently \cite{lin2020near}, which requires an optimal number of queries to the block encoding of the Hamiltonian. The problem of ground state preparation/property estimation is also considered to be one of the first problems to be solved using near/intermediate-term quantum computers. In Sec.~\ref{subsec:gsp-single-ancilla} and Sec.~\ref{subsec-app:gsp-qls}, we substantively compare our procedure with other early fault-tolerant quantum algorithms for ground energy estimation, as well as with state-of-the-art quantum algorithms, suitable for fully fault-tolerant quantum computers.

Ever since the seminal algorithm by Harrow, Hassidim and Lloyd \cite{harrow2009quantum}, quantum linear systems has been analyzed extensively. In particular, the LCU-based approach of Childs, Kothari, and Somma \cite{childs2017quantum} provided a linear dependence on the condition number of the underlying sparse matrix and an exponentially improved dependence on the error. This algorithm was improved to also work for non-sparse matrices \cite{wossnig2018quantum} and in the more general block encoding framework \cite{chakraborty2023quantum}. Recently, QSVT-based approaches for this problem have also been developed \cite{gilyen2019quantum, martyn2021grand, chakraborty2023quantum}. Another direction of research has been to develop quantum algorithms for this problem in the adiabatic quantum computing framework \cite{subacsi2019quantum, lin2020optimal, costa2022optimal}. This approach provides an optimal dependence on all parameters but requires access to a block encoding of the matrix to be inverted. The possibility of applying quantum linear systems algorithm on the near-term quantum devices has been explored in \cite{huang2021near}. However, the techniques discussed are either variational and hence heuristic, or requires substantially higher resources as compared to our method. We rigorously compare our approach with these techniques in Sec.~\ref{subsec:single-ancilla-qls} and the Appendix (Sec.~\ref{subsec-app:gsp-qls}).

Our second technique is a continuous-time variant of LCU, which is more physical. The key idea is to couple discrete systems with continuous-variable systems. Such interactions have been explored in the context of quantum phase estimation, where the system Hamiltonian is coupled with a one-dimensional free particle, acting as the pointer variable - which is the so-called von Neumann measurement model \cite{childs2002quantum, chakraborty2020analog}. In Ref.~\cite{apers2022quadratic}, the continuous-time quantum walk Hamiltonian $H$ was coupled to a one-dimensional quantum Harmonic oscillator to implement $e^{-tH^2}$. This was a key ingredient of their spatial search algorithm by continuous-time quantum walk. In this work, we formalize this technique and show that it is more widely applicable and in fact, can serve as a continuous-time variant to any LCU-based quantum algorithm. We develop an analog variant of the quantum linear systems algorithm of Childs et al. \cite{childs2017quantum} and a new quantum algorithm for this problem (using only Gaussian states) that is more suited for intermediate-term implementation.

Finally, we apply the \textit{Ancilla-free LCU} technique to develop optimal quantum spatial search algorithms. LCU has been central to the development of several quantum walk algorithms. The quantum fast-forwarding scheme by Apers and Sarlette \cite{apers2018FF} quadratically fast-forwards the dynamics of a discrete-time random walk by implementing a linear combination of discrete-time quantum walk steps. Recently, Ambainis et al. \cite{ambainis2019quadratic} proved that for the spatial search problem, fast-forwarding $T$ random walk steps on an interpolated Markov chain, prepares a quantum state that has, on average, a $\tilde{\Omega}(1)$ overlap with the marked space for $T=\widetilde{O}(\sqrt{HT})$, where $HT$ is the classical hitting time of the random walk. Overall, this algorithm requires $O(\log n + \log HT)$ ancilla qubits, for any reversible Markov chain of $n$ nodes. Their LCU-based quantum spatial search algorithm for discrete-time quantum walks completely solves the spatial search problem quadratically faster than classical random walks, for any number of marked nodes. This closed a long line of work which made partial progress towards solving this problem (such as for only particular reversible Markov chains or when only a single node was marked) \cite{szegedy2004quantum, magniez2007search, krovi2016quantum}. Subsequently, Apers et al. \cite{apers2019unified} provided a unified framework that connected the different variants of discrete-time quantum walk search. Therein, the authors observed that the LCU procedure of \cite{ambainis2019quadratic} could be replaced by randomly sampling quantum walk steps. In the continuous-time quantum walk framework, whether the spatial search problem offered a generic quadratic speedup was also open for a long time, with a series of works that partially resolved this problem \cite{childs2004spatial, chakraborty2016spatial, chakraborty2020optimality, chakraborty2020finding}. It was only recently solved in \cite{apers2022quadratic}. Their analog quantum algorithm indeed managed to bypass the LCU procedure by evolving the system under the quantum walk Hamiltonian for a random time. In this work, we demonstrate that by using \textit{Ancilla-free LCU} instead of using the full LCU procedure (thereby requiring fewer ancilla qubits), one can develop two spatial search algorithms by discrete-time quantum walks: one that relies on fast-forwarding discrete-time random walks and the other, relying on fast-forwarding continuous-time random walks. Similar to \cite{ambainis2019quadratic}, these quantum algorithms provide a generic quadratic speedup over their classical counterparts, while requiring $O(\log HT)$ fewer ancilla qubits.

In the following section, we briefly define some of the key preliminary concepts that we will borrow to develop our results.   

\section{Preliminaries}
\label{sec:preliminaries}
In this section, we introduce some of the techniques that we use in this article as well as discuss the key algorithmic primitives required to develop our results. We begin by stating the complexity theoretic notations we shall be using throughout the article.
\subsection{Notation}
\label{subsec:notation}

\textbf{Complexity theoretic notations:~} Throughout the article, we shall be using the standard complexity theoretic notations. The \textit{Big-O} notation, $g(n)= O(f(n))$ or $g(n)\in O(f(n))$, implies that $g$ is upper bounded by $f$. That is, there exists a a constant $k>0$ such that $g(n)\leq k\cdot f(n)$. The \textit{Big-Omega} notation, $g(n)=\Omega(f(n))$, implies $g(n)\geq k f(n)$ ($g$ is lower bounded by $f$). The \textit{Theta} notation is used when $g$ is both bounded from above and below by $f$, i.e.\ $g(n)=\Theta(f(n))$, if there exists non-negative constants $k_1$, and $k_2$ such that $k_1 f(n) \leq g(n) \leq k_2 f(n)$. The \textit{Little-o} notation, $g(n)=o(f(n))$, when $g$ is dominated by $f$ asymptotically, i.e.\ $\lim_{n\rightarrow\infty} g(n)/f(n)=0$. 

For each of these notations, it is standard to use \textit{tilde} ($\sim$) to hide polylogarithmic factors. For instance, $\widetilde{O}(f(n))=O(f(n)\mathrm{polylog}(f(n)))$. This applies to the other notations as well.
~\\~\\
\textbf{Trace, Expectation and Probability:~} The trace of an operator $A$ will be denoted by $\Tr[A]$, while the expectation value of the operator will be denote by $\mathbb{E}[A]$. The probability of an event $X$ occurring will be denoted by $\Pr[X]$. 
~\\~\\
\textbf{Schatten norms:~} The Schatten p-norm of the operator $X$ is defined as 
$$
\norm{X}_p =\left(\sum_{j} \sigma^{p}_j(X)\right)^{1/p},
$$
where $\sigma_j(X)$ is the $j^{\mathrm{th}}$ singular value of $X$. So if $\sigma_{\max}(X)$ denotes the maximum singular value of $X$, we have 
$$
\lim_{p\rightarrow\infty} \norm{X}_p = \sigma_{\max} \cdot \lim_{p\rightarrow\infty} \left(\sum_{j} \dfrac{\sigma^{p}_j(X)}{\sigma^{p}_{\max}(X)}\right)^{1/p}=\sigma_{\max},
$$
which is the spectral norm of the operator $X$, which we denote as simply $\norm{X}$.
~\\~\\ 
\textbf{LCU coefficients:~} For LCU, we implement some operator $V=\sum_{j}c_j U_j$, where each $U_j$ is a unitary. Note that $c_j$ can, in general, be any non-zero real or complex number (positive or negative). When $c_j$ is complex or negative, we can absorb the sign as well as the imaginary phase into the description of the unitary itself. Consequently, without loss of generality, it suffices to consider that $c_j\in \mathbb{R}^{+} \backslash \{0\}$. This is what we shall be considering throughout the article. In the next section, we describe the formalism of LCU.

\subsection{Linear Combination of Unitaries}
\label{subsec:prelim-lcu}
We will begin by stating the general framework of Linear Combination of Unitaries (LCU). Throughout the article, we shall refer to this as the \textit{Standard LCU} procedure. Suppose for $H\in \mathbb{C}^{N\times N}$, we wish to apply some function of the Hamiltonian to an initial state $\ket{\psi_0}$. More precisely, if $H$ has spectral decomposition $H=\sum_{j=1}^{N}\lambda_j\ket{v_j}\bra{v_j}$, define $f(H)=\sum_{j=1}^{N}f(\lambda_j)\ket{v_j}\bra{v_j}$. Now suppose $f(H)$ can be well approximated by linear combinations of unitaries. That is, for some $\gamma\in (0,1)$ suppose,
$$
\norm{f(H)-\sum_{j=1}^M c_j U_j}\leq \gamma,
$$
where each $U_j$ is a unitary matrix that we have access to, i.e $f(H)$. Without loss of generality let us define the parameters $c_j\in \mathbb{R}^{+} / \{0\}$. Even though $f(H)$ is not necessarily unitary, the LCU technique allows us to implement $f(H)\ket{\psi_0}$.

For this, the procedure requires $m=\lceil \log_2 M \rceil$ ancilla qubits. First, a \textit{prepare} unitary $R$ is applied to this ancilla register such that 
$$
R\ket{\bar{0}}=\sum_{j=1}^{M}\sqrt{\dfrac{c_j}{\norm{c}_1}}\ket{j},
$$
where $c=(c_1,\dots, c_M)^T$. Suppose the cost of implementing this unitary is $\tau_R$. 

Furthermore, a \textit{select} unitary $Q$, defined in the following way,
$$
Q=\sum_{j}\ket{j}\bra{j}\otimes U_j,
$$
is also used. Note that $Q$ is a sophisticated operation, controlled on each of the $m$ ancilla qubits. Suppose the cost of implementing $Q$ is $\tau_Q$.

Then,
\begin{equation}
\ket{\psi_t}=(R^{\dag}\otimes I) Q (R\otimes I)\ket{\bar{0}}\ket{\psi_0}=\dfrac{1}{\norm{c}_1}\ket{\bar{0}}\sum_{j=1}^{M} c_jU_j\ket{\psi_0}+\ket{\Phi}^{\perp}
\end{equation}
where $(\ket{\bar{0}}\bra{\bar{0}}\otimes I)\ket{\Phi}^{\perp}=0$. Note that, controlled on $\ket{\bar{0}}$, the state in the second register is $\gamma/\norm{c}_1$-close to
$$
\ket{\psi}=\dfrac{f(H)\ket{\psi_0}}{\norm{f(H)\ket{\psi_0}}},
$$
with probability $p=\norm{f(H)\ket{\psi_0}}^2/\norm{c}^2_1$. The total cost of using the LCU procedure is then
\begin{equation}
\label{eq:standard-lcu-complexity}
\Gamma_{\max}= 2\tau_R+\tau_Q + \tau_{\psi_0}.
\end{equation}
Let us now explore the applicability of this procedure. 
For instance, $f(H)$ could be well approximated by a Fourier series, in which case, $U_j=e^{-ijH}$. Since, for several applications, it indeed boils down to applying some such $f(x)$ to an initial state, LCU provides a versatile framework that has wide applicability. Consequently, several near-optimal quantum algorithms have been designed in this framework ranging from quantum algorithms for linear systems \cite{childs2017quantum}, ground state preparation \cite{ge2019faster, keen2021quantum}, sampling from thermal states \cite{chowdhury2017quantum, chowdhury2021DQC1} to Hamiltonian simulation \cite{childs2012hamiltonian, berry2014exponential, berry2015hamiltonian, berry2015simulating} and many others. 

However, the twin requirements of several ancilla qubits, as well as sophisticated multi-qubit controlled operations, imply that this standard technique to implement any LCU is useful for only full-scale fault-tolerant quantum computers. Additionally, in most applications, simply preparing the quantum state $\ket{\psi}$ is not useful, as access to the entries of the underlying quantum state are required. However, even state-of-the-art techniques in quantum state tomography result in an exponential overhead. Thus, in most cases, we are interested in predicting certain properties of the state $\ket{\psi}$, such as estimating the expectation values of observables of interest, i.e.\ $\braket{\psi|O|\psi}$. 
\subsection{Block encoding and Quantum Singular Value Transformation}
\label{subsec:prelim-block-encoding-qsvt}
For some of the algorithms in this article, we will consider the framework of \textit{block encoding}, wherein it is assumed that the input matrix $H$ (up to some sub-normalization) is stored in the left block of some unitary. The advantage of the block encoding framework, which was introduced in a series of works \cite{low2019hamiltonian, chakraborty2019power, gilyen2018quantum}, is that it can be applied to a wide variety of input models.
\begin{definition}[Block Encoding~\cite{chakraborty2019power}]
    \label{def:block_encoding}
Suppose that $H$ is an $s$-qubit operator, $\alpha, \varepsilon \in \mathbb{R}^+$ and $a \in \mathbb{N} $, then we say that the $(s + a)$-qubit unitary $U_H$ is an $(\alpha, a, \varepsilon)$-block encoding of $H$, if

\begin{equation}
 \norm{ H - \alpha (\bra{0}^{\otimes a} \otimes I)U_H(\ket{0}^{\otimes a} \otimes I) } \leq \varepsilon .
\end{equation}

\end{definition}
Let $\ket{\psi}$ be an $s$-qubit quantum state. Then applying $U_H$ to $\ket{\psi}\ket{0}^{\otimes a}$ outputs a quantum state that is $\frac{\varepsilon}{\alpha}$-close to 
$$
\ket{0}^{\otimes a}\dfrac{H}{\alpha}\ket{\psi}+\ket{\Phi^{\perp}},
$$
where $\left(\ket{0}^{\otimes a}\bra{0}^{\otimes a}\otimes I_s \right)\ket{\Phi^{\perp}}=0$. Equivalently, suppose $\tilde{H} := \alpha\paren{\bra{0}^{\otimes a} \otimes I_s} U_H \paren{\ket{0}^{\otimes a} \otimes I_s}$ denotes the actual matrix that is block-encoded into $U_H$, then $\norm{H-\tilde{H}}\leq\varepsilon$. 

Quantum Singular Value Transformation (QSVT) applies a polynomial transformation to the singular values of a block-encoded matrix \cite{gilyen2019quantum}. Formally, let $P \in \mathbb{R}[x]$ be a polynomial with real coefficients of degree $n \ge 2$, such that $P$ is either even or odd (has parity-$n\mod 2$), and for all $x\in [-1,1]$, $|P(x)|\leq 1$. Then, QSVT allows us to implement any polynomial $P(x)$ that satisfies the aforementioned requirements. Next, we formally introduce QSVT formally via the following theorem.

\begin{theorem}[Quantum Singular Value Transformation \cite{gilyen2018quantum}]
\label{thm:qsvt-def}
Suppose $A\in \mathbb{R}^{N\times d}$ is a matrix with singular value decomposition $A=\sum_{j=1}^{d_{\min}}\sigma_j\ket{v_j}\bra{w_j}$, where $d_{\min}=\min\{N,d\}$ and $\ket{v_j}$ $(\ket{w_j})$ is the left (right) singular vector with singular value $\sigma_j$. Furthermore, let $U_A$ be a unitary such that $A=\widetilde{\Pi}U_A\Pi$, where $\Pi$ and $\widetilde{\Pi}$ are orthogonal projectors. Then, if any real polynomial $P(x)$ of degree $n$ is even or odd, and satisfies $|P(x)|\leq 1$, for all $x\in[-1,1]$, there exists a vector $\Phi=\left(\phi_1,\phi_2,\cdots\phi_n\right)\in\mathbb{R}^n$ and a unitary
\begin{equation}
    \label{eqn:qsvt_sequence}
    U_{\Phi} = \begin{cases}
        e^{i \phi_1 ( 2 \widetilde{\Pi} - I )} U_A \left[ \prod_{k = 1}^{(n - 1) / 2} e^{i \phi_{2k} ( 2 \widetilde{\Pi} - I )} U_A^{\dagger} e^{i \phi_{2k+1} ( 2 \widetilde{\Pi} - I )} U_A \right], & n \text{ is odd } \\
        \left[ \prod_{k = 1}^{n/ 2} e^{i \phi_{2k-1} ( 2 \widetilde{\Pi} - I )} U_A^{\dagger} e^{i \phi_{2k} ( 2 \widetilde{\Pi} - I )} U_A \right], & n \text{ is even},
    \end{cases}
\end{equation}
such that
\begin{equation}
    \label{eqn:qsvt_sequence_projected}
    P^{SV} (A) = \begin{cases}
        \left(\bra{+}\otimes \widetilde{\Pi} \right) \left(\ket{0}\bra{0}\otimes U_{\Phi} +\ket{1}\bra{1}\otimes U_{-\Phi}\right)\left(\ket{+}\otimes \Pi\right), & n \text{ is odd} \\
        \left(\bra{+}\otimes \Pi \right) \left(\ket{0}\bra{0}\otimes U_{\Phi} +\ket{1}\bra{1}\otimes U_{-\Phi}\right)\left(\ket{+}\otimes \Pi\right), & n \text{ is even},
    \end{cases}
\end{equation}
where $ P^{SV} (A) $ is the polynomial transformation of the matrix $ A $ defined as
\begin{equation}
    \label{eqn:poly_sv_transformation_of_matrix}
    P^{SV} (A) := \begin{cases}
        \sum_j P(\sigma_j) \ket{v_j} \bra{w_j}, & P \text{ is odd} \\
        \sum_j P(\sigma_j) \ket{w_j} \bra{w_j}, & P \text{ is even}.
    \end{cases}
\end{equation}
\end{theorem}

Theorem \ref{thm:qsvt-def} tells us that for any real, bounded $n$-degree polynomial $P$ with definite parity, we can implement $P^{SV}(A)$ using one ancilla qubit, $\Theta(n)$ applications of $U_A$, $U^{\dag}_A$ and controlled reflections $I-2\Pi$ and $I-2\widetilde{\Pi}$. Furthermore, if in some well-defined interval, some function $f(x)$ is well approximated by an $n$-degree polynomial $P(x)$, then Theorem \ref{thm:qsvt-def} also allows us to implement a transformation that approximates $f(A)$, where

\begin{equation}
    \label{eqn:sv_transformation_of_matrix}
     f(A) := \begin{cases}
        \sum_j f(\sigma_j) \ket{v_j} \bra{w_j}, & P\text{ is odd} \\
        \sum_j f(\sigma_j) \ket{w_j} \bra{w_j}, & P \text{ is even}.
    \end{cases}
\end{equation}

The following theorem from Ref.~\cite{gilyen2019quantum} deals with the robustness of the QSVT procedure, i.e. how errors propagate in QSVT. In particular, for two matrices $A$ and $\tilde{A}$, it shows how close their polynomial transformations ($P^{SV}(A)$ and $P^{SV}(\widetilde{A})$, respectively) are, as a function of the distance between $A$ and $\tilde{A}$.
~\\~\\
\begin{lemma}[Robustness of Quantum Singular Value Transformation, \cite{gilyen2019quantum}]
    \label{lem:robustness_of_QSVT}
    Let $P \in \mathbb{R}[x]$ be a $n$-degree real polynomial that satisfies the requirements of QSVT. Let $A, \tilde{A} \in \mathbb{C}^{N \times M}$ be matrices of spectral norm at most 1. Then,
    \begin{equation*}
        \norm{P^{SV}(A) - P^{SV}(\tilde{A})} \leq 4n \sqrt{\norm{A - \tilde{A}}}. 
    \end{equation*}
\end{lemma}

Having discussed the preliminary concepts, in the next section, we explain the three variants of LCU we consider in this article.
\section{Three different approaches for implementing LCU on intermediate-term quantum computers}
\label{sec:new-approaches-LCU}
In this section, we present the three different LCU techniques.  Before stating these techniques, we prove a general lemma that will be invoked while proving results related to our approaches.

\subsection{Robustness of expectation values of observables}
In this section, we develop general results on the robustness of expectation values of observables which we shall use for both the \textit{Single-Ancilla LCU} (Sec.~\ref{subsec:one-ancilla-LCU}) and the \textit{Ancilla-free LCU} (Sec.~\ref{subsec:ancilla-free-LCU})  approaches. 

Consider that there exist two operators $P$ and $Q$ such that $\norm{P-Q}\leq \gamma$. In this section, we demonstrate that the expectation value of $O$ with respect to $P\rho P^{\dag}$ is not far off from the expectation value of $O$ with respect to $Q\rho Q^{\dag}$, for any density matrix $\rho$. More precisely, we prove
$$
\left|\Tr[O~P\rho P^{\dag}]-\Tr[O~Q\rho Q^{\dag}]\right|\leq 3\norm{P}\norm{O}\gamma.
$$
In order to prove this result, we need to use the tracial version H\"{o}lder's inequality which is stated below for completeness:
~\\
\begin{lemma}[Tracial version of H\"{o}lder's inequality \cite{ruskai1972inequalities}]
\label{thm:holder}
Define two operators $A$ and $B$ and parameters $p,q\in [1,\infty]$ such that $1/p+1/q =1 $. Then the following holds:
$$
\Tr[A^{\dag}B]\leq \norm{A}_p \norm{B}_q.
$$
\end{lemma}
Here $\norm{X}_p$ corresponds to the Schatten p-norm of the operator $X$. For the special case of $p=\infty$ and $q=1$, the statement of Lemma \ref{thm:holder} can be rewritten as
\begin{equation}
\label{eq:holder-special-case}
\Tr[A^{\dag}B]\leq \norm{A}_{\infty} \norm{B}_1=\norm{A} \norm{B}_1.
\end{equation}
Now we are in a position to formally state the main result.
~\\
\begin{theorem}
\label{thm:distance-expectation}
Suppose $P$ and $Q$ are operators such that $\norm{P-Q}\leq \gamma$ for some $\gamma\in [0,1]$. Furthermore, let $\rho$ be any density matrix and $O$ be some Hermitian operator with spectral norm $\norm{O}$. Then, if $\norm{P}\geq 1$, the following holds:
$$
\left|\Tr[OP\rho P^{\dag}]-\Tr[OQ\rho Q^{\dag}]\right| \leq 3\norm{O}\norm{P}\gamma.
$$
\end{theorem}
\begin{proof}
For the operators $P$ and $Q$, we have
\begin{align}
(Q-P)\rho(P^{\dag}-Q^{\dag})=Q\rho (P^{\dag}-Q^{\dag})-P\rho(P^{\dag}-Q^{\dag})
\end{align}
Now adding and subtracting $P\rho P^{\dag}$ in the RHS we obtain
\begin{align}
(Q-P)\rho(P^{\dag}-Q^{\dag})+P\rho(P^{\dag}-Q^{\dag})&= Q\rho (P^{\dag}-Q^{\dag})+P\rho P^{\dag}-P\rho P^{\dag}\\
					&= P\rho P^{\dag} - Q\rho Q^{\dag} - (P-Q)\rho P^{\dag}
\end{align}
This gives us
\begin{equation}
\label{eq:diff-operators-applied-to-states}
P\rho P^{\dag} - Q\rho Q^{\dag} = (Q-P)\rho(P^{\dag}-Q^{\dag})+P\rho(P^{\dag}-Q^{\dag})+(P-Q)\rho P^{\dag}.
\end{equation}
Now, using the tracial version of Holder's inequality with $p=\infty$ and $q=1$, we have
\begin{equation}
\label{eq:robustness-exp-inequality}
\left|\Tr[O~P\rho P^{\dag}]-\Tr[O~Q\rho Q^{\dag}]\right|\leq \norm{O}\cdot \norm{P\rho P^{\dag}-Q\rho Q^{\dag}}_1.
\end{equation}
For the second term in the RHS of the above equation, we can use the expression in Eq.~\eqref{eq:diff-operators-applied-to-states}. That is,
\begin{align}
\norm{P\rho P^{\dag}-Q\rho Q^{\dag}}_1 &=\norm{(Q-P)\rho(P^{\dag}-Q^{\dag})+P\rho(P^{\dag}-Q^{\dag})+(P-Q)\rho P^{\dag}}_1.
\end{align}
At this stage we can successively apply the tracial version of Holder's inequality (Lemma \ref{thm:holder} with $p=\infty$ and $q=1$) to the RHS of the expression above to obtain
\begin{align}
\norm{P\rho P^{\dag}-Q\rho Q^{\dag}}_1 &\leq \norm{Q-P}\norm{\rho(P^{\dag}-Q^{\dag})}_1+\norm{P}\norm{\rho(P^{\dag}-Q^{\dag})}_1+\norm{P-Q}\norm{\rho P^{\dag}}_1\\
                                       & \leq \norm{P-Q}^2 \norm{\rho}_1+ \norm{P} \norm{P-Q} \norm{\rho}_1+\norm{P-Q} \norm{P} \norm{\rho}_1\\
                                       &\leq \norm{P-Q}^2+2\norm{P-Q}\norm{P}~~~~~~~~~~~~[\mathrm{~Using~} \norm{\rho}_1=1~]
\end{align}
Now, substituting this back in the RHS of Eq.~\eqref{eq:robustness-exp-inequality}, we obtain
\begin{align}
\left|\Tr[O P\rho P^{\dag}] - \Tr[O Q\rho Q^{\dag}]\right| &\leq \norm{O} \norm{P-Q}^2+2\norm{O}\norm{P}\norm{P-Q}\\
								&\leq \gamma^2\norm{O}+2\norm{O}\norm{P}\gamma~~~~~~~~~~~~~~~~~~~~~~~~~~~~[~\mathrm{As~} \norm{P-Q}\leq \gamma]\\
								&\leq 3\gamma\norm{O}\norm{P}~~~~~~~~~~~~~~~~~~~~~~~~~~~~~~~~~~~~~~~~[~\mathrm{As~} \norm{P}\geq 1]
\end{align}
\end{proof}
It is easy to see why Theorem \ref{thm:distance-expectation} is useful to develop robust versions of \textit{Single Ancilla LCU} and \textit{Ancilla-free LCU}.  Typically, $f(H)$ is often not exactly equal to a linear combination of unitaries but is $\gamma$-close to it. Formally, $\norm{f(H)-g(H)}\leq \gamma$, where $g(H)$ can be exactly expressed as an LCU. Consequently, for the variants of LCU that we develop in the subsequent sections, we will be estimating $\Tr[O~g(H)\rho g(H)^{\dag}]$. But, by Theorem \ref{thm:distance-expectation}, we can always bound
$$
\left|\Tr[O~f(H)\rho f(H)^{\dag}]-\Tr[O~g(H)\rho g(H)^{\dag}]\right|\leq 3\norm{O}\norm{f(H)}\gamma.
$$ 
We shall be using this result in the subsequent sections.
\subsection{Single-Ancilla LCU: Estimating expectation values of observables}
\label{subsec:one-ancilla-LCU}
In this section, we describe the \textit{Single-Ancilla LCU} technique, which allows us to sample from quantum states obtained by applying LCU. Suppose we are given a Hermitian matrix $H\in \mathbb{C}^{N\times N}$ and we wish to implement $f(H)$, which can be approximated by a linear combination of unitaries. That is, for some $\gamma\in [0,1)$,
$$
\norm{f(H)-\sum_{j=1}^{M} c_j U_j}\leq \gamma,
$$
for unitaries $U_j$ and $c_j\in\mathbb{R}^{+}\backslash \{0\}$. Let us define the $\ell_1$-norm of the LCU coefficients as $\norm{c}_1=\sum_{j=1}^{M} c_j$. Furthermore, suppose we have access to the quantum circuit for $U_j$. Then given any initial state $\rho_0$, we can estimate the expectation value 
$$
\Tr[O\rho]=\dfrac{\Tr[O~f(H)\rho_0 f(H)^{\dag}]}{\Tr[f(H)\rho_0 f(H)^{\dag}]},
$$
to arbitrary accuracy, for any observable $O$, using the quantum circuit in Fig.~\ref{fig:single-ancilla-circuit}. We estimate both the expectation value (numerator) and the norm (denominator), by making use of Algorithm \ref{algo:randomized-time-evolution},whose steps are stated thusly: the ancilla qubit is prepared in the state $\ket{+}$ so that the overall initial state is $\rho_1=\ket{+}\bra{+}\otimes \rho_0$. We sample two unitaries $V_1$ and $V_2$, independently from the ensemble $\mathcal{D}=\{U_j, c_j/\norm{c}_1\}$ and then implement controlled and anti-controlled versions of $V_1$ and $V_2$, respectively. Finally, we make a measurement of the observable $X\otimes O$ and store the outcome.
\RestyleAlgo{boxruled}
\begin{algorithm}[ht]
  \caption{$\texttt{Expectation-observable~}(O, \{c_j,U_j\},\rho_0, T$)}\label{algo:randomized-time-evolution}
  ~\\~\\
   \begin{itemize}
  \item[1.~] Prepare the state $\rho_1=\ket{+}\bra{+}\otimes \rho_0$.
  \item[2.~] Obtain i.i.d.~samples $V_1, V_2$ from the distribution $\left\{U_j, \dfrac{c_j}{\norm{c}_1}\right\}$.
  \item[3.~] For $\tilde{V}_1=\ket{0}\bra{0}\otimes I+\ket{1}\bra{1}\otimes V_1$ and $\tilde{V}_2=\ket{0}\bra{0}\otimes V_2+\ket{1}\bra{1}\otimes I$, measure \\ $(X\otimes O)$ on the state
  $$
  \rho'= \tilde{V}_2 \tilde{V}_1\rho_1 \tilde{V}^{\dag}_1\tilde{V}^{\dag}_2.
  $$
  \item[4.~] For the $j^{\mathrm{th}}$ iteration, store into $\mu_j$, the outcome of the measurement in Step 4.
  \item[5.~] Repeat Steps 1 to 4, a total of $T$ times.
  \item[6.~] Output 
  $$
  \mu=\dfrac{\norm{c}^2_1}{T}\sum_{j=1}^{T} \mu_j.
  $$
\end{itemize}
\end{algorithm}
We then sample from this quantum circuit $T$ times and calculate the mean of all the outcomes. This is the final output of Algorithm \ref{algo:randomized-time-evolution}. 

If $\tau_j$ is the cost of implementing the unitary $U_j$, then the cost of implementing $V_1$ (or $V_2$) depends on this quantity. Indeed, the average cost of implementing these unitaries (in each coherent run of Algorithm \ref{algo:randomized-time-evolution}) is given by $2\langle\tau\rangle$, where
\begin{equation}
\label{eq:avg-cost-per-run}
\langle \tau \rangle =\dfrac{1}{\norm{c}_1}\sum_{j=1}^{M} c_j\tau_j.
\end{equation} 
Additionally, if $\tau_{\rho_0}$ is the cost of preparing the initial state $\rho_0$, we can we can upper bound the average cost of each coherent run of Algorithm \ref{algo:randomized-time-evolution} is $2\langle \tau\rangle +\tau_{\rho_0}$. In the worst case, if $\tau_{\max}=\max_j \tau_j$, we have $\langle\tau\rangle\leq \tau_{\max}$. Thus, a pessimistic upper bound on the cost per coherent run would be $2\tau_{\max}+\tau_{\rho_0}$. Next, we calculate the number of classical iterations $T$ required to estimate the desired quantity. 

In order to estimate $\Tr[O\rho]$, we need to run this Algorithm twice: first to estimate the expectation value, and then to estimate the norm. First, we formally show via Theorem \ref{thm:randomized-time-evolution} that if the number of samples obtained $T$ is large enough, the sample mean of the outcomes converges to the expectation value $\Tr[O~f(H)\rho_0 f(H)^{\dag}]$.
 \\~\\
\begin{restatable}[Estimating expectation values of observables]{theorem}{thmbodyrandomizedtimeevolve}
\label{thm:randomized-time-evolution}
Let $\varepsilon, \delta, \gamma \in (0,1)$ be some parameters. Let $O$ be some observable and $\rho_0$ be some initial state. Suppose there is a Hermitian matrix $H\in\mathbb{C}^{N\times N}$, such that $\norm{f(H)-\sum_{j}c_j U_j}\leq \gamma$, where $U_j$ is some unitary such that 
$$
\gamma \leq \dfrac{\varepsilon}{6\norm{O}\norm{f(H)}}.
$$ 
Furthermore, let
$$
T\geq \dfrac{8\norm{O}^2\ln(2/\delta)\norm{c}_1^4}{\varepsilon^2}. 
$$ 
Then, Algorithm \ref{algo:randomized-time-evolution} estimates $\mu$ such that
$$
\left|\mu - \Tr[O~f(H)\rho_0f(H)^{\dag}]\right| \leq \varepsilon,
$$
with probability at least $1-\delta$, using of one ancilla qubit and $T$ repetitions of the quantum circuit in Fig.~\ref{fig:single-ancilla-circuit}. 
\end{restatable}
~\\~\\
\begin{proof}
Let $g(H)=\sum_{j}c_j U_j$. First observe from Algorithm \ref{algo:randomized-time-evolution} that the initial state $\rho_1=\ket{+}\bra{+}\otimes\rho_0$ transforms to
\begin{align}
\rho'&= \tilde{V}_2 \tilde{V}_1 \rho_1 \tilde{V}^{\dag}_1\tilde{V}^{\dag}_2\\
	 &=\dfrac{1}{2}\left[\ket{0}\bra{0}\otimes V_2 \rho_0 V^{\dag}_2 +\ket{0}\bra{1}\otimes V_2 \rho_0 V^{\dag}_1+\ket{1}\bra{0}\otimes V_1 \rho_0 V^{\dag}_2+\ket{1}\bra{1}\otimes V_1 \rho_0 V^{\dag}_1\right].	
\end{align}
So after measuring the observable $X\otimes O$, we have
$$
\Tr\left[(X\otimes O)\rho'\right]=\dfrac{1}{2}\Tr\left[O\left(V_1 \rho_0 V^{\dag}_2+ V_2 \rho_0 V^{\dag}_1\right)\right].
$$
Note that the expected values, 
$$
\mathbb{E}\left[V_1\right]=\mathbb{E}\left[V_2\right]=\dfrac{1}{\norm{c}_1}\sum_{j}c_j U_j.
$$
So, the expected outcome of the $j^{\mathrm{th}}$ iteration is
$$
\mathbb{E}\left[\mu_j\right]=\mathbb{E}\left[\Tr[(X\otimes O)\rho']\right]=\dfrac{1}{\norm{c}_1^2}\Tr[O~g(H)\rho_0 g(H)^{\dag}].
$$
Next, we need to estimate two things: 
\begin{itemize}
\item[(a)]~How fast does the sample mean $\mu=\sum_{j}\norm{c}^2_1\mu_j/T$ converge to its expectation value? For this, we use Hoeffding's inequality.  
\item[(b)]~What is the accuracy of the observation with respect to $f(H)$ as a function of the distance between $f(H)$ and $g(H)$? For this, we invoke Theorem \ref{thm:distance-expectation}. 
\end{itemize}
~\\
Observe that the POVM measurement yields some outcome of $O$ in the range $[-\norm{O},\norm{O}]$. So each random variable lies in the range 
$$
-{\norm{O}\norm{c}^2_1 \leq \norm{c}^2_1\mu_j\leq +\norm{O}\norm{c}^2_1}.
$$ 

We evaluate (a), by using Hoeffding's inequality to obtain
\begin{align}
\Pr\left[\left|\mu - \Tr[O~g(H)\rho_0~g(H)^{\dag}]\right|\geq \varepsilon/2\right] \leq 2\exp\left[-\dfrac{T\varepsilon^2}{8\norm{c}^4_1\norm{O}^2}\right].
\end{align}
This immediately gives us that for
\begin{equation}
\label{eq:number-of-repititions}
T\geq \dfrac{8\norm{O}^2\ln(2/\delta)\norm{c}_1^4}{\varepsilon^2},
\end{equation}
$$
\left|\mu -\Tr[O~g(H)~\rho_0~g(H)^{\dag}]\right| \leq \varepsilon/2,
$$
with probability at least $1-\delta$. Now, in order to evaluate (b), we first apply triangle inequality, we obtain
\begin{align}
\left|\mu - \Tr[O~f(H)\rho_0 f(H)^{\dag}]\right| \leq &\left|\mu - \Tr[O~g(H)\rho_0 g(H)^{\dag}]\right|+ \\
&\left|\Tr[O~f(H)\rho_0 f(H)^{\dag}]-
\Tr[O~g(H)\rho_0 g(H)^{\dag}]\right|.
\end{align}
The first term in the RHS of the above inequality is upper bounded by $\varepsilon/2$. In order to bound the second term, note that $\norm{f(H)-g(H)}\leq \gamma$. For any such operators that are at most $\gamma$-separated, we can use Theorem \ref{thm:distance-expectation} to obtain:
$$
\left|\Tr[O~f(H)\rho_0 f(H)^{\dag}]-
\Tr[O~g(H)\rho_0 g(H)^{\dag}]\right|\leq 3\norm{O}\norm{f(H)}\gamma\leq \varepsilon/2,
$$
for $\gamma$ upper bounded as in the statement of Theorem \ref{thm:randomized-time-evolution}.
So, overall we have
$$
\left|\mu - \Tr[O~f(H)\rho_0 f(H)^{\dag}]\right| \leq \varepsilon,
$$
which completes the proof.
\end{proof}
~\\~\\
When $f(H)$ is unitary, the operator $\rho=f(H)\rho_0 f(H)^{\dag}$ is a normalized quantum state and hence, $\Tr[O\rho]=\Tr[O~f(H)\rho_0 f(H)^{\dag}]$. However, this is not the case in general. In such scenarios, as described previously, 
$$
\Tr[O\rho]=\dfrac{\Tr[O~f(H)\rho_0 f(H)^{\dag}]}{\Tr[f(H)\rho_0 f(H)^{\dag}]}.
$$ 
Thus far, we could only obtain the numerator of the RHS of this equation. Next, we describe the procedure to obtain an estimate of the norm $\ell^2=\Tr[f(H)\rho_0 f(H)^{\dag}]$. 

Note that we do not have an accurate knowledge of this quantity apriori. Provided we have knowledge of some rudimentary lower bound  of the norm, using Algorithm \ref{algo:randomized-time-evolution} and Theorem \ref{thm:randomized-time-evolution}, we can obtain an estimate $\tilde{\ell}$ of $\ell^2$ to arbitrary accuracy by simply setting $O=I$. The crucial question in this regard is, how accurate should the estimate $\tilde{\ell}$ be such that $\mu/\tilde{\ell}$ is $\varepsilon$-close to $\Tr[O\rho]$? 

Suppose we have knowledge of some lower bound (say $\ell_*$) of this quantity, i.e.\ $\Tr[f(H)\rho_0f(H)^{\dag}]=\ell^2\geq \ell_*$. Then, we prove the following: 

\begin{restatable}[Robustness of normalization factors]{theorem}{thmbodynormrobustness}
\label{thm:norm-approx}
Let $\varepsilon\in (0,1)$, $\rho_0$ be some initial state and $P$ be an operator. Furthermore, let $\ell_*\in \mathbb{R}^+$ satisfies $\ell^2=\Tr[P\rho_0 P^{\dag}]\geq \ell_*$, and $O$ be some observable with $\norm{O}\geq 1$. Suppose we obtain an estimate $\tilde{\ell}$ such that
\begin{equation}
\label{eq:normalization-dist}
\left|\tilde{\ell}-\ell^2\right|\leq \dfrac{\varepsilon \ell_*}{3\norm{O}},
\end{equation}
and some parameter $\mu$ such that,
\begin{equation}
\label{eq:estimation-dist-1}
\left|\mu-\Tr[O~P \rho_0 P^{\dag}]\right| \leq  \dfrac{\varepsilon\ell_*}{3},
\end{equation}
then,
$$
\left|\dfrac{\mu}{\tilde{\ell}}-\dfrac{\Tr[O~P \rho_0 P^{\dag}]}{\ell^2}\right| \leq \varepsilon.
$$
\end{restatable}
\begin{proof}
This is proven in Sec.~\ref{sec:appendix-robustness-norm} of the Appendix.
\end{proof}
This Theorem tells us the precision with which both $\mu$ and $\tilde{\ell}$ should be estimated so that $\mu/\tilde{\ell}$ is $\varepsilon$-close to $\Tr[O\rho]$. The overall procedure makes use of Algorithm \ref{algo:randomized-time-evolution}, and is formally stated via Algorithm \ref{algo:single-ancilla-overall}. This involves running Algorithm \ref{algo:randomized-time-evolution} twice: first obtain $\mu$ as stated previously, and then obtain $\tilde{\ell}$, by setting $O=I$ (the identity matrix) and following the same steps. 
\RestyleAlgo{boxruled}
\begin{algorithm}[ht]
  \caption{$\texttt{Single-ancilla-LCU~}(O, \{c_j,U_j\},\rho_0$)}\label{algo:single-ancilla-overall}
  ~\\~\\
  Choose some $T\in O\left(\dfrac{\norm{O}^2\ln(1/\delta)\norm{c}_1^4}{\varepsilon^2\ell_*^2}\right)$.~\\
   \begin{itemize}
  \item[1.~] Run $\texttt{Expectation-observable~}\left(O, \{c_j,U_j\},\rho_0, T\right)$ to obtain $\mu$ such that 
  $$|\mu-\Tr[O~f(H)\rho_0 f(H)^{\dag}]|\leq \varepsilon \ell_*/3.$$
  \item[2.~] Run $\texttt{Expectation-observable~}\left(I, \{c_j,U_j\},\rho_0, T\right)$ to obtain $\tilde{\ell}$ such that 
  $$|\tilde{\ell}-\ell^2|\leq \dfrac{\varepsilon\ell_*}{3\norm{O}}.$$
  \item[3.~] Output $\mu/\tilde{\ell}$.
\end{itemize}
\end{algorithm}

The correctness of the Algorithm is stated formally through the following Theorem.
~\\
\begin{theorem}
\label{thm:single-ancilla-overall}
Let $\varepsilon,\delta\in (0,1)$, $O$ be some observable and $\rho_0$ be some initial state, prepared in cost $\tau_{\rho_0}$. Suppose $H\in\mathbb{C}^{N\times N}$ be a Hermitian matrix such that for some function $f:[-1,1]\mapsto \mathbb{R}$ and unitaries $\{U_j\}_{j}$, $\norm{f(H)-\sum_{j}c_j U_j}\leq \gamma$, where
$$
\gamma\leq\dfrac{\varepsilon\ell_*}{18\norm{O}\norm{f(H)}},
$$ 
and $\ell^2=\Tr[f(H)\rho_0 f(H)^{\dag}]\geq \ell_*$. Furthermore, suppose each unitary $U_j$ is implementable in cost $\tau_j$ such that $\langle\tau\rangle=\sum_{j}c_j\tau_j/\norm{c}_1$. Then Algorithm \ref{algo:single-ancilla-overall} outputs $\mu$ and $\tilde{\ell}$ such that
$$
\left|\mu/\tilde{\ell} - \Tr[O\rho]\right| \leq \varepsilon,
$$
with probability $(1-\delta)^2$, using 
$$
T=O\left(\dfrac{\norm{O}^2 \norm{c}_1^4 \ln(1/\delta)}{\varepsilon^2\ell^{2}_{*}}\right) 
$$  
repetitions of the quantum circuit in Fig.~\ref{fig:single-ancilla-circuit}, the average cost of each coherent run is $2\langle\tau\rangle+\tau_{\rho_0}$. 
\end{theorem}
\begin{proof}
Algorithm \ref{algo:single-ancilla-overall} calls Algorithm \ref{algo:randomized-time-evolution} twice: first to estimate $\mu$ and then to estimate $\tilde{\ell}$. For the upper bound of $\gamma$ defined in the statement of the theorem, we can indeed obtain a $\mu$, using Algorithm \ref{algo:randomized-time-evolution} such that
$$
\left|\mu - \Tr[O~f(H)\rho_0 f(H)^{\dag}]\right|\leq \varepsilon \ell^*/3.
$$ 
This follows from Theorem \ref{thm:randomized-time-evolution}, by simply replacing $\varepsilon$ with $\varepsilon\ell^*/3$. The number of iterations of Theorem \ref{thm:randomized-time-evolution} scales as
$$
T_1=O\left(\dfrac{\norm{O}^2\norm{c}^4_1\ln(1/\delta)}{\varepsilon^2 \ell^2_*}\right),
$$
with each coherent run costing at most $\tau_{\rho_0}+2\langle\tau\rangle$.

For obtaining the estimate $\tilde{\ell}$, we set $O=I$. Furthermore, in Theorem \ref{thm:randomized-time-evolution}, we replace $\varepsilon$ by $\frac{\varepsilon\ell_*}{3~\norm{O}}$. For these parameters, Algorithm \ref{algo:randomized-time-evolution} outputs $\tilde{\ell}$ such that
$$
\left|\tilde{\ell}-\ell^2\right|\leq \dfrac{\varepsilon\ell^*}{3\norm{O}}.
$$
The total number of iterations is
$$
T_2=O\left(\dfrac{\norm{O}^2\norm{c}^4_1\ln(1/\delta)}{\varepsilon^2\ell^2_*}\right),
$$ 
with each coherent run having an average cost $\tau_{\rho_0}+2\langle \tau\rangle$.

Finally from Theorem \ref{thm:norm-approx}, we have that these estimates $\mu$ and $\tilde{\ell}$ satisfy
$$
\left|\dfrac{\mu}{\tilde{\ell}}-\Tr[O\rho]\right|\leq \varepsilon,
$$
where the total number of iterations scales as
$$
T=T_1+T_2=O\left(\dfrac{\norm{O}^2\norm{c}^4_1\ln(1/\delta)}{\varepsilon^2\ell^2_*}\right),
$$ 
with each coherent run requiring an average cost $\tau_{\rho_0}+2\langle\tau\rangle$.
\end{proof}

\subsubsection{Comparison with the standard LCU procedure}

First, we compare the performance of the \textit{Single-Ancilla LCU} technique with the standard LCU procedure. In order to fairly compare the two approaches, we will analyze the various methods by which the standard LCU technique can be used to estimate $\Tr[O\rho]$. As described in Sec.~\ref{subsec:prelim-lcu}, standard LCU prepares the quantum state
\begin{equation}
\label{eq:single_round_LCU}
\ket{\psi_t}=\dfrac{\norm{f(H)\ket{\psi_0}}}{\norm{c}_1}\ket{\bar{0}}\ket{\psi}+\ket{\Phi}^{\perp},
\end{equation}
where
\begin{equation}
\label{eq:normalized-lcu-state}
\ket{\psi}=\dfrac{f(H)\ket{\psi_0}}{\norm{f(H)\ket{\psi_0}}}.
\end{equation}
One obvious advantage of \textit{Single-Ancilla LCU} is that it requires only a single ancilla qubit, while the standard LCU procedure and its variants described below require atleast $\lceil\log_2 M\rceil$ ancilla qubits, and more sophisticated controlled operations. 

Also from Sec.~\ref{subsec:prelim-lcu}, the cost of preparing $\ket{\psi_t}$ is $O(2\tau_R+\tau_Q+\tau_{\psi_0})$. Here, $\tau_{\psi_0}$ is the cost of preparing the initial state $\ket{\psi_0}$, $\tau_R$ is the cost of implementing the \textit{prepare} unitary $R$ that initializes the state of the ancilla registers, and $\tau_Q$ is the cost of implementing the multi-qubit controlled \textit{select} unitary $Q=\sum_{j=1}^{M} \ket{j}\bra{j}\otimes U_j$. 

Let us begin by comparing $\tau_Q$ with $\langle \tau \rangle$, the average cost of implementing unitaries $V_1$ and $V_2$ in each coherent run of the \textit{Single-Ancilla LCU} algorithm. Clearly, $\langle \tau \rangle \leq \tau_Q$. Also, in the general setting, the average cost $\langle\tau\rangle$ cannot exceed the cost of implementing the most expensive unitary, $\tau_{\max}$. At the same time, $\tau_Q$ cannot be lower than the cost of implementing the most expensive unitary. Combining these facts, we have
\begin{equation}
\langle \tau \rangle \leq \tau_{\max} \leq \tau_Q.
\end{equation}
Indeed, when there is no apriori information about the $U_j$'s, $\tau_Q$ can be much greater than even the cost of implementing the most expensive unitary, $\tau_{\max}$. In fact, in the worst case, if every $U_j$ costs $\tau_{\max}$, $\tau_Q=O(M\tau_{\max})$. However, in particular cases, when the $U_i$'s are related, it is possible that both $\tau_Q$ and $\langle\tau\rangle$ scale as $O(\tau_{\max})$. In fact in Lemma 8 of Ref.~\cite{childs2017quantum}, it was shown that when the $U_j$'s  are powers of one single unitary, i.e.\ $U_j=Y^j$, for some unitary $Y$, $\tau_Q=O(\tau_{\max})$. 

Thus, despite requiring only a single ancilla qubit, the cost per coherent run of \textit{Single-Ancilla LCU} is lower than \textit{Standard LCU} (as it additionally requires implementing the \textit{prepare} unitary, requiring $\tau_R$ cost), provided the cost of implementing each $U_i$ is the same for both procedures. Additionally, our method requires no multi-qubit controlled gates, which are unlikely to be implemented in the intermediate-term.

After having prepared $\ket{\psi_t}$, there are three ways in which $\Tr[O\rho_t]$ may be estimated. The overall cost varies based on the choice of the particular variant. We compare each of these three approaches with \textit{Single-Ancilla LCU} by considering the cost of each coherent run, the number of ancilla qubits needed, and the overall cost (which is the product of the cost per run and the number of classical repetitions needed). Overall, there is a trade-off between the cost of each run and the number of classical repetitions as we discuss below.    
~\\~\\
\textbf{Standard LCU with amplitude amplification and classical repetitions:~} One can use the generic LCU procedure to prepare the state $\ket{\psi_t}$ with probability $\norm{f(H)\ket{\psi_0}}^2/\norm{c}^2_1$, and then apply quantum amplitude amplification in order to prepare the state $\ket{\psi}$, with a high probability. This procedure requires a cost
$$
\Gamma_{\max}=O\left(\dfrac{\norm{c}_1}{\sqrt{\ell_*}} \left(2\tau_R+ \tau_Q + \tau_{\psi_0}\right)\right).
$$
This is already higher than the cost of each run by \textit{Single-Ancilla LCU}. 

It is then possible to estimate $\Tr[O\rho]=\braket{\psi|O|\psi}$, by repeatedly measuring the observable $O$, which requires several classical repetitions of the procedure to prepare $\ket{\psi}$, costing $\Gamma_{\max}$. 

For this, we first prove a general result. Broadly, consider any state preparation procedure $V$ that prepares a quantum state $\ket{x}$ to (roughly) $\varepsilon/\norm{O}$-accuracy. Then, for any observable $O$, we can obtain an $\varepsilon$-accurate estimate of $\braket{x|O|x}$ using $O(~\norm{O}^2\ln(1/\delta)/\varepsilon^2)$ runs of $V$. Formally, we have the following lemma:
~\\
\begin{lemma}
\label{lem:state-preparation-observable}
Suppose, there exists a quantum procedure $V$ which, starting from the state $\ket{\psi_0}$, prepares a quantum state $\ket{\tilde{x}}$, such that 
$$
\norm{\ket{\tilde{x}}-\ket{x}}\leq \dfrac{\varepsilon}{2\sqrt{2}\norm{O}}.
$$
Then, to in order to output $\mu$ such that
$$
\left|\mu-\braket{x|O|x}\right|\leq \varepsilon,
$$
with probability at least $1-\delta$,
$$
T=O\left(\dfrac{\norm{O}^2\ln(1/\delta)}{\varepsilon^2}\right)
$$
repetitions of the procedure $V$ is required.
\end{lemma}
\begin{proof}
A single run of $V$, prepares the state $\rho_{\tilde{x}}=\ket{\tilde{x}}\bra{\tilde{x}}$. Then, measuring the observable $O$, outputs an estimate  that lies between $[-\norm{O}, \norm{O}]$. Then, in order to ensure that the estimate is $\varepsilon/2$-close to $\Tr[O \rho_{\tilde{x}}]$, with probability $1-\delta$, we require
$$
T=O\left(\dfrac{\norm{O}^2\ln(1/\delta)}{\varepsilon^2}\right),
$$ 
repetitions of the procedure $V$, which follows from Hoeffding's inequality. Also, 
\begin{align}
\left|\Tr[O~\ket{x}\bra{x}]-\Tr[O~\ket{\tilde{x}}\bra{\tilde{x}}]\right|&\leq \norm{O}_{\infty}\norm{\ket{x}\bra{x}-\ket{\tilde{x}}\bra{\tilde{x}}}_1\\
&\leq 2\norm{O}\sqrt{1-|\braket{x|\tilde{x}}|}\\
&\leq \varepsilon/2.
\end{align}
Thus, after $T$ repetitions, we are able to output $\mu$ such that
\begin{align}
\left|\mu-\braket{x|O|x}\right|&\leq |\mu-\Tr[O~\ket{\tilde{x}}\bra{\tilde{x}}]+\left|\Tr[O~\ket{x}\bra{x}]-\Tr[O~\ket{\tilde{x}}\bra{\tilde{x}}]\right|\\
&\leq \varepsilon/2+\varepsilon/2=\varepsilon.
\end{align}
\end{proof}

We now compare our method to the standard LCU approach to estimate the expectation value of $O$. For any $H$, suppose 
$$
\norm{f(H)-\sum_{j=1}^{M} c_j U_j}\leq \dfrac{\varepsilon}{2\sqrt{2}\norm{O}},
$$ 
such that $\norm{c}_1=\sum_{j}c_j$. Then the standard LCU procedure prepares a quantum state $\ket{\tilde{\psi}}$ such that
$$
\norm{\ket{\psi}-\ket{\tilde{\psi}}}\leq \Theta\left(\dfrac{\varepsilon}{2\sqrt{2}\norm{O}}\right).
$$

Furthermore, from Lemma \ref{lem:state-preparation-observable}, 
$$
T=O\left(\dfrac{\norm{O}^2\ln(1/\delta)}{\varepsilon^2}\right)
$$
repetitions of the standard LCU procedure to estimate $\braket{\psi|O|\psi}$ with probability at least $1-\delta$, where each such run costs $\Gamma_{\max}$. Thus the overall cost of $O(\Gamma_{\max}||O||^2/\varepsilon^2)$ has a better dependence on $\norm{c}_1$ and $\ell_*$, as compared to \textit{Single-Ancilla LCU}.

However, in addition to requiring more ancilla qubits and sophisticated controlled operations, the cost of each run of this procedure is significantly larger than our algorithm. Each run of our algorithm classically samples only two $U_j$'s and implements controlled (over a single ancilla qubit) versions of these sampled unitaries. Thus, the average cost of each run for our procedure $O(\langle \tau \rangle+\tau_{\psi_0}) < \Gamma_{\max}$. In order to estimate $\braket{\psi|O|\psi}$ however, our procedure requires more classical repetitions, and so the overall cost of the \textit{Single-Ancilla LCU} procedure is higher.
~\\~\\
\textbf{Coherent estimation of $\mathbf{\Tr[O\rho]}$ using Standard LCU:~}  It is possible to use quantum amplitude estimation to coherently estimate $\braket{\psi|O|\psi}$ to an additive accuracy $\varepsilon$. This procedure then does not require more than $O(1)$ classical repetitions, and hence the overall cost is lower than the \textit{Single Ancilla LCU procedure}. However, this also increases the cost of each run, as well as number of ancilla qubits substantially. 

First, we need to access (any general, non-unitary) $O$ via a block encoding. For instance, in many cases $O$ is itself a linear combination of unitaries. Constructing this block encoding requires additional ancilla qubits and controlled operations. Note that modern versions of quantum amplitude estimation procedure require only one additional ancilla qubit \cite{rall2020quantum, grinko2021iterative}.

\begin{table}[ht]
\begin{center}
    \resizebox{\columnwidth}{!}{
    \renewcommand{\arraystretch}{3} 
    \begin{tabular}{|c|c|c|c|c|}
    \hline
    Method & No. of. Ancilla qubits & Cost of each coherent run & Classical repetitions \\ \hline\hline
    \vtop{\hbox{\strut ~~~~~~~~~~~~~Standard LCU} \hbox{\strut (with QAA + classical repetitions)}} & $O\left(\log M\right)$~ & $O\left(\dfrac{\norm{c}_1}{\sqrt{\ell_*}} \left(2\tau_R+ \tau_Q + \tau_{\psi_0}\right)\right)$ & $O\left(\dfrac{\norm{O}^2}{\varepsilon^2}\right)$ \\ \hline
   
    \vtop{\hbox{\strut Standard LCU} \hbox{\strut ~~(with QAE)}} & $O(a_O+\log M)$ & $O\left(\dfrac{\alpha_O}{\varepsilon}\left(\dfrac{\norm{c}_1}{\sqrt{\ell_*}}\left(\tau_Q+2\tau_R+\tau_{\psi_0}\right)+T_O\right)\right)$ & $O(1)$  \\ \hline
    \vtop{\hbox{\strut ~~~~~~~Standard LCU} \hbox{\strut (without QAA or QAE)}} & $O(\log M)$ & $O\left(2\tau_R+\tau_Q+\tau_{\psi_0}\right)$ & $O\left(\dfrac{\norm{O}^2\norm{c}^2_1}{\varepsilon^2\ell_*}\right)$ \\ \hline
    \vtop{\hbox{\strut ~~Single-Ancilla LCU} \hbox{\strut ~~~~~~~(this work)}} & $1$ & $2\langle\tau\rangle +\tau_{\psi_0}$ & $O\left(\dfrac{\norm{O}^2\norm{c}^4_1}{\varepsilon^2\ell^2_*}\right)$ \\ 
    \hline
    \end{tabular}}
   \caption{For any $H\in \mathbb{C}^{N\times N}$, this table compares the \textit{Single-Ancilla LCU} method with standard LCU for estimating $\Tr[O\rho]$ to additive accuracy $\varepsilon$, for any measurable observable $O$. Here, $\rho$ is defined in Eq.~\eqref{eq:lcu-state} such that $f(H)\approx \sum_{i}c_i U_i$ and $\norm{c}_1=\sum_i c_i$. Here $\tau_{\psi_0}$ is the cost the unitary that prepares the initial state, $\tau_R$ and $\tau_Q$ denote the cost of the \textit{prepare} unitary $R$, and the \textit{select} unitary $Q$, respectively. Also, $\langle\tau\rangle=\sum_{j}c_j\tau_j/\norm{c}_1$, where $\tau_j$ is the cost of implementing $U_j$. As described in Sec.~\ref{subsec:ancilla-free-LCU}, the standard LCU method can be used to estimate $\Tr[O\rho]$ in three ways: (i) with quantum amplitude amplification (QAA) followed by classical repetitions, (ii) with quantum amplitude estimation (QAE) to coherently estimate the desired expectation value, and (iii) incoherently with just classical repetitions. In particular, for method (ii), we assume that there exists an $(\alpha_O, a_O, 0)$-block encoding of $O$, implementable in a cost $\tau_O$.   
As compared to the \textit{Single-Ancilla LCU} procedure, each of these variants require higher ancilla qubits and multi-qubit controlled operations. Furthermore, as $\langle \tau \rangle \leq \tau_{\max}\leq \tau_{Q}$, the cost of each coherent run of the \textit{Single-Ancilla LCU} is lower than \textit{Standard LCU}, provided the cost of implementing the unitaries $U_j$ is the same for both. This is because:~(a) Standard LCU and its variants require implementing the \textit{prepare} unitary, requiring cost $\tau_R$, and (b) also, $\tau_{\max}$, which is upper bounded by the maximum cost of implementing the most expensive $U_j$, cannot exceed $\tau_Q$, the cost of implementing the \textit{select} unitary. However, the overall complexity (product of the cost of each coherent run and the total number of classical repetitions) of \textit{Standard LCU} scales better than \textit{Single-Ancilla LCU}.  \label{table:comparison-standard-LCU}}
    \end{center}
\end{table}
\renewcommand{\arraystretch}{1}

Let $U_O$ be an $(\alpha_O, a_O, 0)$-block encoding of $O$, requiring cost $\tau_O$. Then, if $T_{\psi}$ is cost of preparing the state $\ket{\psi}$, then the cost of estimating $\Tr[O\rho]$ to $\varepsilon$-additive accuracy is given by
$$
\Gamma_{\max}=O\left(\dfrac{\alpha_O}{\varepsilon}\left(T_{\psi}+T_O\right)\right)=O\left(\dfrac{\alpha_O}{\varepsilon}\left(\dfrac{\norm{c}_1}{\sqrt{\ell_*}}\left(\tau_Q+2\tau_R+\tau_{\psi_0}\right)+T_O\right)\right),
$$
where we have replaced $T_{\psi}$ by the cost of preparing $\ket{\psi}$ using standard LCU with quantum amplitude amplification. As compared to the \textit{Single-Ancilla LCU} method, we find that the overall complexity of coherently estimating $\Tr[O\rho]$ via standard LCU is lower. In particular, the overall dependence on precision is now $O(1/\varepsilon)$, instead of $O(1/\varepsilon^2)$. However the cost of of each coherent run is substantially increased. Moreover, the total number of ancilla qubits needed is now $O(a_O+\log(M))$.
~\\~\\
\textbf{Standard LCU without amplitude amplification or estimation:~}Without amplitude amplification, the standard LCU procedure, prepares the desired quantum state in the second register with probability $p=\norm{f(H)\ket{\psi_0}}^2/\norm{c}^2_1$. Then this would require cost,
$$
\Gamma_{\max}=O\left(2\tau_R+ \tau_Q + \tau_{\psi_0}\right),
$$
which is less than the previous two variants, but still higher than \textit{Single Ancilla LCU} (as $\tau_Q\geq \tau_{\max}\geq \langle\tau\rangle$). Then, the following strategy can be used to estimate the desired expectation value: Measure the first register. Whenever the first register is in $\ket{\bar{0}}$, apply the observable $O$ to the state of the second register. The first register is in state $\ket{\bar{0}}$ with probability $p$, and so, $O(\ln(1/\delta)/p)$ classical repetitions are needed to obtain the same. Overall, these many repetitions of the \textit{Standard LCU} procedure are needed per measurement of $O$. 
Also, one needs to measure $O$, a total of $O(~||O||^2\ln(1/\delta)/\varepsilon^2)$ to estimate $\Tr[O\rho]$ with $\varepsilon$-additive accuracy. Thus, combining these, we obtain that overall
$$
T=O\left(\dfrac{\norm{O}^2\norm{c}^2_1\ln^2(1/\delta)}{\varepsilon^2\ell_*}\right),
$$
repetitions are needed to output the desired expectation value with $\varepsilon$-additive accuracy and at least $(1-\delta)^2$ success probability. This is quadratically lower than the total number of iterations required by the \textit{Single Ancilla LCU} technique, in terms of $\norm{c}_1$ and $\ell_*$. Thus, even in the worst case, this method has a cost per coherent run that is higher than our method, while requiring quadratically lower overall cost. Recall that unlike this procedure, our method required  (i) only a single ancilla qubit, (ii) no multi-qubit controlled operations.

Overall, \textit{Single-Ancilla LCU} provides a generic framework to implement any LCU, using only a single ancilla qubit, and no multi-qubit controlled operations. This makes the method particularly suitable for early fault-tolerant quantum computers. The cost of each coherent run of \textit{Single-Ancilla LCU} is lower than \textit{Standard LCU}, provided the unitaries $U_j$'s are implemented at the same cost for both methods. However, the overall cost of estimating expectation values is lower for standard LCU approaches, as the \textit{Single-Ancilla LCU} technique requires more classical repetitions. In Table \ref{table:comparison-standard-LCU}, we compare the cost (number of ancilla qubits, cost of each coherent run and the number of classical repetitions) of estimating expectation values of observables using the standard LCU approach, with \textit{Single-Ancilla LCU}. The wide applicability of LCU in the development of various quantum algorithms, also implies that our method can be employed to develop novel algorithms for these problems. We apply this technique for several problems of interest such as Hamiltonian simulation (Sec.~\ref{sec:ham-sim}), ground state property estimation (Sec.~\ref{subsec:gsp-single-ancilla}), and extracting useful information from the solution of quantum linear systems (Sec.~\ref{subsec:single-ancilla-qls}). Next, we discuss the \textit{Analog LCU} method.

\subsection{Analog LCU: coupling a discrete primary system to a continuous-variable ancilla}
\label{subsec:analog-LCU}
Suppose we have some Hermitian matrix $H\in \mathbb{C}^{N\times N}$ of unit spectral norm and we wish to implement $f(H)$ for some function $f:[-1,1]\mapsto \mathbb{R}$ which satisfies: 
$$
\left|f(x)-\int_{a}^{b} dz~c(z)\cdot e^{-itxz}\right|\leq \varepsilon,
$$
where $c:\mathbb{R}\mapsto\mathbb{R}^{+}\backslash\{0\}$.

Now suppose $H$ is coupled to a continuous variable system such that the resulting interaction Hamiltonian is $H'=H\otimes \hat{z}$. Suppose the first register is prepared in some initial state $\ket{\psi_0}$ and the ancilla system is prepared in the continuous-variable quantum state
$$
\ket{\bar{0}}_c=\int_{a}^{b}dz~\sqrt{\dfrac{c(z)}{\norm{c}_1}}\ket{z},
$$
where $\norm{c}_1=\int_{a}^b dz~ |c(z)|$. For instance, $\hat{z}$ can represent a degree of freedom (position or momentum) of a one-dimensional quantum Harmonic oscillator, and the state $\ket{\bar{0}}_c$, could be its ground state (a Gaussian), a free resource state for continuous variable systems. For several of our applications, we shall see that this is indeed the case.

Now we shall simply evolve the system according to the interaction Hamiltonian $H'$ to obtain
\begin{align}
\label{eq:final-state-analog-lcu-0}
\ket{\eta_t}=e^{-iH't}\ket{\psi_0}\ket{\bar{0}}_c&=\int_{a}^{b} dz~\sqrt{\dfrac{c(z)}{\norm{c}_1}} e^{-iHtz}\ket{\psi_0}\ket{z}\\
									 &=\dfrac{1}{\norm{c}_1}\int_{a}^{b} dz~c(z) e^{-iHtz}\ket{\psi_0}\ket{\bar{0}}_c+\ket{\Phi}^\perp,
\label{eq:final-state-analog-lcu-1}									 
\end{align}
where $\ket{\Phi}^{\perp}$ is a quantum state (not normalized) such that $(I\otimes \ket{\bar{0}}_c\bra{\bar{0}}_c)\ket{\Phi}^{\perp}=0$. We arrive at the Eq.~\eqref{eq:final-state-analog-lcu-1} from
Eq.~\eqref{eq:final-state-analog-lcu-0} by observing that
$$
(I\otimes \bra{\bar{0}}_c)\int_{a}^{b} dz~\sqrt{\dfrac{c(z)}{\norm{c}_1}} e^{-iHtz}\ket{\psi_0}\ket{z}=\dfrac{1}{\norm{c}_1}\int_{a}^{b} dz~dz'~\delta_{z,z'} \sqrt{c(z)c(z')^*} e^{-itHz}\ket{\psi_0}.
$$
Thus, we have prepared a quantum state that is $O(\varepsilon/\norm{c}_1)$-close to
\begin{equation}
\ket{\psi}=\dfrac{f(H)}{\norm{c}_1}\ket{\psi_0}\ket{\bar{0}}_c+\ket{\Phi}^\perp.
\end{equation}
Now post-selecting on having $\ket{\bar{0}}_c$ in the second register we obtain  a state that is $\varepsilon$-close to $f(H)\ket{\psi_0}/\norm{f(H)\ket{\psi_0}}$ in the first register with probability $\norm{f(H)\ket{\psi_0}}^2/\norm{c}^2_1$. We will use this procedure to develop an analog quantum algorithm for preparing ground states of Hamiltonians in Sec.~\ref{sec:gsp}.

This continuous-time algorithm can be naturally generalized to the scenario where we want to implement $f(H)$ for some function $f:[-1,1]\mapsto \mathbb{R}$ such that 
$$
\left|f(x)-\int_{a_1}^{b_1}dz_1~c(z_1)\int_{a_2}^{b_2}dz_2~c(z_2)\cdots\int_{a_k}^{b_k}dz_k~c(z_k)e^{-itxz_1z_2\cdots z_k}\right|\leq \varepsilon.
$$
This can be implemented by coupling the Hamiltonian $H$ with $k$ different ancillary continuous-variable systems such that the effective interaction Hamiltonian is $\tilde{H}=H\otimes \hat{z}_1\otimes\cdots\otimes\hat{z}_k$. The $j$-th ancilla system is prepared in the quantum state 
$$
\ket{\bar{0}}_{c_j}=\int_{a_j}^{b_j} dz_j~\sqrt{\dfrac{c(z_j)}{\norm{c_j}_1}}\ket{z_j}.
$$
Then by evolving the initial state $\ket{\psi_0}\ket{\bar{0}}_{c_1}\cdots\ket{\bar{0}}_{c_k}$ according to $\tilde{H}$ for time $t$ results in a quantum state that is $O\left(\frac{\varepsilon}{\Pi_{j=1}^k\norm{c_j}_1}\right)$-close to 
\begin{equation}
\label{eq:final-state-analog-lcu-2}
\ket{\eta_t}=\dfrac{f(H)}{\Pi_{j=1}^k\norm{c_j}_1}\ket{\psi_0}\ket{\bar{0}}_{c_1}\cdots\ket{\bar{0}}_{c_k}+\ket{\Phi}^\perp.
\end{equation}
In Sec.~\ref{sec:qls}, our analog quantum linear systems algorithm requires coupling the system Hamiltonian to two ancillary systems, which is captured by this generalization of analog LCU. Interestingly, for the applications we consider, the ancillary states are the ground or the first excited state of a one-dimensional quantum Harmonic oscillator or the ground state of a ``particle in a ring''.
\subsection{Ancilla-free LCU: Randomized unitary sampling} 
\label{subsec:ancilla-free-LCU}
As in the previous sections, suppose we are given a Hermitian matrix $H$ and we wish to implement $f(H)$, which can be approximated by a linear combination of unitaries. That is, for some $\gamma\in [0,1)$,
$$
\norm{f(H)-\sum_{j=1}^{M} c_j U_j}\leq \gamma,
$$
for unitaries $U_j$, $c_j\in\mathbb{R}^{+}\backslash \{0\}$, and $\norm{c}_1=\sum_{j=1}^{M} c_j$. Here, we assume that we are interested in the projection of $f(H)\ket{\psi_0}$ in some subspace of interest, and in the resulting average success probability. Then, we formally prove that instead of implementing \textit{Standard LCU}, it suffices to simply sample $U_j$ according to the distribution of the LCU coefficients. The resulting projection on the subspace of interest is at least as high, on average. This is the key idea behind the \textit{Ancilla-free LCU} technique.

Suppose we obtain a unitary $V$ sampled from the ensemble $\mathcal{D}\sim \{c_j/\norm{c}_1,U_j\}$, and apply it to some initial state $\rho_0$, such that $\rho=V\rho_0 V^{\dag}$. Then, on average, this leads to the following mixed state
\begin{equation}
\label{eq:avg-denisty-matrix-ancilla-free-lcu}
\bar{\rho}=\mathbb{E}\left[\rho\right]=\mathbb{E}[V\rho_0V^{\dag}]=\dfrac{1}{\norm{c}_1}\sum_{j=1}^{M} c_j U_j \rho_0 U^{\dag}_j. 
\end{equation}

If each $U_j$ costs $\tau_j$, then the average cost of preparing $\bar{\rho}$ is given by $\langle\tau\rangle = \sum_j c_j\tau_j/\norm{c}_1$. As with \textit{Single-Ancilla LCU}, this is upper bounded by $\tau_{\max}=\max_j\tau_j$, the cost of implementing the most expensive $U_j$. Now, for some projector $\Pi$ onto the subspace of interest, the average projection of the resulting density matrix $\rho$, is given by $\mathbb{E}[\Tr(\Pi \rho)]=\Tr[\Pi\bar{\rho}]$. Then, it is possible to prove that the projection of $\bar{\rho}$ in this subspace is at least as large as the projection of $f(H)\rho_0 f(H)^{\dag}$, if $\norm{c}_1\leq 1$. We  formally prove this via the following theorem:
~\\
\thmbodyunisamp*
\begin{proof}
Let $g(H)=\sum_{j=1}^{M} c_j U_j$. Then, given any initial state $\rho_0=\ket{\psi_0}\bra{\psi_0}$, if $R$ is the \textit{prepare} unitary while $Q=\sum_{j}\ket{j}\bra{j}\otimes U_j$ is the \textit{select} unitary, then the standard LCU procedure first prepares the state $\ket{\psi'_t}$ where
\begin{align*}
\ket{\psi'_t}&=Q(R\otimes I)\ket{\bar{0}}\ket{\psi_0}\\
			 &=\sum_{j=1}^{M}\sqrt{\dfrac{c_j}{\norm{c}_1}}\ket{j}U_j\ket{\psi_0}.
\end{align*}
Then, the standard LCU procedure implements $(R^{\dag}\otimes I)$ to prepare: 
$$
\ket{\psi_t}=\ket{\bar{0}} \dfrac{g(H)}{\norm{c}_1}\ket{\psi_0} +\ket{\Phi^{\perp}}.
$$
Note that in the last step, the unitary $R^{\dag}$ acts only on the first register, whereas we are interested in the measurement outcomes of the second register. That is, consider the projection $I\otimes\Pi$. Indeed, as $\Pi$ acts only on the second register, we can ignore the last step of \textit{Standard LCU}, (i.e.\ the action of $R^{\dag} \otimes I$) and by the equivalence of partial trace, we obtain the following equality:
\begin{equation}
\label{eq:equivalence-partial-trace}
\Tr\left[(I\otimes \Pi)\ket{\psi_t}\bra{\psi_t}\right]=\Tr\left[(I\otimes \Pi)\ket{\psi'_t}\bra{\psi'_t}\right].
\end{equation}
Next, by randomly sampling $U_j$ according to $\{c_j/\norm{c}_1\}$, we prepare, on average, the density matrix $\bar{\rho}$. It is easy to verify that $\ket{\psi'_t}$ is a purification of $\bar{\rho}$. Then it follows that,
\begin{equation}
\label{eq:pre-lcu-avg-density-matrix}
\Tr\left[(I\otimes \Pi)\ket{\psi'_t}\bra{\psi'_t}\right]=\Tr[\Pi\bar{\rho}].
\end{equation}

Then by combining Eq.~\eqref{eq:equivalence-partial-trace} and Eq.~\eqref{eq:pre-lcu-avg-density-matrix}, we obtain:
\begin{equation}
\label{eq:lcu-avg-density-matrix}
\Tr[\Pi\bar{\rho}]=\Tr\left[(I\otimes \Pi)\ket{\psi_t}\bra{\psi_t}\right].
\end{equation}
Thus, we have
\begin{align}
\Tr[\Pi\bar{\rho}]&=\Tr[(I\otimes \Pi)\ket{\psi_t}\bra{\psi_t}]\\
		   &=\bra{\psi_t}\left(\ket{\bar{0}}\bra{\bar{0}}\otimes \Pi\right)\ket{\psi_t}+\bra{\psi_t}\left[(I-\ket{\bar{0}}\bra{\bar{0}})\otimes \Pi\right]\ket{\psi_t}\\
		   &= \dfrac{1}{\norm{c}^2_1}\bra{\psi_0}g(H)^{\dag} \Pi g(H)\ket{\psi_0}+\bra{\Phi^\perp}\left[(I-\ket{\bar{0}}\bra{\bar{0}})\otimes\Pi \right]\ket{\Phi^\perp}\\
		   &\geq \dfrac{1}{\norm{c}^2_1}\bra{\psi_0}g(H)^{\dag} \Pi g(H)\ket{\psi_0}\\
		   &\geq \bra{\psi_0}f(H)^{\dag} \Pi f(H)\ket{\psi_0}-\varepsilon,
		   \label{eq:final-inequality-ancilla-free-lcu}
\end{align}
where in the last line we have invoked Theorem \ref{thm:distance-expectation}, and the fact that $\norm{c}_1\leq 1$. This completes the proof.
\end{proof}

Theorem \ref{thm:ancilla-free-LCU-theorem} shows that if we are only interested in the projection of $f(H)\rho_0f(H)^{\dag}$ in some subspace of interest (say $\Pi$ is a projector onto this subspace), we can drop the ancilla registers of LCU. Instead, we simply sample a unitary $V$ according to $\{c_j/\norm{c}_1, U_j\}$. Then, on average, the projection of $\rho=V\rho_0V^{\dag}$ in this subspace is at least as large. A similar result naturally extends to the analog LCU framework as well. While the \textit{Standard LCU} procedure requires $\lceil\log M\rceil$ ancilla qubits and sophisticated multi-qubit control gates to implement $f(H)\rho_0 f(H)^{\dag}$ before making the projective measurement, our method establishes that this can be done without any ancilla qubits. At the same time, the average cost per coherent run of this method is given by $\langle\tau\rangle$, which, as discussed previously is upper bounded by $\tau_Q$ (the cost of implementing the \textit{select} gate for \textit{Standard LCU}). Note that sampling the unitaries does not guarantee a higher success probability: it is at least as high, on average. However, one can leverage \textit{Ancilla-free LCU} when we only care about the average success probability of the underlying projective measurement. This is indeed the case for quantum spatial search \cite{ambainis2019quadratic, apers2022quadratic} as we shall discuss next. 

In the spatial search problem, we are interested in the expected number of steps required to find a subset of nodes (marked nodes) of any reversible Markov chain $P$. Starting from the stationary distribution of $P$, the number of steps required by a classical random walk, on average, to find a marked node is known as the hitting time ($HT$). Analogously, for quantum walks, starting from a quantum state proportional to the stationary distribution of $P$, we are also interested in finding the expected number of steps after which a quantum walk has a large overlap with the marked nodes of $P$. In fact, demonstrating that quantum walks require $O(\sqrt{HT})$ steps for any $P$ and any number of marked nodes, had been a long-standing open problem. Through a series of works \cite{szegedy2004quantum, magniez2007search, krovi2016quantum}, a quadratic speedup had been proven only in specific cases (such as particular instances of graphs or when only a single node is marked). The full generic quadratic speedup (up to a log factor) has only been recently proven \cite{ambainis2019quadratic}, using the \textit{Standard LCU} procedure. We show that the framework of quantum spatial search fits into the framework of \textit{Ancilla-free LCU}: here we are interested in the average projection of the quantum walk on to the marked subspace. Thus, the \textit{Standard LCU} of Ref.~\cite{ambainis2019quadratic} can be bypassed, which leads to saving on ancilla qubits while obtaining the same quadratic speedup. Moreover, we show that new quantum spatial search algorithms can be obtained using this framework which also retain a generic quadratic speedup. These results have been described in detail in Sec.~\ref{sec:quantum-random-walks}. 

\section{Applying Single-Ancilla LCU: Hamiltonian simulation}
\label{sec:ham-sim}
In this section, we will use the \textit{Single-Ancilla LCU} procedure to develop a quantum Hamiltonian simulation algorithm that is tailored to early fault-tolerant quantum computers. Consider any Hamiltonian $H$ which is a linear combination of Pauli operators, i.e.\ $H=\sum_{k=1}^{L}\lambda_k P_k$, where $P_k$ is a sequence of Pauli operators, such that $\beta=\sum_k |\lambda_k|$. Let $O$ be some observable and $\rho_0$ be some initial state. Then, we outline a procedure using Algorithm \ref{algo:randomized-time-evolution} (and Theorem \ref{thm:randomized-time-evolution}) that outputs $\mu$ such that 
\begin{equation}
\left|\mu-\Tr[O~e^{-iHt}\rho_0 e^{iHt}]\right|\leq \varepsilon
\end{equation}

Since, $f(H)=e^{-iHt}$ is unitary, $\rho_t=f(H)\rho_0 f(H)^{\dag}$ is a normalized quantum state.  Consequently, $\Tr[O\rho_t]$ can be obtained by simply using Algorithm \ref{algo:randomized-time-evolution} as we do not need to estimate the norm separately. Thus it suffices to apply Theorem \ref{thm:randomized-time-evolution} to output $\mu$.

We decompose $e^{-itH}$ as an approximate linear combination of unitaries, for which we use ideas from the Truncated Taylor series method by Berry et al.\cite{berry2015simulating}, as well the LCU decomposition from \cite{wan2022randomized}. We summarize the key ideas here, while the details can be found in the Appendix (Sec.~\ref{sec-app:ham-sim}). Therein, we show that we can divide the time evolution operator $e^{-itH}$ into $r$ segments and truncate the Taylor series expansion of each such segment after $K$ terms. More precisely, we show that if,
$$
\tilde{S}_r=\sum_{k=0}^{K} \dfrac{(-it\tilde{H}/r)^k}{k!},
$$
then
$\norm{\tilde{S}_r-e^{-itH/r}}\leq \gamma/r$,
for some
$$
K=O\left(\dfrac{\log(r/\gamma)}{\log\log(r/\gamma)}\right).
$$
Moreover, we prove in the Appendix (Sec.~\ref{sec-app:ham-sim}) that each $\tilde{S}_r$ can be decomposed as a linear combination of unitaries $\sum_{j\in M}\alpha_j U_j$, where $M$ can be defined as a multi-indexed set, i.e.\ 
$$
M =\left\{(k, \ell_1, \ell_2,\cdots \ell_k, m): 0\leq k \leq K;  \ell_1,\ell_2, \cdots \ell_k, m \in \{1,2,\cdots, L\}\right\},
$$
and each $U_j$ is a product of $k\leq K$ Pauli operators and a single Pauli rotation, given by
\begin{equation}
\label{eq:simulation-segment-unitary}
U_j=(-i)^k P_{\ell_1}P_{\ell_2}\cdots P_{\ell_k} e^{-i\theta_{m} P_m}.
\end{equation}
Moreover, the sum of the coefficients of this LCU decomposition satisfies $\sum_{j}|\alpha_j|\leq \exp[\beta^2t^2/r^2]$. 

Then, the product of $r$ such segments is also an LCU. That is, $S=\tilde{S}^r_r$ is itself an LCU, i.e.\ $S=\sum_{j}c_jW_j$, $\norm{c}_1=(\sum_{j}|\alpha_j|)^r$ can be proven to be $O(1)$, for $r=\beta^2t^2$. We show that for $\gamma=\frac{\varepsilon}{6~\norm{O}}$, the operator $S/\norm{c}_1$ is close to the overall time evolution operator. That is,
$$
\norm{e^{-itH}-S/\norm{c}_1}\leq\dfrac{\varepsilon}{6\norm{O}}.
$$
This allows us to leverage Theorem \ref{thm:randomized-time-evolution} and Algorithm \ref{algo:randomized-time-evolution} to estimate $\Tr[O\rho_t]$ to $\varepsilon$-accuracy. In order to apply Algorithm \ref{algo:randomized-time-evolution}, we intend to sample $V_1$, $V_2$ such that $\mathbb{E}[V_1]=\mathbb{E}[V_2]=S/\norm{c}_1$. We provide a brief sketch of the sampling strategy next. 
~\\~\\
\textbf{Sampling $\mathbf{V_1}$ and $\mathbf{V_2}$:~} While the sampling strategy is described in detail in the Appendix (Sec.~\ref{sec-app:ham-sim}), we provide a brief outline here. We sample some integer $k\in [0,K]$, according to $\alpha_j/\sum_{j}|\alpha_j|$, and select $k+1$ unitaries comprising of $k$ Pauli operators and a single Pauli rotation. This is then repeated $r$ times to obtain a product of unitaries $W=W_r\cdots W_1$, such that $\mathbb{E}[W]=S/\norm{c}_1$. This allows us to use Algorithm \ref{algo:randomized-time-evolution} and Theorem \ref{thm:randomized-time-evolution}.
~\\~\\
\textbf{Running time:~} The circuit corresponding to Algorithm \ref{algo:randomized-time-evolution} implements controlled (anti-controlled) $V_1$ ($V_2$). So, in order to estimate the cost of each run of the Algorithm, we need to estimate the gate depth of $V_1$ and $V_2$. The sampling procedure described above, outputs a product of $r$ unitaries $W= W_r\cdots W_1$ such that each $W_j$ is itself a product of at most $K+1$ unitaries - $K$ Pauli operators and a single Pauli rotation. Thus the gate depth of $V_1$ is at most $(K+1)r$. So, the overall gate depth is given by $2(K+1)r=O(Kr)$. From the choice of $r$ and $K$, this implies that the gate depth for each run is at most $2\tau_{\max}+\tau_{\rho_0}$, where,
\begin{equation}
\label{eq:t-max-ham-sim}
\tau_{\max}=O(Kr)=O\left(\beta^2 t^2 \dfrac{\log(\beta t\norm{O}/\varepsilon)}{\log\log(\beta t\norm{O}/\varepsilon)}\right).
\end{equation}

The total number of repetitions needed can be obtained from Theorem \ref{thm:randomized-time-evolution}. In this case, for the choice of $r=\beta^2 t^2$, we have $\norm{c}_1=O(1)$. So,
\begin{equation}
T=O\left(\dfrac{\norm{O}^2\ln(1/\delta)}{\varepsilon^2}\right),
\end{equation}
repetitions ensures that Algorithm \ref{algo:randomized-time-evolution} outputs a $\mu$ such that $|\mu - \Tr[O~e^{-iHt}\rho_0 e^{iHt}]|\leq \varepsilon$, with probability $1-\delta$. The overall gate depth is given by 
$$
O(\tau_{\max}.T)=O\left(\dfrac{\beta^2 t^2\norm{O}^2}{\varepsilon^2}\ln(1/\delta)\cdot\dfrac{\log(\beta t\norm{O}/\varepsilon)}{\log\log(\beta t\norm{O}/\varepsilon)}\right).
$$

We formally state our results via the following Theorem:
~\\
\begin{theorem}
\label{thm:ham-sim}
Let $H$ be a Hamiltonian such that $H=\sum_{k=1}^{L} c_k P_k$, where $P_k$ is a sequence of Pauli operators and $c_k>0$ such that $\beta=\sum_{k}c_k$. Suppose $\rho_0$ be some initial state, prepared in gate depth $\tau_{\rho_0}$ and $O$ be any observable. Then provided,
$$
\gamma\leq \dfrac{\varepsilon}{6\norm{O}},
$$ 
such that $\norm{e^{-itH}-S}\leq \gamma$ and,
$$
T=O\left(\dfrac{\norm{O}^2\ln(1/\delta)}{\varepsilon^2}\right),
$$
Algorithm \ref{algo:randomized-time-evolution} outputs $\mu$ such that
$$
\left|\mu-\Tr[O\rho_t]\right|\leq \varepsilon,
$$
using $T$ repetitions of the quantum circuit in Fig.~\ref{fig:single-ancilla-circuit}, with probability $1-\delta$. Each such run has gate depth at most $2\tau_{\max}+\tau_{\rho_0}$ such that
$$
\tau_{\max}=O\left(\beta^2 t^2 \dfrac{\log(\beta t\norm{O}/\varepsilon)}{\log\log(\beta t\norm{O}/\varepsilon)}\right).
$$
\end{theorem}

\begin{proof}
We use Theorem \ref{thm:randomized-time-evolution}, substituting $f(H)=e^{-itH}$ and so $\norm{f(H)}=1$. For the choice of $\gamma$, we require only 
$$
T=O\left(\dfrac{\norm{O}^2\ln(1/\delta)}{\varepsilon^2}\right),
$$
repetitions of Algorithm \ref{thm:randomized-time-evolution}. The gate depth of each coherent run is at most $2\tau_{\max}+\tau_{\rho_0}$, where $\tau_{\max}$ has been obtained in Eq.~\eqref{eq:t-max-ham-sim}.
\end{proof}
~\\
\textbf{Comparison with prior works:~}Let us now compare the cost of our quantum simulation procedure with other established techniques. Throughout, we consider that $H$ is a linear combination of unitaries, as stated in Eq.~\eqref{eq:ham-sim-hamiltonian}. We shall compare the cost of each Hamiltonian simulation algorithm to estimate $\Tr[O\rho_t]$, where $\rho_t=e^{-itH}\rho_0 e^{itH}$. Note that in order to estimate this quantity with additive accuracy $\varepsilon$, these methods need to prepare a quantum state that is $\varepsilon/\norm{O}$-close to $\rho_t$ (from Lemma \ref{lem:state-preparation-observable}). We measure the cost in terms of gate depth per coherent run, number of ancilla qubits required, number of classical repetitions needed, and the overall gate depth given by the product of the gate depth per coherent run and the number of classical repetitions. 
\begin{table}[ht!!]
\begin{center}
    \resizebox{\columnwidth}{!}{
    \renewcommand{\arraystretch}{3} 
    \begin{tabular}{|c|c|c|c|c|c|}
    \hline
    Algorithm & Variant & No. of. Ancilla qubits & Gate depth per coherent run & Classical repetitions \\ \hline\hline
   \multirow{2}{*}{1st order Trotter \cite{childs2021theory}} & Incoherent & $0$ & $O(L\beta^2 t^2\norm{O}/\varepsilon)$ & $O\left(\dfrac{\norm{O}^2}{\varepsilon^2}\right)$ \\ \cline{2-5}
    
    & Coherent & $O(a_O)$ & $O\left(\dfrac{\alpha^2_O L  \beta^2 t^2 }{\varepsilon^2}+\dfrac{\alpha_O T_O}{\varepsilon}\right)$ & $O(1)$ \\ \hline
   
    \multirow{2}{*}{qDRIFT \cite{campbell2019random}} & Incoherent & $0$ & $O\left(\beta^2 t^2\norm{O}/\varepsilon\right)$ & $O\left(\dfrac{\norm{O}^2}{\varepsilon^2}\right)$ \\ \cline{2-5}
    
   & Coherent & $O(a_O)$ & $O\left(\dfrac{\alpha^2_O\beta^2 t^2}{\varepsilon^2}+\dfrac{\alpha_O T_O}{\varepsilon}\right)$ & $O(1)$ \\ \hline
    
   \multirow{2}{*}{$2k$-order Trotter \cite{childs2021theory}} & Incoherent & $0$ & $O\left(L(\beta t)^{1+\frac{1}{2k}}\left(\dfrac{\norm{O}}{\varepsilon}\right)^{\frac{1}{2k}}\right)$ & $O\left(\dfrac{\norm{O}^2}{\varepsilon^2}\right)$ \\ \cline{2-5}
    
   & Coherent & $O(a_O)$ & $O\left(L(\beta t)^{1+\frac{1}{2k}} \left(\dfrac{\alpha_O}{\varepsilon}\right)^{1+\frac{1}{2k}}+\dfrac{\alpha_O T_O}{\varepsilon}\right)$ & $O(1)$ \\ \hline 
   
   \multirow{2}{*}{Truncated Taylor Series \cite{berry2015simulating}} & Incoherent & $O\left(\log(L)\dfrac{\log(\norm{O} t/\varepsilon)}{\log\log(\norm{O} t/\varepsilon)}\right)$ & $O\left(L\beta t \dfrac{\log(\beta t\norm{O}/\varepsilon)}{\log\log(\beta t\norm{O}/\varepsilon)}\right)$ & $O\left(\dfrac{\norm{O}^2}{\varepsilon^2}\right)$ \\ \cline{2-5}
    
   & Coherent & $O\left(a_O+\log(L)\dfrac{\log(\alpha_O t/\varepsilon)}{\log\log(\alpha_O t/\varepsilon)}\right)$ & $O\left(\dfrac{\alpha_O}{\varepsilon}\left(L\beta t \dfrac{\log(\beta t\alpha_O/\varepsilon)}{\log\log(\beta t\alpha_O/\varepsilon)} +T_O\right)\right)$ & $O(1)$\\ \hline
    
     \multirow{2}{*}{Qubitization \cite{low2019hamiltonian}} & Incoherent & $O(\log L)$ & $O\left(L\left(\beta t +\log(\norm{O}/\varepsilon)\right)\right)$ & $O\left(\dfrac{\norm{O}^2}{\varepsilon^2}\right)$ \\ \cline{2-5}
    
   & Coherent & $O(a_O+\log L)$ & $O\left(\dfrac{\alpha_O}{\varepsilon}\left(L\beta t + L\log(\alpha_O/\varepsilon)\right) +\dfrac{\alpha_O T_O}{\varepsilon}\right)$ & $O(1)$\\ \hline
   
 This work  & -- & $1$ & $O\left(\beta^2 t^2 \dfrac{\log(\beta t\norm{O}/\varepsilon)}{\log\log(\beta t\norm{O}/\varepsilon)}\right)$ & $O\left(\dfrac{\norm{O}^2}{\varepsilon^2}\right)$\\
    \hline
    
    \end{tabular}}
  \caption{ Given a Hamiltonian $H=\sum_{j=1}^{L} \lambda_j P_j$, where $P_j$'s are strings of Pauli matrices and $\beta=\sum_{j}|\lambda_j|$, we compare the cost of Hamiltonian simulation by \textit{Single-Ancilla LCU} with other methods. In particular, we compare the number of ancilla qubits, gate depth per coherent run, and the number of classical repetitions needed to output an $\varepsilon$-additive estimate of $\Tr[O\rho_t]$, where $\rho_t=e^{-iHt}\rho_0 e^{iHt}$, for some initial state $\rho_0$ and any measurable observable $O$. The overall gate depth is the product of the last two columns of the table. Here the \textit{incoherent} approach refers to the method of directly measuring $O$, having prepared $\rho_t$. For the \textit{coherent} approach, quantum amplitude estimation is used to estimate $\Tr[O\rho_t]$, given access to an $(\alpha_O, a_O, 0)$ block encoding of $O$.
Finally, Hamiltonian simulation by qubitization \cite{low2019hamiltonian}, assumes a block encoding access to $H$. In our case, the block encoding can be constructed by an LCU approach requiring $O(\log L)$ ancilla qubits and gate depth $O(L)$. \label{table:comparison-ham-sim}}
    \end{center}
\end{table}
\renewcommand{\arraystretch}{1}
Having prepared $\rho_t$, there are two ways in which $\Tr[O\rho_t]$ can be estimated. First, by simply measuring $O$, in which case $O(~\norm{O}^2/\varepsilon^2)$ repetitions are needed in order to output an $\varepsilon$-accurate estimate of $\Tr[O\rho_t]$, with  $\Omega(1)$ probability. The second technique involves estimating this quantity coherently using quantum amplitude estimation. For this, we assume access to a $(\alpha_O,a_O,0)$-block encoding of $O$, which can be implemented in gate depth $T_O$. Then, if $T_{H}$ is the gate depth of the underlying Hamiltonian simulation procedure (to output $\rho_t$ with $O(\varepsilon/\alpha_O)$ accuracy), quantum amplitude estimation estimates $\Tr[O\rho_t]$ to additive accuracy $\varepsilon$ in 
$$
O\left(\dfrac{\alpha_O}{\varepsilon}(T_H+T_O)\right),
$$
with a constant success probability. We are now in a position to compare the cost of our Hamiltonian simulation algorithm with other methods.

Let us begin by considering the first order Trotter method \cite{childs2021theory}. This algorithm has a gate depth of $O(L\beta^2 t^2\norm{O}/\varepsilon)$. So the circuit depth of our procedure has an exponentially better dependence on the precision. Additionally, the total number of terms of $H$, i.e.\ $L$ can be quite large. For instance for several quantum chemistry Hamiltonians $L=\mathrm{poly}(n)$ for an $n$-qubit Hamiltonian, while $\beta \ll L$. However, Trotter-based methods such as First and higher order methods, and qDRIFT do not require any ancilla qubit. Even if $\Tr[O\rho_t]$ is measured coherently using quantum amplitude estimation, not only does the total circuit depth increase, but so does the overall gate depth, given by $O(\alpha^2_O L\beta^2 t^2 \varepsilon^{-2}+\alpha_O T_O \varepsilon^{-1})$. Furthermore, the block encoding of $O$ and the amplitude estimation procedure together require $O(a_O)$ ancilla qubits. Note that this is the worst case complexity of the First order Trotter method. It has been shown that in certain specific instances, when specific bounds on the commutators of the local terms of the Hamiltonian are known, the complexity scales better \cite{haah2021quantum, childs2021theory}.

The complexity of our procedure retains some similarity to randomized quantum simulation schemes such as qDRIFT \cite{campbell2019random} in that it does not depend on $L$. The gate depth for implementing Hamiltonian simulation via qDRIFT is $O(\beta^2 t^2 \norm{O}/\varepsilon)$. Thus, our approach has an exponentially better dependence on precision and $\norm{O}$. The incoherent approach to estimate $\Tr[O\rho_t]$ would need $O(~\norm{O}^2/\varepsilon^2)$ repetitions, which is the same as our method. Even with the coherent approach (using quantum amplitude estimation), the gate depth is $O(\alpha^2_O\beta^2 t^2/ \varepsilon^{2}+ \alpha_O T_O/\varepsilon)$, which is higher than the overall gate depth of our method. Just as the First order Trotter method, qDRIFT requires no ancilla qubits in order to prepare $\rho_t$ and then incoherently estimate $\Tr[O\rho_t]$. However, the coherent approach requires $O(a_O)$ ancilla qubits. 

For higher order Trotter methods \cite{childs2021theory}, the $2k$-th order Trotter procedure has a gate depth per coherent run, given by 
$$
O\left(L (\beta t)^{1+\frac{1}{2k}}\left(\dfrac{\norm{O}}{\varepsilon}\right)^{\frac{1}{2k}}\right),
$$
in order to prepare $\rho_t$ with $\varepsilon/\norm{O}$ accuracy. As compared to the Hamiltonian simulation procedure based on \textit{Single-Ancilla LCU}, the circuit depth of higher order Trotter is sub-exponentially worse in terms of $1/\varepsilon$, but the dependence on $t$ is better. Just as in the previous cases, obtaining an $\varepsilon$-accurate estimate of $\Tr[O\rho_t]$ by measuring $O$, requires $O(\norm{O}^2/\varepsilon^2)$ classical repetitions, each costing $T_H$. Interestingly, the overall gate depth of the coherent technique using quantum amplitude estimation, given by 
$$
O\left(L(\beta t)^{1+\frac{1}{2k}} \left(\dfrac{\alpha_O}{\varepsilon}\right)^{1+\frac{1}{2k}} +\dfrac{\alpha_O}{\varepsilon}T_O\right),
$$
has a better scaling as compared to \textit{Single Ancilla LCU} (in terms of $1/\varepsilon$), at the cost of an increased gate depth and additional $O(a_O)$ ancilla qubits. As with first order Trotter, we have noted the worst case complexity. It has been shown that higher order Trotter methods also perform better in practice for specific local Hamiltonians for which sums of nested commutators of the local terms of $H$ are small \cite{childs2021theory}.

The truncated Taylor series method by Berry et al. makes use of the \textit{Standard LCU} procedure \cite{berry2015hamiltonian}. This requires $O(\log(L)\log(||O||t/\varepsilon)/\log\log(||O||t/\varepsilon))$ ancilla qubits and sophisticated multi-qubit controlled operations. Moreover it needs to implement involved subroutines such as oblivious amplitude amplification for each segment, which cannot be avoided. However, the gate depth per run of this procedure is linear in $t$ and $\beta$, which is quadratically better than our method. However, implementing the LCU has an overhead of $L$ in terms of the gate depth, and thus, in the regime where $\beta\ll L$, there is a range of values of $t$ for which our method requires lower gate cost per coherent run. This, despite the fact that our algorithm requires only a single ancilla qubit and does not require operations that are controlled over multiple qubits. The coherent approach to estimate $\Tr[O\rho_t]$ has a better dependence on $\beta, t$ as well as $1/\varepsilon$ at the cost of an exponentially increased gate depth per coherent run (in terms of $1/\varepsilon$), and $O(a_O+\log(L)\log(\alpha_O t/\varepsilon)/\log\log(\alpha_O t/\varepsilon))$ ancilla qubits.

The state-of-the-art quantum simulation technique makes use of quantum singular value transformation and requires a coherent access to $H$ via a block encoding \cite{low2019hamiltonian}. The complexity of this procedure is optimal (linear in $t+\log(1/\varepsilon)$) in terms of the number of queries made to any such block encoding. Moreover it requires only one more ancilla qubit as compared to our method. However, when $H$ is a linear combination of unitaries, constructing the block encoding itself requires $O(\log L)$ ancilla qubits, implementing multi-qubit controlled unitaries, and gate depth $O(L)$. Then the gate depth of this procedure to prepare $\rho_t$ with $\varepsilon/\norm{O}$ accuracy is given by $O(L\left(\beta t +\log(~\norm{O}/\varepsilon)\right))$. This has optimal scaling in $t$ and $1/\varepsilon$, but has an overhead in terms of $L$, which can be quite large for several Hamiltonians of interest. Thus, for Hamiltonians satisfying $\beta \ll L$, there is a range of $t$ for which our method provides an advantage in terms of the gate depth per coherent run, even with respect to state-of-the-art algorithms. 

In Table \ref{table:comparison-ham-sim}, we compare the cost of the Hamiltonian simulation procedure using \textit{Single-Ancilla LCU} with other methods.

\section{Applications to Ground state preparation and property estimation}
\label{sec:gsp}
The problem of preparing (or extracting useful information from) the ground states of a Hamiltonian finds widespread interest across physics and computer science. Generally, this problem is known to be computationally hard, even for a quantum computer \cite{kempe2006complexity}. However, owing to its wide applicability, novel improved quantum algorithms for ground state preparation (GSP), and ground state property estimation are of extreme importance and interest. We will apply the \textit{Analog LCU} and the \textit{Single-Ancilla LCU} approaches, introduced in Sec.~\ref{sec:new-approaches-LCU}, to tackle these problems. 

This section is organized as follows. We begin by formally describing the ground state preparation (GSP) problem. Next, we discuss an analog quantum algorithm for GSP using the \textit{Analog LCU} framework. Then, we use the \textit{Single-Ancilla LCU} technique to estimate the expectation value of observables with respect to the ground states of a Hamiltonian. We start by describing the ground state preparation problem.
\\~\\
\textbf{The ground state Preparation problem:~} The set up of the problem is similar to prior works \cite{keen2021quantum, he2022quantum, apers2022quadratic}. Suppose we have a Hamiltonian $H$ with ground state $\ket{v_0}$ and ground energy $\lambda_0$, and assume that we are given a lower bound on the gap between the ground state and the first excited state of $H$ (spectral gap), i.e.\ we have knowledge of $\Delta$ such that $|\lambda_1-\lambda_0|\geq \Delta$.\\ 
For clarity of exposition, we assume that the ground space of $H$ is non-degenerate. If this is not the case, e.g.~if the degeneracy of the ground space is $d$ and is spanned by mutually orthonormal eigenstates $\{\ket{v^{(\ell)}_0}\}_{l=1}^{d}$, then we will be preparing a quantum state $\ket{v_0}$ which is a projection onto the ground space given by
$$
\ket{v_0}=\dfrac{1}{\sqrt{\sum_{\ell=1}^{d} |c^{(\ell)}_0|^2}}\sum_{\ell=1}^d c^{(\ell)}_0\ket{v^{(\ell)}_0}.
$$ 
In addition, suppose we have access to some initial state $\ket{\psi_0}$ and a lower bound on the overlap $|\braket{\psi_0|v_0}|=c_0\geq \eta$. \\
Furthermore, for some desired accuracy $\eps\in (0,1)$, we will assume that we know the value of the ground energy to some precision parameter $\eps_g$ such that $\eps_g = \Oo\left(\Delta/\sqrt{\log\frac{1}{\eta\eps}}\right)$. That is, we know some $E_0$ such that 
\begin{equation}
|\lambda_0-E_0| \leq \eps_g.
\end{equation}
By implementing $H-(E_0-\eps_g)I$, we ensure that $0\leq\lambda_0\leq 2\eps_g$. This transformation also ensures that the lower bound for the spectral gap of $H$ remains $\Delta$. Without loss of generality, we assume that the spectrum of $H$ is in the interval $[0,1]$. Otherwise, if $\norm{H}$ is an upper bound on the maximum eigenvalue of $H$, we would consider the Hamiltonian $H/\|H\|$, which has a spectral gap of at least $\Delta/\|H\|$. Thus, the complexities of our algorithm would get rescaled by this $\|H\|$ factor.

\subsection{Applying Analog LCU: A continuous-time quantum algorithm for ground state preparation} 
\label{subsec:analog-gsp}
In this section, we will use the \textit{Analog LCU} framework to develop an analog quantum algorithm for the GSP problem. This algorithm was described in the Supplemental Material of \cite{apers2022quadratic}. Here, we place it in the broader context of the \textit{Analog LCU} framework. Moreover, these results will be used to develop a quantum algorithm for estimating expectation values of observables with respect to the ground states of Hamiltonians in the next Section.  

Consider some quantum system in state $\ket{\psi_0}$ coupled to an ancillary system in a Gaussian state
\begin{equation}
\label{eq:gaussian-state}
\ket{\psi_g}=\int_{-\infty}^{+\infty} \dfrac{dz}{(2\pi)^{1/4}} e^{-z^2/4}\ket{z}.
\end{equation}
The Gaussian state is typically easy to prepare in this setting. This state can be seen as the ground state of a one-dimensional quantum harmonic oscillator. The coupling is done via interaction Hamiltonian $H'=H\otimes\hat{z}$, where $\hat{z}$ corresponds to the position (or momentum) operator. Evolving $\ket{\psi_0}\ket{\psi_g}$, under $H'$ for a time $t$ results in the state
\begin{align}
\ket{\eta_t}&=e^{-itH'}\ket{\psi_0}\ket{\psi_g} \nonumber \\
			\label{eqmain:evolved-state-fast-forwarding-1}
			 &=\int_{-\infty}^{+\infty}\dfrac{dz}{(2\pi)^{1/4}}e^{-z^2/4}e^{-itHz}\ket{\psi_0}\ket{z} \\		 
			 &=\int_{-\infty}^{+\infty}\dfrac{dz}{\sqrt{2\pi}}e^{-z^2/2}e^{-itHz}\ket{\psi_0}\ket{\psi_g}+\ket{\Phi}^\perp, \nonumber
\end{align}
where $\ket{\Phi}^\perp$ is a quantum state with the ancillary system being orthogonal to $\ket{\psi_g}$. Now the Fourier transform of a Gaussian is a Gaussian, i.e.\ we have for any $x \in \mathbb{R}$,
\begin{equation}
\label{eq:hubbard-stratonovich}
e^{-y^2/2}=\int_{-\infty}^{\infty} \dfrac{dz}{\sqrt{2\pi}}~e^{-z^2/2}e^{-iyz}.
\end{equation}
So using this we obtain
\begin{equation}
\label{eqmain:evolved-state-fast-forwarding-3}
    \ket{\eta_t}=e^{-t^2H^2/2}\ket{\psi_0}\ket{\psi_g}+\ket{\Phi}^\perp.
\end{equation}
By post-selecting on obtaining $\ket{\psi_g}$ in the second register, we are able to prepare a quantum state proportional to $e^{-t^2H^2/2}\ket{\psi_0}$ in the first register. In \cite{apers2022quadratic} it was shown that this state is close to the ground state $\ket{v_0}$, with $\eta^2$ probability, provided $t=\widetilde{O}(\Delta^{-1})$. Now we formally state the ground state preparation algorithm of \cite{apers2022quadratic} and its complexity, via the following lemma.
~\\
\begin{lemma}[\cite{apers2022quadratic}]
\label{lem:analog-gsp-ham}
Suppose $\eps\in (0,1)$ and $\eta\in (0, 1/\sqrt{2}]$. Furthermore, suppose we have a Hamiltonian $H$ with ground state $\ket{v_0}$ with $\Delta$ being a lower bound on the spectral gap. Also, the ground state energy of $H$ is known up to a precision $\eps_g\in \Oo\left(\Delta/\sqrt{\log\frac{1}{\eta\eps}}\right)$. Then, given an initial state $\ket{\psi_0}$ satisfying $|\braket{\psi_0|v_0}|\geq \eta$, we output, with probability $\Oo(\eta^2)$, a state $\ket{\phi}$ such that $\nrm{\ket{\phi}-\ket{v_0}}\leq \eps$ by evolving the Hamiltonian $H'=H\otimes \hat{z}$ for time
$$
T = \Oo\left(\dfrac{1}{\Delta}\sqrt{\log \left(\dfrac{1}{\eta\eps}\right)}\right).
$$
\end{lemma}
~\\
While the detailed proof can be found in \cite{apers2022quadratic}, the idea is to evolving the overall system according to the interaction Hamiltonian $H'$ for a time $\sqrt{2t)}$ to prepare the quantum state
\begin{equation}
\ket{\eta_t}=e^{-tH^2}\ket{\psi_0}\ket{\psi_g}+\ket{\Phi}^\perp, \nonumber.
\end{equation} 
The next observation is that by choosing any 
\begin{align}
t>\dfrac{1}{2\Delta^2}\log\left(\dfrac{1-\eta^2}{\eta^2\eps^2}\right),
\end{align}
we ensure that with probability at least $\eta^2$, we prepare in the first register, the state 
$$\ket{\phi}=\dfrac{e^{-tH^2}\ket{\psi_0}}{\norm{e^{-tH^2}\ket{\psi_0}}},
$$ 
which is $\varepsilon$-close to the ground state $\ket{v_0}$. For this choice of $t$, the time evolution of the interaction Hamiltonian should scale as
\begin{equation}
T= \sqrt{2t}=\Oo\left(\dfrac{1}{\Delta}\sqrt{\log \left(\dfrac{1}{\eta\eps}\right)}\right).
\end{equation}
  
Overall, this physically motivated quantum algorithm is significantly simpler than implementing standard LCU in the circuit model. Moreover, hybrid qubit-qumode systems are currently being engineered in a number of quantum technological platforms. In the future, we intend to provide an experimental proposal to implement \textit{Analog LCU} on experimental platforms such as ion traps or superconducting systems. 

Note that evolving the system according to the interaction Hamiltonian $H'$, we obtain the ground state $\ket{v_0}$ in the first register, with probability $\eta^2$ (postselected on measuring $\ket{\psi_g}$ in the ancilla register). Thus, $1/\eta^2$ repetitions of this procedure suffices to prepare $\ket{v_0}$, resulting in a total cost of $T=O\left(\Delta^{-1}\eta^{-2}\sqrt{\log(\eta^{-1}\varepsilon^{-1}})\right)$. Alternatively, $O(1/\eta)$-rounds of quantum amplitude amplification can also be used to prepare $\ket{v_0}$, which would bring down the overall cost by a factor of $1/\eta$. However, amplitude amplification is a discrete procedure with no continuous-time analogue. Indeed in the Appendix (Sec.~\ref{subsec:gsp-qsvt}), we develop a quantum algorithm for ground state preparation in the circuit model for fully fault tolerant quantum computers whose complexity matches that of the state-of-the-art quantum algorithms for this problem. Therein, we implement a polynomial that approximates the function $e^{-tx^2}$, using QSVT.   

Next, we describe how the \textit{Single-Ancilla LCU} technique can be used to develop a randomized quantum algorithm for ground state property estimation.
\subsection{Applying Single-Ancilla LCU: Ground state property estimation}
\label{subsec:gsp-single-ancilla}
In this section, we assume that we can access the Hamiltonian $H$ through the time evolution operator $U_{t}=\exp[-itH]$. Furthermore, given access to $U_{t}$, we can perform the time evolution controlled on a single ancilla qubit. This is referred to as the Hamiltonian evolution (HE) model as has been used in prior works for ground energy estimation using early fault-tolerant quantum computers \cite{lin2022heisenberg,zhang2022computing,wang2022quantum}. Much like these works, we calculate: (a) the maximal time of evolution  of $H$ (controlled by a single ancilla qubit) required in each coherent run, given by $\tau_{\max}$, and (b) the total number of repetitions of the circuit $T$. The total evolution time is then $O(\tau_{\max}\cdot T)$.

Given any Hamiltonian $H$ with ground state $\ket{v_0}$, we will use Algorithm \ref{algo:single-ancilla-overall} and Theorem \ref{thm:single-ancilla-overall} to estimate $\braket{v_0|O|v_0}$ to $\varepsilon$-accuracy, for any measurable observable $O$. The cost of each coherent run of our algorithm will be measured in terms of the maximal time for each $H$ is evolved ($\tau_{\max}$). Additionally, we shall also estimate $T$, the number of classical repetitions required for our procedure, and the number of ancilla qubits. As mentioned previously in Sec.~\ref{subsec:one-ancilla-LCU}, our method requires only a single ancilla qubit, and we compare the complexity of this procedure (in the HE model) with other methods.

Given any initial state $\rho_0=\ket{\psi_0}\bra{\psi_0}$, with overlap of at least $\eta$ with the ground state, prepared in cost $\tau_{\psi_0}$, we use Algorithm \ref{algo:single-ancilla-overall} to obtain an accurate estimate $\Tr[O\rho]$, where
$$
\rho=\dfrac{e^{-tH^2}\rho_0 e^{-tH^2}}{\Tr[e^{-tH^2}\rho_0 e^{-tH^2}]},
$$
for some 
$$
t\in O\left(\dfrac{1}{\Delta^2}\log\dfrac{1}{\eta\varepsilon}\right).
$$
Using Lemma \ref{lem:analog-gsp-ham}, we know that an accurate enough estimate of $\Tr[O\rho]$ is also an accurate estimate of $\braket{v_0|O|v_0}$.  

For this, we consider a discretized version of this LCU decomposition, i.e.\ we approximate $e^{-tH^2}$ as a linear combination of roughly $\sqrt{t}$ terms.  This decomposition has already shown up in prior works \cite{chowdhury2017quantum,chowdhury2021DQC1}. We formally state this via the following Lemma.

\begin{restatable}[LCU decomposition of $e^{-tH^2}$ \cite{chowdhury2017quantum}]{lemma}{lemmalcuexp}
\label{lem:lcu-decomp-exp}
Let $0<\gamma<1$ and consider a Hamiltonian $H$ of unit spectral norm. Furthermore, for any $t>1$, let us define 
\begin{equation*}
X_M= \sum_{j=-M}^M c_j e^{-i j \delta_t \sqrt{2 t} H },   
\end{equation*}
where $M=\left \lceil{\sqrt{2}\left(\sqrt{t}+\sqrt{\log(5/\gamma)}\right)\sqrt{\log(4/\gamma)}}\right \rceil $, $
\delta_t= \left(\sqrt{2t}+\sqrt{2\log(5/\gamma)}\right)^{-1}$ and,
$$
c_j=\frac{\delta_t}{\sqrt{2 \pi}} e^{- j^2 \delta_t^2/2}.
$$
Then,
$$
\nrm{X_M-e^{-t H^2}} \leq \gamma.
$$
\end{restatable}

We will use this LCU decomposition for our randomized quantum algorithm. First note that the $\ell_1$-norm of the LCU coefficients in Lemma \ref{lem:lcu-decomp-exp}, can be upper bounded by a constant. In fact,
\begin{align}
||c||_1&=\sum_{j=-M}^M |c_j|\\&\leq |c_0|+2\sum_{j=1}^{\infty}\frac{\delta_t}{\sqrt{2\pi}}e^{-j^2\delta_t^2/2}\\
&\leq |c_0|+2 \int_0^\infty\frac{e^{-x^2/2}}{\sqrt{2\pi}}dx
=1+|c_0|\leq 1+\delta_t=O(1).
\label{eq:l1-norm-exp}
\end{align}
Second, we will shortly prove that it suffices to consider
$$
\gamma=\dfrac{\varepsilon\eta^2}{30~\norm{O}}.
$$
Now, we are in a position to use Algorithm \ref{algo:single-ancilla-overall}. Recall that, Algorithm \ref{algo:single-ancilla-overall} estimates the expectation value and the norm, by making separate calls to Algorithm \ref{algo:randomized-time-evolution}.

Let us first estimate the cost of each run of our algorithm in the Hamiltonian Evolution model. Each iteration of our randomized quantum algorithm requires implementing (controlled and anti-controlled) versions of unitaries $V_1$, $V_2$ which are i.i.d samples from the ensemble $\{U_j, c_j/\norm{c}_1\}$, where, from Lemma \ref{lem:lcu-decomp-exp}, each $U_j=e^{-ij\delta_t\sqrt{2t}H}$. For this case, it suffices to obtain two integers $j_1$, $j_2$ according to $\{c_j/\norm{c}_1\}$ and then implement $V_1=e^{-ij_1\delta_t\sqrt{2t}H}$ and $V_2=e^{-ij_2\delta_t\sqrt{2t}H}$, respectively. So, the cost of each coherent run can be upper bounded by $2\tau_{\max}+\tau_{\psi_0}$, where $\tau_{\max}$ is the maximum time of evolution for $H$, which can be obtained from Lemma \ref{lem:lcu-decomp-exp} as
\begin{equation}
\label{eq:gstate-max-time-evolution}
\tau_{\max}=M \delta_t\sqrt{2t}=O\left(\sqrt{t\log\left(1/\gamma\right)}\right)=O\left(\dfrac{1}{\Delta}\log\left(\dfrac{\norm{O}}{\eta\varepsilon}\right)\right).
\end{equation}  

We prove the correctness of our algorithm, as well as the total number of repetitions $T$, via the following theorem:
~\\
\begin{theorem}
\label{thm:ground-state-estimation}
Let $\varepsilon, \delta, \gamma \in (0,1)$ and $\eta \in (0,1/\sqrt{2}]$. Suppose $H=\sum_{k=1}^{L} \lambda_k P_k$ is a Hermitian matrix, with ground state $\ket{v_0}$ and let $\ket{\psi_0}$ be some initial state, prepared in cost $\tau_{\psi_0}$, such that $|\braket{v_0|\psi_0}|=\eta$. Let $O$ be some observable. Furthermore, for
$$
t=O\left(\dfrac{1}{\Delta}\log\left(\dfrac{\norm{O}}{\eta\varepsilon}\right)\right),
$$
and,
$$
\gamma=\dfrac{\varepsilon \eta^2}{30~\norm{O}},
$$
suppose,
$$
\norm{e^{-tH^2}-X_M}\leq \gamma.
$$
Then for
$$
T= O\left(\dfrac{\norm{O}^2\ln(1/\delta)}{\varepsilon^2\eta^4}\right),
$$
Algorithm \ref{algo:single-ancilla-overall} outputs, with probability at least $(1-\delta)^2$, parameters $\mu, \tilde{\ell}$ such that
$$
\left|\dfrac{\mu}{\tilde{\ell}}-\braket{v_0|O|v_0}\right|\leq \varepsilon,
$$
using $T$ repetitions of the quantum circuit in Fig.~\ref{fig:single-ancilla-circuit}, and only one ancilla qubit. The maximal time of evolution of $H$ is at most
$$
\tau_{max}=O\left(\dfrac{1}{\Delta}\log\left(\dfrac{\norm{O}}{\varepsilon \eta}\right)\right).
$$
\end{theorem}
\begin{proof}
Define
$$
\ket{\phi}=\dfrac{e^{-tH^2\ket{\psi_0}}}{\norm{e^{-tH^2}\ket{\psi_0}}}.
$$
First observe that $\braket{\psi_0|e^{-2tH^2}|\psi_0}\geq \ell_*=\Omega(\eta^2)$. Then, for the chosen value of $\gamma$, from Theorem \ref{thm:single-ancilla-overall}, the first call to Algorithm \ref{algo:randomized-time-evolution} outputs an estimate $\mu$ such that
$$
\left|\mu-\Tr[O~e^{-tH^2}\rho_0 e^{-tH^2}]\right|\leq \dfrac{\varepsilon\eta^2}{5}.
$$
The second call to Algorithm \ref{algo:randomized-time-evolution} outputs $\tilde{\ell}$ such that
$$
\left|\tilde{\ell}-\ell^2\right|\leq \dfrac{\varepsilon\eta^2}{5\norm{O}}.
$$
Then Algorithm \ref{algo:single-ancilla-overall} outputs $\mu/\tilde{\ell}$, which from Theorem \ref{thm:norm-approx-general} (substituting $a=b=5$), guarantees
$$
\left|\dfrac{\mu}{\tilde{\ell}}-\braket{\phi|O|\phi}\right|\leq \varepsilon/2.
$$
Then, we have:
\begin{align}
\left|\braket{v_0|O|v_0}-\braket{\phi|O|\phi}\right|&\leq \norm{O}_{\infty}\norm{\ket{\phi}\bra{\phi}-\ket{v_0}\bra{v_0}}_1~~~~[\text{~Using Lemma \ref{thm:holder} with~}p=\infty,~q=1]\\
&\leq 2\norm{O}\sqrt{1-|\braket{v_0|\phi}|}.
\label{eq:distance-expectation-gsp}
\end{align}
Now we have
\begin{align}
\left|\braket{\phi|v_0}\right|&\geq 1-\dfrac{1}{2}\norm{\ket{\phi}-\ket{v_0}}^2 \\
                              &\geq 1 - \dfrac{\eta^2}{2(1-\eta^2)}e^{-2t\Delta^2}.
\end{align}
So, by choosing some
$$
t>\dfrac{1}{2\Delta^2}\log\left(\dfrac{8\norm{O}^2(1-\eta^2)}{\varepsilon^2\eta^2}\right),
$$
we ensure that
\begin{align}
\left|\braket{\phi|v_0}\right|\geq 1-\dfrac{\varepsilon^2}{16\norm{O}^2}.
\end{align}
Substituting this back in Eq.~\eqref{eq:distance-expectation-gsp} gives us,
$$
\left|\braket{v_0|O|v_0}-\braket{\phi|O|\phi}\right|\leq \varepsilon/2,
$$
as desired. By triangle inequality, we obtain,
\begin{align}
\left|\dfrac{\mu}{\tilde{\ell}}-\braket{v_0|O|v_0}\right|\leq \left|\dfrac{\mu}{\tilde{\ell}}-\braket{\phi|O|\phi}\right|+\left|\braket{v_0|O|v_0}-\braket{\phi|O|\phi}\right|\leq \varepsilon.
\end{align}
The maximal time of evolution, 
$$
\tau_{\max}=O(\sqrt{t})=O\left(\dfrac{1}{\Delta}\log\left(\dfrac{\norm{O}}{\varepsilon\eta}\right)\right).
$$
The total number of repetitions of the underlying quantum circuit can be obtained by simply substituting in the value of $T$ (in Theorem \ref{thm:ancilla-free-LCU-theorem}), $\ell_*=\eta^2$, and $\norm{c}_1=O(1)$ to obtain
\begin{align}
\label{eq:sample-complexity-gstate}
T=O\left(\dfrac{\norm{O}^2\ln(1/\delta)}{\varepsilon^2 \eta^4}\right),
\end{align}
which completes the proof.
\end{proof}
Thus, the total evolution time is $O(\tau_{\max}.T)=\widetilde{O}(\Delta^{-1}\eta^{-4}\norm{O}^2/\varepsilon^2)$.
~\\~\\
\textbf{Comparison with prior works:~} Here, we compare our method with prior works on ground state preparation: (a) algorithms that make use of Standard LCU, and also, (b) state-of-the-art quantum algorithms for ground state preparation by using QSVT, and (c) early fault-tolerant quantum algorithms for ground state preparation and ground energy estimation. While this is discussed at length in the subsequent paragraphs, the comparison has been summarized in Table \ref{table:comparison-gsp}.

Let us consider ground state preparation algorithms that make use of the \textit{Standard LCU} procedure. We will estimate the performance of these algorithms in the Hamiltonian evolution (HE) model. It is important to note that for \textit{Standard LCU}, we will be dealing with a multi-qubit controlled Hamiltonian evolution oracle: we calculate the maximal time evolution of the multi-qubit controlled unitary $c_m{\text-}U_t$ ($U_t=\exp[-itH]$ controlled over $m$-ancilla qubits), in addition to the number of classical repetitions (overall complexity is measured in terms of the total evolution time which is the product of these two quantities), and the number of ancilla qubits needed. Also, as outlined in Table \ref{table:comparison-standard-LCU}, any generic LCU procedure for preparing $\ket{v_0}$ has three ways which it can estimate $\braket{v_0|O|v_0}$. We will consider each of them.

\begin{itemize}
\item Ge et al. \cite{ge2019faster} use the standard LCU procedure to prepare $\ket{v_0}$. Implementing the LCU requires $O(\log(\log(||O||\eta^{-1}\varepsilon^{-1})/\Delta))$ qubits as ancillae, and multi-qubit controlled operations, which prepares $\ket{v_0}$ with $\varepsilon/\norm{O}$-accuracy. It is possible to use the LCU algorithm to first prepare $\ket{v_0}$ with a constant probability using quantum amplitude amplification, and then measure $O$, using $O(~||O||^2/\varepsilon^2)$ classical repetitions. The maximal time evolution of $H$ per coherent run is $\tilde{O}(\eta^{-1}\Delta^{-1})$, which is higher than our method. However, the number of classical repetitions needed (and also the total evolution time) is lower than \textit{Single-Ancilla LCU}. Overall, the advantage of our method is that requires only a single ancilla qubit, no multi-qubit controlled gates, and lower maximal time of evolution of $H$ per coherent run.
~\\
\item The number of classical repetitions can be reduced to a constant if $\braket{v_0|O|v_0}$ is estimated using quantum amplitude estimation. This technique also reduces the dependence on precision to $1/\varepsilon$ instead of $1/\varepsilon^2$. In this case, we consider an $(\alpha_O,a_O,0)$-block encoding of the observable $O$. The maximal time of evolution per coherent run is 
$\widetilde{O}\left(\alpha_O\varepsilon^{-1} \eta^{-1} \Delta^{-1}\right)$, which is also the total evolution time. Thus, the total evolution time of this approach is lower than our method. However, the maximal evolution time of $H$ per coherent run is exponentially higher (in terms of $1/\varepsilon$). Furthermore, this approach requires $O(a_O+\log(\log(||O||\eta^{-1}\varepsilon^{-1})/\Delta))$ ancilla qubits, while our method requires only a single ancilla qubit and no multi-qubit controlled operations.
~\\
\item Finally, if simply standard LCU is used without quantum amplitude amplification or estimation, the maximal time of evolution matches our method while the number of classical repetitions (and hence the total evolution time) is quadratically lower in terms of $1/\eta$. However, as with the previously mentioned approaches, \textit{Standard LCU} requires: $O(\log(\log(\eta^{-1}\varepsilon^{-1})/\Delta))$ ancilla qubits, and implementing multi-qubit controlled gates. 
\end{itemize}

The quantum algorithm for ground state preparation by Lin and Tong \cite{lin2020near} uses the framework of Quantum Singular Value Transformation. Consequently, it does not require the Hamiltonian evolution operator, and hence cannot be directly compared with our approach. Its complexity however is measured in terms of the number of queries to a block encoding of $H$. Given access to an $(\alpha_H, a, 0)$-block encoding of $H$, their algorithm requires $O(\alpha_H\Delta^{-1}\eta^{-1}\log(1/\varepsilon))$ queries and only $O(1)$ ancilla qubits to prepare $\ket{v_0}$. However, constructing a block encoding of $H$ can be resource demanding and may lead to the use of multiple ancilla qubits and multi-qubit controlled gates. For instance if $H$ is a linear combination of $L$ Pauli terms with $\beta$ being the total weight of the coefficients, the block encoding of $H$ requires $O(\log L)$ ancilla qubits and gate depth $O(L)$ (also, $\alpha_H=\beta$). Thus, this method is not suitable in the early fault-tolerant regime. Nevertheless, much like the \textit{Standard LCU} approach, this algorithm can also be adapted to estimate $\braket{v_0|O|v_0}$ in three ways (as shown in Table \ref{table:comparison-gsp}). The overall query complexity is always lower than our total evolution time. However, besides needing more ancilla qubits, the number of queries per coherent run of the Lin and Tong algorithm can be higher (as $\alpha_H\gg 1$) than the maximal time of evolution $H$ in the \textit{Single-Ancilla LCU} method.

\begin{table}[ht!!]
\begin{center}
    \resizebox{\columnwidth}{!}{
    \renewcommand{\arraystretch}{3} 
    \begin{tabular}{|c|c|c|c|c|c|c|}
    \hline
    Algorithm & Access & Variant & Ancilla & Cost per coherent run & Classical repetitions \\ \hline\hline
  \multirow{3}{*}{Standard LCU \cite{ge2019faster}} & \multirow{3}{*}{HE (MQC)}& QAA + classical repetitions & $O\left(\log\left(\log\left(\frac{\norm{O}}{\eta\varepsilon}\right)/\Delta\right)\right)$ & $\widetilde{O}\left(\Delta^{-1}\eta^{-1}\right)$ & $O\left(\dfrac{\norm{O}^2}{\varepsilon^2}\right)$ \\ \cline{3-6}
    
   &  & QAE & $O\left(a_O+\log\left(\log\left(\frac{\norm{O}}{\eta\varepsilon}\right)/\Delta\right)\right)$ & $\widetilde O\left(\dfrac{\alpha_O}{\varepsilon\eta\Delta}\right)$ & $O(1)$ \\ \cline{3-6}

   &  & Without QAA or QAE & $O\left(\log\left(\log\left(\frac{\norm{O}}{\eta\varepsilon}\right)/\Delta\right)\right)$ & $\widetilde{O}\left(\dfrac{1}{\Delta}\right)$ & $O\left(\dfrac{\norm{O}^2}{\varepsilon^2\eta^2}\right)$ \\ \hline 
   
   \multirow{3}{*}{QSVT \cite{lin2020near}} & \multirow{3}{*}{BE (MQC)}& QAA + classical repetitions & $O(a_H)$ & $\widetilde{O}\left(\dfrac{\alpha_H}{\Delta\eta}\right)$ & $O\left(\dfrac{\norm{O}^2}{\varepsilon^2}\right)$ \\ \cline{3-6}
    
   &  & QAE & $O(a_O+a_H)$ & $\widetilde O\left(\dfrac{\alpha_O\alpha_H}{\varepsilon\eta\Delta}\right)$ & $O(1)$ \\ \cline{3-6}
   
   &  & Without QAA or QAE & $O\left(a_H\right)$ & $\widetilde{O}\left(\dfrac{\alpha_H}{\Delta}\right)$ & $O\left(\dfrac{\norm{O}^2}{\varepsilon^2\eta^2}\right)$ \\ \hline
   
 \multirow{2}{*}{Early Fault-tolerant \cite{dong2022ground}} & HE* & Without QAA & $1$ & $\widetilde{O}\left(\dfrac{1}{\Delta}\right)$ & $O\left(\dfrac{\norm{O}^2}{\varepsilon^2\eta^2}\right)$ \\ \cline{2-6}
 
 & HE* (MQC) & With QAA & $2$ & $\widetilde{O}\left(\Delta^{-1}\eta^{-1}\right)$ & $O\left(\dfrac{\norm{O}^2}{\varepsilon^2}\right)$ \\ \hline
 
 Early Fault-tolerant \cite{zhang2022computing} & HE+BE (MQC) & -- & $O(a_O)$ & $\widetilde{O}\left(\Delta^{-1}\right)$ & $O\left(\dfrac{\alpha^2_O}{\eta^4\varepsilon^2}\right)$ \\ \hline
     
    This work & HE & -- & $1$ & $\widetilde{O}\left(\dfrac{1}{\Delta}\right)$ & $O\left(\dfrac{\norm{O}^2}{\varepsilon^2\eta^4}\right)$\\
    \hline  
    \end{tabular}}
  \caption{Consider any Hamiltonian $H$ with ground state $\ket{v_0}$, an initial quantum state with overlap at least $\eta$ with $\ket{v_0}$, and any measurable observable $O$. Furthermore, suppose we have knowledge of the ground energy with precision $\varepsilon_g=O(\Delta/\sqrt{\log(\eta^{-1}\varepsilon^{-1})})$. In this table, we compare with our technique, the cost of various quantum algorithms for estimating $\braket{v_0|O|v_0}$ to additive accuracy $\varepsilon$. For both Standard LCU \cite{ge2019faster} and our method (\textit{Single-Ancilla LCU}), we consider that $H$ can only be accessed through the time evolution operator $U_{t}=\exp[-iHt]$. This is the Hamiltonian Evolution (HE) model, where the cost of each coherent run is estimated in terms of the maximal time evolution of $H$. The number of repetitions refers to the total number of times the circuit is run. The total evolution time is then the product of these two quantities. The Standard LCU approach requires implementing multi-qubit controlled (MQC) time evolution operators, which is indicated in the Table along with the access model. In the quantum algorithm by Dong et al. \cite{dong2022ground}, the authors estimate the complexity in terms of the number of queries made to the unitary $U=e^{-iH}$. The algorithm involves implementing interleaved $U$ and $U^{\dag}$, controlled over a single qubit ancilla, which implements single qubit phase rotations. This is slightly different from the HE model considered earlier, and hence is denoted as HE* in the Table. Furthermore, the number of queries made to this operator in one run of the underlying circuit determines the cost per coherent run. The QSVT-based algorithm by Lin and Tong \cite{lin2020near} considers the block encoding model (denoted by BE in the Table). More precisely, the cost per coherent run is estimated in terms of the number of queries made to an $(\alpha_H, a_H, 0)$-block encoding of $H$. The product of the number of queries per run and the total number of classical repetitions is the total number of queries to the block encoding of $H$. In order to estimate the desired expectation value directly via quantum amplitude estimation (QAE), we assume that $O$ is accessed via an $(\alpha_O, a_O, 0)$-block encoding. Also the algorithm of Zhang et al. \cite{zhang2022computing} assumes this block encoding for estimating $\braket{v_0|O|v_0}$.
\label{table:comparison-gsp}}
    \end{center}
\end{table}
\renewcommand{\arraystretch}{1}
 
A number of quantum algorithms have been developed for estimating the ground energy of Hamiltonians in the Hamiltonian evolution model, tailored to early fault-tolerant quantum computers. Most of these algorithms also make use of a single ancilla qubit. For instance, in Ref.~\cite{lin2022heisenberg}, Lin and Tong use the Hadamard test (and classical post processing) to achieve the so-called Heisenberg scaling ($1/\varepsilon$ - dependence) for estimating the ground energy to $\varepsilon$-additive accuracy. Their algorithm requires a maximal evolution time $\widetilde{O}(\varepsilon^{-1}\polylog{1/\eta})$, while the total evolution time scales as $\widetilde{O}(\varepsilon^{-1}\eta^{-2})$. The algorithm of Wang et al. \cite{wang2022quantum} on the other hand, uses the same circuit but a different post-processing methodology through which they are able to exponentially improve the dependence on precision with respect to the Hamiltonian evolution time per coherent run. Their result requires a maximal evolution time of $O(\Delta^{-1}\polylog{\varepsilon^{-1}\eta^{-1}\Delta})$, and a total evolution time scaling as $O(\eta^{-2}\varepsilon^{-2}\Delta\polylog{\eta^{-1}\varepsilon^{-1}\Delta})$.

The \textit{Single-Ancilla LCU} method for ground state property estimation assumes that the ground energy of $H$ is known to $O(\Delta/\sqrt{\log (\eta^{-1}\varepsilon^{-1})})$ precision. Thus, one can make use of the algorithm of Wang et al. \cite{wang2022quantum} to first estimate the ground energy and then run our algorithm, without adding any asymptotic overhead either in terms of the maximal time evolution per coherent run or the total evolution time.

Dong, Lin and Tong \cite{dong2022ground} provide ground energy estimation and ground state preparation algorithms for early fault-tolerant quantum computers. Their access model is slightly different: it measures the complexity in terms of the number of queries to $U=e^{-iH}$ (and $U^{\dag}$). The underlying approach is reminiscent of quantum signal processing: interleaved applications of $c{\text-}U$ and $c{\text-}U^{\dag}$ along with single qubit rotations (with adjustable phases). They refer to this as Quantum Eigenvalue Transformation of Unitaries (QETU). Thus, as compared to other early fault-tolerant approaches, this approach has an overhead in terms of the number of single qubit gates needed, which scales linearly with the number of queries made to $U$ and $U^{\dag}$. The authors provide two algorithms for ground energy estimation: (a) The first one requires $\widetilde{O}(\varepsilon^{-1}\log(1/\eta))$ queries to $c{\text-}U$ and $c{\text-}U^{\dag}$ per coherent run, while using only a single ancilla qubit. The total number of queries needed is $\widetilde{O}(\varepsilon^{-1}\eta^{-2})$. (b) The second one makes use of quantum amplitude amplification and binary amplitude estimation to improve the overall query complexity by a factor of $1/\eta$, but at the same time the maximal query-depth per run, also increases by this factor. Furthermore, this approach requires three ancilla qubits and hence, multi-qubit controlled operations. 

Similarly, for ground state preparation there are two algorithms:  the first is a near-optimal ground state preparation algorithm has a maximal query depth of $\widetilde{O}(\eta^{-1}\Delta^{-1})$, which is also the total number of queries, while requiring only two ancilla qubits. Unlike Ref.~\cite{lin2020near}, this does not assume a block encoding access to $H$, but makes us of quantum amplitude amplification, which is hard to implement in the early fault-tolerant regime. This algorithm can be used to estimate $\braket{v_0|O|v_0}$, requiring additional queries scaling as $\norm{O}^2/\varepsilon^2$. Thus, this has a lower overall evolution time as compared to our method, at the cost of requiring only one additional ancilla qubit. However, a higher maximal evolution time per coherent run is needed as compared to \textit{Single-Ancilla LCU}.

The second algorithm requires shorter query depth per coherent run and uses only a single ancilla qubit. The algorithm requires a maximal query depth of $\widetilde{O}(1/\Delta)$, to prepare the ground state with probability $\eta^2$. In order to estimate $\braket{v_0|O|v_0}$, this algorithm will require $O(\norm{O}^2\eta^{-2}/\varepsilon^2)$ classical runs of their circuit, which is quadratically better than our approach. 

However, besides the single qubit gate overhead, there is another drawback of this algorithm which concerns translating query depth to actual gate depth. For this, a specific Hamiltonian simulation procedure must be chosen to implement $U=\exp[-iH]$. State-of-the-art techniques require access to a block encoding of $H$ which leads to an increase in the number of ancilla qubits, as well as the overall gate depth. Thus, Trotter-based methods or Hamiltonian simulation by \textit{Single-Ancilla LCU} need to be leveraged in order to keep the overall ancilla qubit count to one. However, the later method cannot be efficiently incorporated into the framework of Dong et al. This is primarily because our Hamiltonian simulation procedure implements some $S$ such that $S/\norm{c}_1\approx e^{-iH}$, where $\norm{c}_1=O(1)$. Since the subnormalization factor is not unity, many queries to $U$ and $U^{\dag}$ would lead to an exponential overhead ($d$ queries lead to an overhead of $\norm{c}_1^d$) in the final complexity. Consequently, those techniques can be employed which implements $U$ without any (even constant) sub-normalization factor while still keeping the overall ancilla qubits to one. Thus, only Trotter-based Hamiltonian simulation techniques can be incorporated which has a sub-exponentially worse dependence on the gate depth per coherent run (in terms of $1/\varepsilon$). On the other hand, the gate depth per coherent run of our ground state property estimation algorithm which uses the Hamiltonian simulation algorithm by \textit{Single-Ancilla LCU} to implement the time evolution operator (proven in Appendix Sec.~\ref{subsec-app:gsp-qls}) has a $O(\polylog{1/\varepsilon})$ dependence. 

Finally, the recent early fault-tolerant quantum algorithm by Zhang et al. \cite{zhang2022computing} also estimates properties of ground states of Hamiltonians, i.e.\ $\braket{v_0|O|v_0}$. Their algorithm assumes (i) access to $H$ in the Hamiltonian evolution model, and (ii) for any generic observable $O$, an $(\alpha_O, a_O, 0)$-block encoding of $O$. While the maximal time evolution of each run is $\tilde{O}(1/\Delta)$ (same as our method), the number of classical repetitions needed is $O(\alpha_O^2\eta^{-4}\varepsilon^{-2})$, can be higher depending on the specific block encoding. Moreover, unlike our algorithm, the technique of Zhang et al. \cite{zhang2022computing} cannot estimate $\braket{v_0|O|v_0}$ using a single ancilla qubit. This is because constructing the block encoding of $O$ requires several ancilla qubits and multi-qubit controlled operations.

We reiterate that $\tau_{\max}$ is different from the actual circuit depth, which depends on how the Hamiltonian evolution unitary is implemented. Recall that in the early fault-tolerant regime, we are limited by a small ancillary qubit space and the inability to perform multi-qubit controlled operations. This restricts the choice of the underlying simulation algorithm. If $H$ is a linear combination of strings of Pauli operators, i.e.\ $H=\sum_{k=1}^{L}\lambda_k P_k$, (first and higher order) Trotter methods, as well as the \textit{Single Ancilla LCU}- based Hamiltonian simulation algorithm are suitable options. This is because $c{\text-}U_t$ can be implemented using only a single ancilla qubit and no multi-qubit controlled gates. However, these methods require a circuit depth which is super-linear in $\tau$, which would in turn increase the circuit depth of our algorithm, and also the overall cost. On the other hand, state-of-the-art algorithms such as qubitization have an optimal dependence on $t$ (measured in terms of the number of queries made to a block encoding of the Hamiltonian $H$), which would mean that the cost to implement $c{\text-}U_{\tau}$ would be $\tilde{O}(\tau)$. Moreover, the procedure itself requires only two ancilla qubits. However, constructing a block encoding of $H$  could require several ancilla qubits, multi-qubit controlled operations, and also adds to the overall gate depth. For instance when $H$ is a linear combination of Paulis as defined above, the block encoding would require $\lceil\log_2 L\rceil+2$ ancilla qubits, and has a gate depth of $O(L)$. 

In fact, in Appendix \ref{subsec-app:gsp-qls}, we analyze the cost (in terms of the gate depth, ancilla qubits, and number of classical repetitions) of our ground state property estimation algorithm when $U_{\tau}$ is implemented according to the Hamiltonian simulation algorithm in Sec.~\ref{sec:ham-sim}, as well as $2k$-order Trotter \cite{childs2021theory}. We demonstrate that our algorithm can still be implemented using a single ancilla qubit and no multi-qubit controlled gates. Moreover, despite these restrictions, we show that even when compared to state-of-the art techniques, there are regimes where our method has a shorter gate depth per coherent run (See Table \ref{table:comparison-gsp-depth}).

\section{Applications to Quantum linear systems}
\label{sec:qls}
The quantum linear systems algorithm can be stated as follows: Given access to a Hermitian matrix $H\in\mathbb{C}^{N\times N}$ and some initial state $\ket{b}$, prepare the quantum state $\ket{x}=H^{-1}\ket{b}/\norm{H^{-1}\ket{b}}$. Ever since the first quantum algorithm for this problem by Harrow, Hassidim, and Lloyd \cite{harrow2009quantum}, the quantum linear system algorithm has been widely studied. The complexity of this algorithm has been progressively improved through a series of results \cite{childs2017quantum, chakraborty2019power, gilyen2019quantum}. Recently, adiabatic-inspired approaches have also been reported \cite{subacsi2019quantum, lin2020optimal}, which optimally solve this problem \cite{costa2022optimal}. For several applications, simply preparing the state $\ket{x}$ may not be useful. Rather, one is often interested in extracting useful information from this state, such as estimating the expectation value of an observable $O$, i.e.\ $\braket{x|O|x}$.

Just like in the previous section, we apply \textit{Analog LCU} to develop two quantum linear systems algorithms in continuous-time (Sec \ref{subsec:analog-qls}): the first one is an analog variant of the direct approach in \cite{childs2017quantum} while the second one is more amenable to near-term implementation. Following this, we use the \textit{Single-Ancilla LCU} approach to develop a randomized quantum algorithm for estimating $\braket{x|O|x}$ which is implementable on early fault-tolerant quantum computers (Sec.~\ref{subsec:single-ancilla-qls}). 

Let us begin by formally stating the quantum linear systems problem.
~\\~\\
\textbf{Quantum linear systems:~} Suppose we have access to a Hermitian matrix $H\in\mathbb{C}^{N\times N}$ such that its eigenvalues lie in the interval $[-1,-1/\kappa]\cup[1/\kappa, 1]$. Then, given a procedure that prepares the $N$-dimensional quantum state $\ket{b}$, a quantum linear systems algorithm prepares a quantum state that is $O(\varepsilon)$ - close to
$$
\ket{x}=\dfrac{H^{-1}\ket{b}}{\norm{H^{-1}\ket{b}}}.
$$
It is worth noting that the quantum linear systems algorithm is different from its classical counterpart in that by preparing $\ket{x}$, one does not have access to the entries of the classical vector $\vec{x}$. To extract the entire solution vector $\vec{x}$ from the quantum state $\ket{x}$, would require $\Omega(N)$ cost (via tomography). In quantum linear systems, thus, often one is interested in extracting useful information out of the state $\ket{x}$, such as estimating the expectation value $\braket{x|O|x}$, for some observable $O$. 

Also, the assumption that $H$ is a Hermitian matrix is without loss of generality. Given any non-Hermitian $H\in \mathbb{C}^{M\times N}$, there exist efficient procedures to obtain a Hermitian matrix $\tilde{H}$ of dimension $(M+N)\times (M+N)$ \cite{harrow2009quantum}. Then, one may instead implement quantum linear systems with $\tilde{H}$ instead of $H$.

\subsection{Applying Analog LCU: Continuous-time quantum linear systems algorithms}
\label{subsec:analog-qls}
In this section, we develop analog quantum algorithms for solving quantum linear systems. Following the exposition in Sec.~\ref{subsec:analog-LCU}, we shall assume that we are given a system Hamiltonian $H$. We couple this Hamiltonian (the primary system) to two ancillary continuous-variable systems via the interaction Hamiltonian
\begin{equation}
\label{eq:interaction-hamiltonian}
H'=H\otimes \hat{y} \otimes \hat{z}.
\end{equation}
The primary system will be initialized in the quantum state $\ket{b}$ while the two ancillary systems will be in some continuous-variable states. The quantum algorithms developed in this subsection involve evolving the overall system according to $H'$ for some time. Following this, we shall show that the primary system is in the state $\ket{x}$ (or close to it) with an amplitude of $\Omega(1/\kappa)$. 

We begin with the first quantum algorithm, which is an analog analogue of the quantum linear systems algorithm of \cite{childs2017quantum}.
~\\~\\
\textbf{Continuous-time quantum linear systems algorithm:~} Consider the function $f(y)=ye^{-y^2/2}$, where $y\in\mathbb{R}$. As 
\begin{align}
&\int_{0}^{\infty}dy~f(y)=1\\
\implies & \int_{0}^{\infty}dy~f(xy)=1/x,
\end{align} 
which holds for any $x\neq 0$. For any function $g(y)$, suppose its Fourier transform is $\mathcal{F}(g(y))=F(\omega)$, then $\mathcal{F}(g'(y))=i\omega F(\omega)$. If $g(y)=e^{-y^2/2}$, we have that $g'(y)=-ye^{-y^2/2}=-f(y)$. This implies,
\begin{align}
\label{eq:odd-function-fourier}
\dfrac{i}{\sqrt{2\pi}}\int_{-\infty}^{\infty} dz~ ze^{-z^2/2} e^{-izy}= ye^{-y^2/2},
\end{align}
and,
$$
\dfrac{1}{x}=\dfrac{i}{\sqrt{2\pi}}\int_{0}^{\infty} dt \int_{-\infty}^{\infty}dz~ ze^{-z^2/2} e^{-izxt}.
$$
Next, we will prove via a lemma that the upper limit of the outer integral can be truncated at $T=\widetilde{O}(\kappa)$, without introducing significant error.
~\\
\begin{lemma}
\label{lem:truncation-2}
Suppose $\varepsilon >0, z\in \mathbb{R}$, and $x\in \mathbb{R}\setminus\{0\}$. Then there exists $T\in\Theta\left(\kappa\sqrt{\log(\kappa/\varepsilon)}\right)$, such that on the domain  $[-1,-1/\kappa]\cup[1/\kappa,1]$,
\begin{equation}
\label{eq:error-bound}
\left|\dfrac{1}{x}-\dfrac{1}{\sqrt{2\pi}}\int_{0}^{T} dt~\int_{-\infty}^{\infty}dz~ ze^{-z^2/2} e^{-izxt}\right|\leq \varepsilon.
\end{equation}
\end{lemma}
\begin{proof}
We have to evaluate the quantity
$$
\left|\dfrac{1}{\sqrt{2\pi}}\int_{T}^{\infty}dt~\int_{-\infty}^{\infty}dz~ ze^{-z^2/2} e^{-izxt}\right|
$$
We first evaluate the outer integral and obtain,
\begin{align}
\left|\dfrac{1}{\sqrt{2\pi}}\int_{T}^{\infty}dt~\int_{-\infty}^{\infty}dz~ ze^{-z^2/2} e^{-izxt}\right|&=\left|\int_{T}^{\infty}dt~ xt~e^{-x^2t^2/2}\right|~~~~~~[\text{~Using Eq.~\eqref{eq:odd-function-fourier}}~]\\
&=\left|\dfrac{1}{x}\int_{x^2T^2/2}^{\infty} dy~e^{-y}\right|~~~~~~~~~~~[~y=x^2t^2/2~]\\ 
&=\left|\dfrac{1}{x}\cdot e^{-x^2T^2/2}\right|\\
&\leq \dfrac{1}{|x|}\left|e^{-x^2T^2/2}\right|.
\label{eq:bound-gaussian}
\end{align}
Now for $T=\kappa \sqrt{2\log(\kappa/\varepsilon)}$, we have $\left|e^{-x^2T^2/2}\right|\leq \varepsilon/\kappa$. Now as $|x|\geq 1/\kappa$, we have that Eq.~\eqref{eq:bound-gaussian} is upper bounded by $\varepsilon$. So finally, 
\begin{equation}
\left|\dfrac{1}{x}-\dfrac{1}{\sqrt{2\pi}}\int_{0}^{T} dt~\int_{-\infty}^{\infty}dz~ ze^{-z^2/2} e^{-izxt}\right|=\left|\dfrac{1}{\sqrt{2\pi}}\int_{T}^{\infty}dt~\int_{-\infty}^{\infty}dz~ ze^{-z^2/2} e^{-izxt}\right|\leq \varepsilon.
\end{equation}
\end{proof}

In order to design the analog quantum algorithm, consider that the effective interaction Hamiltonian is $H'=H\otimes \hat{y}\otimes \hat{z}$. While the system Hamiltonian $H$ is prepared in some input state $\ket{b}$, the first ancilla system is prepared in the first-excited state of a one-dimensional quantum Harmonic oscillator
\begin{equation}
\label{eq:first-excited-state-harmonic}
\ket{\psi_h}=\dfrac{1}{(2\pi)^{1/4}}\int_{-\infty}^{\infty} dy~ye^{-y^2/4}\ket{y}.
\end{equation} 
The second ancilla system is in the ground state of a ``particle in a ring'' of diameter $1$, given by
\begin{equation}
\ket{\tau}=\int_{0}^{1}dz~\ket{z}.
\end{equation}
Then evolving the overall system according to $H'$ for time $T$, we obtain 
\begin{align}
\ket{\eta_t}&=e^{-i\tilde{H}T}\ket{b}\ket{\psi_h}\ket{\tau}\\
			&=\int_{0}^{1}dz~\int_{-\infty}^{\infty} \dfrac{dy}{(2\pi)^{1/4}}~y e^{-y^2/4} e^{-iyzHT}\ket{b}\ket{y}\ket{z}\\
           &=\dfrac{1}{T}\int_{0}^{T}dt~\int_{-\infty}^{\infty} \dfrac{dz}{\sqrt{2\pi}}~ze^{-z^2/2} e^{-iztH}\ket{b}\ket{\psi_h}\ket{\tau}+\ket{\Phi}^{\perp}~~~~[~\text{Change of variable $t=Ty$}~]
\end{align}
Now, by choosing time $T=\Theta\left(\kappa\sqrt{\log(\kappa/\varepsilon)}\right)$, from Lemma \ref{lem:truncation-2}, we obtain a quantum state that is $O(\varepsilon/T)$-close to
\begin{align}
\ket{\eta_t}=\dfrac{H^{-1}}{T} \ket{b}\ket{\psi_h}\ket{\tau}+\ket{\Phi}^{\perp}.
\end{align}
The cost of preparing this state is thus, linear in $\kappa$ (upto log factors), which is optimal. For fully fault tolerant quantum computers, the state $\ket{x}$ is obtained by using variable time amplitude amplification which is a complicated subroutine, requiring a large number of controlled operations  \cite{childs2017quantum,chakraborty2019power,chakraborty2023quantum}. However, this procedure ensures that the overall query complexity of the quantum linear systems algorithm is still $\widetilde{O}(\kappa)$. Alternatively, in the circuit model, $\widetilde{O}(\kappa)$-rounds of amplitude amplification can yield a quantum state $O(\varepsilon)$-close to $\ket{x}$. The overall cost of this procedure is $\tilde{O}(\kappa^2)$. Note that both these procedures are suitable for implementation on fully fault-tolerant quantum computers. Moreover, procedures such as amplitude amplification have no known continuous-time analogues. 

Thus, for our analog procedure, after preparing the state $\ket{\eta_t}$, we simply post-selecting on obtaining $\ket{\psi_h}$ in the second register. This allows us to obtain $\ket{x}$ in the first register with probability $\widetilde{\Omega}(1/\kappa^2)$.  Thus $O(\kappa^2)$ repetitions of the continuous-time procedure would allow us to obtain $\ket{x}$.

Although this procedure works in general, the quantum state $\ket{\tau}$ might be difficult to prepare experimentally. In fact, for continuous-variable systems, Gaussian states are the easiest to prepare and manipulate \cite{weedbrook2012gaussian}. So, next, we provide a quantum algorithm for which it suffices to prepare both the ancillary registers in Gaussian states.
~\\~\\
\textbf{Continuous-time quantum linear systems algorithm using only Gaussian states:~} The previous quantum algorithm requires us to prepare the non-Gaussian continuous-variable state
\begin{equation*}
\ket{\tau}=\dfrac{1}{\sqrt{T}}\int_{0}^{T}dz~\ket{z}.
\end{equation*}
Since Gaussian states are typically easier to generate and manipulate, let us design alternative algorithms using Gaussian states only. The general idea is to approximate $\int_{-\infty}^{+\infty}dt$ by $\int_{-\infty}^{+\infty}dt~e^{-t^2/2T^2}$ (rather than $\int_{-T}^{T}dt$) for large enough $T$. The analogue of Lemma~\ref{lem:truncation-2} becomes
\begin{lemma}
\label{lem:truncation-gaussian}
Suppose $\varepsilon >0, z\in \mathbb{R}$, and $x\in \mathbb{R}\setminus\{0\}$. Then, there exists $T\geq\kappa^{3/2}/\sqrt{\varepsilon}$, such that on the domain  $[1/\kappa,1]$,
\begin{equation}
\label{eq:truncated-inverse-gaussian}
\left|\dfrac{1}{x}-\dfrac{1}{2\pi}\int_{-\infty}^{+\infty} dt~e^{-t^2/2T^2}~\int_{-\infty}^{+\infty} dz~ e^{-z^2/2}e^{-ixtz} \right|\leq \Theta(\varepsilon).
\end{equation}
\end{lemma}

\begin{proof}
We have, using the fact that the Fourier transform of a Gaussian is a Gaussian
\begin{align*}
\dfrac{1}{2\pi}\int_{-\infty}^{+\infty} dt~e^{-t^2/2T^2}~\int_{-\infty}^{+\infty} dz~ e^{-z^2/2}e^{-ixtz}
&=\dfrac{1}{\sqrt{2\pi}}\int_{-\infty}^{+\infty} dt~e^{-t^2/2T^2}~ e^{-x^2t^2/2}\\
&=\dfrac{1}{\sqrt{2\pi}}\int_{-\infty}^{+\infty} dt~e^{-\left(x^2+1/T^2\right)t^2/2}
=\dfrac{1}{\sqrt{2\pi}}\int_{-\infty}^{+\infty} dt~e^{-\tilde{x}^2t^2/2}\\
&=\frac{1}{\tilde{x}}
\end{align*}
where, we have set $\tilde{x}=\sqrt{x^2+1/T^2}$. Therefore, it remains to bound
\begin{align*}
\left|\dfrac{1}{x}-\dfrac{1}{\tilde{x}}\right|=\left|\dfrac{1}{x}\left(1-\dfrac{x}{\tilde{x}}\right)\right|
\leq\dfrac{1}{\abs{x}}\left|1-\dfrac{1}{\sqrt{1+1/x^2T^2}}\right|\leq \dfrac{1}{\abs{x}}\cdot\dfrac{1}{x^2T^2}\leq\varepsilon
\end{align*}
\end{proof}
Unfortunately, the scaling of $T$ is worse than for the non-Gaussian approach since $T$ scales as $\kappa^{3/2}$ (instead of linear)  and the dependence of precision is $1/\sqrt{\varepsilon}$ (rather than inverse-logarithmic). Moreover, as the Gaussian function is even, the procedure works only for positive semi-definite Hamiltonians. Nevertheless, this allows us to design a quantum linear systems algorithm using only Gaussian states as ancillae. 

Let us again consider the interaction Hamiltonian $H'=H\otimes \hat{y}\otimes \hat{z}$, where $H$ is now some positive definite Hamiltonian with its eigenvalues lying in $[1/\kappa,1]$. We prepare both the ancilla registers in a Gaussian state which, similarly to Sec.~\ref{subsec:analog-gsp}, is defined as follows
\begin{equation*}
\ket{\psi_g}=\dfrac{1}{(2\pi)^{1/4}} \int_{-\infty}^{\infty} dz~ e^{-z^2/4} \ket{z}.
\end{equation*}
Indeed, it suffices to let the state $\ket{b}\ket{\psi_g}\ket{\psi_g}$ evolve under Hamiltonian $H'$ for time $T$ to obtain
\begin{align*}
e^{-iH'T}\ket{b}\ket{\psi_g}\ket{\psi_g}=\dfrac{1}{\sqrt{2\pi}}\int_{-\infty}^{\infty}dz~\int_{-\infty}^{\infty} dy~e^{-(y^2+z^2)/4} e^{-iyzHT}\ket{b}\ket{y}\ket{z}. 
\end{align*}

If we choose some $T\geq\kappa^{3/2}/\sqrt{\varepsilon}$, we have
\begin{align}
\left(I\otimes\ket{\psi_g}\bra{\psi_g}\otimes \ket{\psi_g}\bra{\psi_g}\right)e^{-i\tilde{H}T}\ket{b}\ket{\psi_g}\ket{\psi_g}
&=\dfrac{1}{2\pi}\int_{-\infty}^{+\infty}dz~\int_{-\infty}^{+\infty} dy~e^{-(y^2+z^2)/2} e^{-iyzHT}\ket{b}\ket{\psi_g}\ket{\psi_g}\\
&=\dfrac{1}{2\pi T}\int_{-\infty}^{+\infty}dt~e^{-t^2/2T^2}~\int_{-\infty}^{+\infty} dz~e^{-z^2/2} e^{-itzH}\ket{b}\ket{\psi_g}\ket{\psi_g}
\end{align}
where we have used the change of variable $t=Ty$.
So,
\begin{align}
e^{-i\tilde{H}T}\ket{b}\ket{\psi_g}\ket{\psi_g}&=\dfrac{1}{2\pi T}\int_{-\infty}^{+\infty}dt~e^{-t^2/2T^2}~\int_{-\infty}^{+\infty} dy~e^{-z^2/2} e^{-itzH}\ket{b}\ket{\psi_g}\ket{\psi_g}+\ket{\phi}^{\perp}\\
												  &=\dfrac{H^{-1}}{T} \ket{b}\ket{\psi_g}\ket{\psi_g}+\ket{\Phi}^{\perp}+O(\varepsilon/T)~~~~~~ [~\text{From Lemma \ref{lem:truncation-gaussian}}~].
\end{align}
So, ~by post-selecting on obtaining $\ket{\psi_g}$in the second  register, we obtain a state that is $O\left(\varepsilon/\kappa\right)^{3/2}$ -close to 
$$
\ket{x}=\dfrac{H^{-1}\ket{b}}{\norm{H^{-1}\ket{b}}},
$$
with amplitude $\tilde{\Omega}\left(\sqrt{\varepsilon}/\kappa^{3/2}\right)$. Although the complexity of this algorithm is worse than the continuous-time quantum algorithm in the previous section, it requires only Gaussian states. Consequently, it is more suitable for being implementable in the near term.

One can improve the complexity of this quantum linear systems algorithm by replacing the Gaussian state in the second register with the flat state $\ket{\tau}$. For positive definite Hamiltonians, if the first ancillary system is in a Gaussian state while the second one is in $\ket{\tau}$, we can still obtain a quantum state that is $O(\varepsilon/\kappa)$-close to the solution of the quantum linear systems in time $\widetilde{O}(\kappa)$. This follows from observing 
\begin{equation}
\dfrac{1}{x}=\int_{-\infty}^{\infty} \dfrac{dt}{\sqrt{2\pi}}\int_{-\infty}^{\infty}\dfrac{dz}{\sqrt{2\pi}} e^{-z^2/2}e^{-ixtz}=2\int_{0}^{\infty} \dfrac{dt}{\sqrt{2\pi}}\int_{-\infty}^{\infty}\dfrac{dz}{\sqrt{2\pi}} e^{-z^2/2}e^{-ixtz}.
\end{equation}
We can truncate the outer integral to $T=\Theta(\kappa\sqrt{\log(\kappa/\varepsilon)})$, and introduce only an additive $\varepsilon$-error.

Our analog approach provides a more physically motivated model for implementing quantum linear systems. We believe several existing quantum technological platforms might already be able to engineer these interactions for system Hamiltonians of small dimensions. It would be interesting to explore whether one can obtain a quantum linear systems algorithm by using just a single continuous-variable ancilla.  

Next, we move on to the problem of estimating expectation values of observables with respect to the solution of quantum linear systems. For this, we make use of the \textit{Single-Ancilla LCU} technique.
\subsection{Applying Single-Ancilla LCU: estimating expectation values of observables}
\label{subsec:single-ancilla-qls}

In the \textit{Single-Ancilla LCU} framework, we shall consider that the Hamiltonian $H$ can be accessed through the time evolution operator $U_t=\exp[-itH]$. Much like the ground state property estimation algorithm in this framework, we estimate: (a) the maximal time of evolution per coherent run $(\tau_{\max})$, and (b) the total number of repetitions $T$. The total time of evolution is then the product of $T$ and $\tau_{\max}$. Our randomized quantum algorithms for estimating $\braket{x|O|x}$, for any measurable observable $O$. We consider the discrete LCU decomposition of $H^{-1}$ of Ref.~\cite{childs2017quantum}. We begin by stating the discretized version of the expression in Lemma \ref{lem:truncation-2}.
~\\
Let,
\begin{equation}
\label{eq:discrete-lcu-inverse}
g(x)=\dfrac{i}{\sqrt{2\pi}}\sum_{j=0}^{J-1}\Delta_y\sum_{k=-K}^{K}\Delta_z z_k e^{-z^2_k/2}e^{-ixy_jz_k},
\end{equation}
where $y_j=j\Delta_y$ and $z_k=k\Delta_z$, for some $J\in \Theta(\frac{\kappa}{\gamma}\log(\kappa/\gamma))$, $K=\Theta(\kappa\log(\kappa/\gamma))$, $\Delta_y=\Theta(\gamma/\sqrt{\log(\kappa/\gamma)})$ and $\Delta_z=\Theta((\kappa\sqrt{\log(\kappa/\gamma)})^{-1})$. Then, Childs et al.~\cite{childs2017quantum} proved that $\left|1/x-g(x)\right|\leq \gamma$ in the domain $[-1,-1/\kappa]\cup [1/\kappa,1]$. From this LCU it is clear that in order to approximate $H^{-1}$ in this domain, the time parameter of the Hamiltonian simulation is at most,
\begin{equation}
\label{eq:max-time-ham-sim-inversion}
t=\Theta(y_J z_K)=\Theta\left(\kappa\log(\kappa/\gamma)\right).
\end{equation}
Furthermore, from \cite{childs2017quantum}, the $\ell_1$-norm of the LCU coefficients were shown to be upper bounded by $\norm{c}_1=\Theta(\kappa\sqrt{\log(\kappa/\gamma)})$. This LCU decomposition allows us to use Algorithm \ref{algo:single-ancilla-overall} to estimate $\braket{x|O|x}$. We will prove that it suffices to choose  
$$
\gamma=\dfrac{\varepsilon}{18 \norm{O}}.
$$ 
From Eq.~\eqref{eq:max-time-ham-sim-inversion} and the aforementioned choice of $\gamma$, the maximal time evolution for each coherent run of our algorithm is 
$$
\tau_{\max}=O\left(\kappa\log\left(\dfrac{\norm{O}\kappa}{\varepsilon}\right)\right).
$$
We formally prove the correctness of our method via the following theorem
~\\
\begin{theorem}[Expectation values of observables with respect to the solution of quantum linear systems]
\label{thm:inversion-estimation-2}
Let $H$ be a Hermitian matrix such that its non zero eigenvalues lie in $[-1,-1/\kappa]\cup [1/\kappa,1]$. Let $O$ be an observable, and $\varepsilon, \delta, \gamma \in (0,1)$. Then if
$$
\gamma = \dfrac{\varepsilon}{18\norm{O}},
$$
such that
$$
\norm{H^{-1}-g(H)}\leq \gamma,
$$
and, 
$$
T= O\left(\dfrac{\norm{O}^2\kappa^4\log^2\left(\frac{\norm{O}\kappa}{\varepsilon}\right)\ln(1/\delta)}{\varepsilon^2}\right).
$$
then Algorithm \ref{algo:single-ancilla-overall} outputs, with probability at least $(1-\delta)^2$, parameters $\mu$ and $\tilde{\ell}$ such that
$$
\left|\dfrac{\mu}{\tilde{\ell}}-\braket{x|O|x}\right|\leq \varepsilon,
$$
using $T$ repetitions of the quantum circuit in Fig.~\ref{fig:single-ancilla-circuit}, and only one ancilla qubit. The maximal time evolution of $H$ in each coherent run is,
$$
\tau_{max}=O\left(\kappa\log\left(\dfrac{\norm{O}\kappa}{\varepsilon}\right)\right).
$$ 
\end{theorem}
\begin{proof}
First observe that $\kappa\geq \norm{f(H)}=\norm{H^{-1}}\geq 1$, and similarly $\kappa^2\geq \ell^2=\norm{H^{-1}\ket{b}}^2\geq \ell_*=1$. So, after the substitution the appropriate parameters, we find that choice of $\gamma$ is the same as in Theorem \ref{thm:single-ancilla-overall}. Also, for this choice of $\gamma$, $\ell_1$-norm of the LCU coefficients of $g(H)$ is
\begin{align*}
\norm{c}_1=O\left(\kappa\sqrt{\log\left(\frac{\kappa\norm{O}}{\varepsilon}\right)}\right).
\end{align*}
So, for $\ell_*=1$, Algorithm \ref{algo:single-ancilla-overall} outputs parameters $\mu$ and $\tilde{\ell}$ such that
$$
\left|\mu-\Tr[O~H^{-1}\ket{b}\bra{b}H^{-1}]\right|\leq \varepsilon/3,
$$
and,
$$
\left|\tilde{\ell}-\ell^2\right|\leq \dfrac{\varepsilon}{3\norm{O}}.
$$
From Theorem \ref{thm:norm-approx}, the parameters $\mu$ and $\tilde{\ell}$ satisfy,
$$
\left|\dfrac{\mu}{\tilde{\ell}}-\braket{x|O|x}\right|\leq \varepsilon.
$$
Each coherent run quantum circuit costs no more than $2\tau_{\max}+\tau_b$, where  
\begin{equation}
\label{eq:max-time-evolve-qls}
\tau_{\max}= O\left(\kappa\log\left(\dfrac{\norm{O}\kappa}{\varepsilon}\right)\right),
\end{equation}
The total number of iterations required can be obtained from Theorem \ref{thm:single-ancilla-overall} by substituting the appropriate values of $\norm{c}_1$ and $\ell^*$ as
\begin{equation}
T=O\left(\dfrac{\norm{O}^2\kappa^4\log^2\left(\frac{\norm{O}\kappa}{\varepsilon}\right)\ln(1/\delta)}{\varepsilon^2}\right).
\end{equation}
\end{proof}
~\\~\\
It is important to note that our quantum algorithm extracts useful information about the quantum state $\ket{x}$. It outputs a number as opposed to a quantum state, which is distinct from the quantum linear systems problem in general. The total time evolution of $H$ is $\widetilde{O}(\kappa^5\norm{O}^2/\varepsilon^2)$.
~\\~\\
\textbf{Comparison with prior works:~}Let us now compare the complexity of our procedure with that of other quantum linear systems algorithms. Much like the \textit{Single-Ancilla LCU} algorithm for ground state property estimation, we compare our algorithm with quantum linear systems algorithms making use of (a) the \textit{Standard LCU} method \cite{childs2017quantum}, (b) QSVT \cite{gilyen2019quantum, chakraborty2023quantum} and the (c) discrete adiabatic theorem \cite{costa2022optimal}. We summarize the comparison in Table \ref{table:comparison-qls}.

First note that the HHL algorithm \cite{harrow2009quantum} requires access to a time-evolution oracle and can estimate $\braket{x|O|x}$. The algorithm also makes use of quantum phase estimation, which using modern methods can be implemented with $O(1)$ ancilla qubits. This algorithm can estimate the expectation value in three ways:  Preparing $\ket{x}$ first using quantum amplitude amplification requiring a maximal time evolution of $\widetilde{O}(\kappa^2\norm{O}/\varepsilon)$, followed by $O(~||O||^2/\varepsilon^2)$ measurements of $O$. Thus, the maximal time evolution per coherent run is exponentially worse than \textit{Single-Ancilla LCU}. However, the total evolution time has a better dependence on $\kappa$, but a worse dependence on other parameters. Moreover, the coherent estimation of the expectation value by quantum amplitude estimation cannot reduce the dependence on $\varepsilon$, which continues to be $O(1/\varepsilon^2)$, while needing more ancilla qubits (owing to the block encoding of $O$). Finally, if quantum amplitude amplification or estimation are not used, the maximal time evolution of $H$ per coherent run is $O(\kappa\norm{O}/\varepsilon)$, which is exponentially worse than our method. On the other hand, the total evolution time is $O(\kappa^3\norm{O}^3/\varepsilon^3)$ has a better dependence on $\kappa$, but a worse dependence on all other parameters. 

As before, for Standard LCU, we consider access to the Hamiltonian evolution oracle (HE access model) while for QSVT and the adiabatic approaches, we consider the block encoding framework (BE access model). For Standard LCU, we consider time evolution of a multi-qubit controlled time evolution operator and compare the maximal time of evolution of $H$, the number of classical repetitions with our method. Although a direct comparison cannot be made between the HE model and the BE model, for QSVT and adiabatic based approaches, we consider the number of queries made to an $(\alpha_H, a_H, 0)$-block encoding of $H$ in one coherent run of the algorithm, and the total number of runs required. In addition to this, for both these access models, we compare the number of ancilla qubits needed.

The LCU-based procedure by Childs, Kothari and Somma \cite{childs2017quantum} requires $O(\log(\kappa\norm{O}/\varepsilon))$ ancilla qubits and sophisticated multi-qubit controlled operations. There are three ways to estimate $\braket{x|O|x}$:

\begin{itemize}
\item If $\ket{x}$ is prepared first using quantum amplitude amplification, the maximal time evolution of (multi-qubit controlled) $H$ per coherent run is $O(\kappa^2\log^2(\kappa\norm{O}/\varepsilon))$, which is higher than our method by a factor of $\kappa$. However, the total evolution time is lower by a factor of $O(\kappa^3)$, as only $O(||O||^2/\varepsilon^2)$ repetitions of this procedure is needed.

\item Given access to an $(\alpha_O, a_O, 0)$-block encoding of $O$, the desired expectation value can be measured coherently using quantum amplitude estimation. This reduces the overall dependence on the precision to $1/\varepsilon$, at the cost of increasing the maximal time of evolution of $H$ to $\widetilde{O}(\kappa^2\alpha_O/\varepsilon)$, which is also the total evolution time. Furthermore, the number of ancilla qubits needed increases by $a_O$, owing to the implementation of the block encoding of $O$.

\item The maximal evolution time per coherent run can be minimized by avoiding the use of quantum amplitude amplification or estimation. Standard LCU followed by a direct measurement of $O$ leads to a maximal time evolution of $O(\kappa\log(||O||\kappa/\varepsilon))$ which matches that of \textit{Single-Ancilla LCU}. On the other hand, the total number of repetitions is $\widetilde{O}(\kappa^2\norm{O}^2/\varepsilon^2)$ which is lower than our method by a factor of $\kappa^2$. However, the overhead due to ancilla qubits and multi-qubit controlled operations remain. 
\end{itemize}

The QSVT based approach \cite{gilyen2019quantum} queries an $(\alpha_H, a_H, 0)$-block encoding of $H$ (say $U_H$), and implements a polynomial approximation of $H^{-1}$ using queries to $U_H$ (and $U^{\dag}_H$), interleaved with a single qubit phase rotations. So, we will measure the complexity in terms of the query complexity per coherent run as well as the overall queries. Much like standard LCU, it can estimate $\braket{x|O|x}$ in three ways: 

\begin{itemize}
\item Preparing $\ket{x}$ by QSVT followed by amplitude amplification requires $\widetilde{O}(\alpha_H\kappa^2\log(\kappa\norm{O}/\varepsilon))$ queries to $U_H$ per coherent run, followed by $O(||O||^2/\varepsilon^2)$ classical repetitions. The procedure requires $O(a_H)$ ancilla qubits and multi-qubit controlled operations to construct the block encoding.

\item Directly using quantum amplitude estimation to estimate the desired expectation value assumes access to an $(\alpha_O, a_O,0)$-block encoding of $O$, and requires $\widetilde{O}\left(\alpha_H\alpha_O\kappa^2\varepsilon^{-1}\right)$ queries per coherent run which is also the overall query complexity. The number of ancilla qubits needed in the overall procedure scales as the number of ancilla qubits required to construct the respective block encodings, i.e.\ $O(a_H+a_O)$. This can be quite large depending on the way this block encoding is constructed. Thus, the overall query complexity of this procedure is lower than the total time of evolution of our algorithm.

\item If quantum amplitude amplification or estimation is not used, $\braket{x|O|x}$ can be estimated by implementing the polynomial approximation of $H^{-1}$ by querying the block encoding followed by a measurement of $O$. This method has a reduced query complexity per coherent run given by $O(\alpha_H\kappa\log(\kappa||O||/\varepsilon))$, which matches the $\tau_{\max}$ of our method (in terms of $\kappa$) but the linear dependence on $\alpha_H$ means there are regimes where our method requires less cost per coherent run. The total number of classical repetitions has a quadratically better dependence on $\kappa$ as compared to our method given by $\widetilde{O}(\kappa^2\norm{O}^2/\varepsilon^2)$. However constructing the block encoding requires $O(a_H)$ ancilla qubits and the cost per coherent run has a dependence on $\alpha_H$ (which can be $O(\log L)$ and $\beta$ respectively, when $H$ is a linear combination of Pauli operators as described in Eq.~\eqref{eq:ham-sim-hamiltonian}). 
\end{itemize}

\begin{table}[ht!!]
\begin{center}
    \resizebox{\columnwidth}{!}{
    \renewcommand{\arraystretch}{3} 
    \begin{tabular}{|c|c|c|c|c|c|c|}
    \hline
    Algorithm & Access & Variant & Ancilla & Cost per coherent run & Classical repetitions \\ \hline\hline
  \multirow{3}{*}{Standard LCU \cite{childs2017quantum}} & \multirow{3}{*}{HE (MQC)}& QAA + classical repetitions & $O\left(\log(\kappa\norm{O}/\varepsilon)\right)$ & $\widetilde{O}(\kappa^2)$ & $O\left(\dfrac{\norm{O}^2}{\varepsilon^2}\right)$ \\ \cline{3-6}
    
   &  & QAE & $O\left(a_O+\log(\kappa\norm{O}/\varepsilon)\right)$ & $\widetilde{O}\left(\alpha_O\kappa^2/\varepsilon\right)$ & $O(1)$ \\ \cline{3-6}

   &  & Without QAA or QAE & $O\left(\log(\kappa\norm{O}/\varepsilon)\right)$ & $O\left(\kappa\log(\kappa\norm{O}/\varepsilon)\right)$ & $\widetilde{O}\left(\dfrac{\kappa^2\norm{O}^2}{\varepsilon^2}\right)$ \\ \hline 
   
   \multirow{3}{*}{QSVT \cite{gilyen2019quantum}} & \multirow{3}{*}{BE (MQC)}& QAA + classical repetitions & $O(a_H)$ & $\widetilde{O}\left(\alpha_H\kappa^2\right)$ & $O\left(\dfrac{\norm{O}^2}{\varepsilon^2}\right)$ \\ \cline{3-6}
    
   &  & QAE & $O(a_O+a_H)$ & $\widetilde {O}\left(\dfrac{\alpha_O\alpha_H\kappa^2}{\varepsilon}\right)$ & $O(1)$ \\ \cline{3-6}
   
   &  & Without QAA or QAE & $O\left(a_H\right)$ & $O\left(\alpha_H\kappa\log(\kappa\norm{O}/\varepsilon)\right)$ & $\widetilde{O}\left(\dfrac{\kappa^2\norm{O}^2}{\varepsilon^2}\right)$ \\ \hline
   
 \multirow{2}{*}{Discrete adiabatic theorem \cite{costa2022optimal}} & \multirow{2}{*}{BE (MQC)} & Classical repetitions & $O(a_H)$ & $O\left(\alpha_H\kappa\log(\norm{O}/\varepsilon)\right)$ & $O\left(\dfrac{\norm{O}^2}{\varepsilon^2}\right)$ \\ \cline{3-6}
 
 &  & QAE & $O(a_O+a_H)$ & $\widetilde {O}\left(\dfrac{\alpha_O\alpha_H\kappa}{\varepsilon}\right)$ & $O(1)$ \\ \hline
     
    This work & HE & -- & $1$ & $O\left(\kappa\log(\kappa \norm{O}/\varepsilon)\right)$ & $\widetilde{O}\left(\dfrac{\kappa^4\norm{O}^2}{\varepsilon^2}\right)$\\
    \hline  
    \end{tabular}}
  \caption{Consider any Hamiltonian $H$ with eigenvalues in $[-1,-1/\kappa]\cup [1/\kappa, 1]$. Then given an input state $\ket{b}$, define $\ket{x}=H^{-1}\ket{b}/\norm{H^{-1}\ket{b}}$, and suppose $O$ is any measurable observable.  In this table, we compare with our technique, the cost of various quantum algorithms for estimating $\braket{x|O|x}$ to additive accuracy $\varepsilon$. For both Standard LCU \cite{ge2019faster} and our method (\textit{Single-Ancilla LCU}), we consider that $H$ can only be accessed through the time evolution operator $U_{t}=\exp[-iHt]$. This is the Hamiltonian Evolution (HE) model, where the cost of each coherent run is estimated in terms of the maximal time evolution of $H$. The number of repetitions refers to the total number of times the circuit is run. The total evolution time is then the product of these two quantities. The Standard LCU approach requires implementing multi-qubit controlled (MQC) time evolution operators, which is indicated in the Table along with the access model. Both the QSVT-based algorithm in Ref.~\cite{gilyen2019quantum}, and the state-of-the-art algorithm by Costa et al. \cite{costa2022optimal}, considers the block encoding access model (denoted by BE in the Table). The cost is estimated in terms of the number of queries made to an $(\alpha_H, a_H, 0)$-block encoding of $H$. The cost per coherent run is the number of queries made, while the product of this quantity and the number of classical repetitions is the overall query complexity.   Finally, in order to estimate the desired expectation value directly via quantum amplitude estimation (QAE), we assume that $O$ is accessed via an $(\alpha_O, a_O, 0)$-block encoding.
\label{table:comparison-qls}}
    \end{center}
\end{table}
\renewcommand{\arraystretch}{1}

It is important to note that for both the aforementioned approaches the query complexity per coherent run can be brought down to $\widetilde{O}(\kappa)$ by using variable time amplitude amplification (VTAA) \cite{ambainis2012variable} instead of standard amplitude amplification. However even with recent improvements \cite{chakraborty2023quantum}, this procedure requires an additional $O(\log \kappa)$ ancilla qubits (overall $O(a_H+\log \kappa)$ ancilla qubits are needed), and even more multi-qubit controlled operations to be implemented. Hence this is significantly out of reach for early fault tolerant quantum computers. We do not compare our method with quantum linear systems algorithms making use to VTAA.

Instead we consider the state-of-the-art quantum linear systems algorithm by Costa et. al. \cite{costa2022optimal}, which assumes an $(\alpha_H, a_H, 0)$-block encoding of $H$, and implements a block encoding of the some interpolated Hamiltonian $H(s)$, similar to other adiabatic approaches for this problem \cite{subacsi2019quantum, lin2020optimal}. It then proceeds to construct an interpolated quantum walk operator out of the block encoding, and carries out a fine grained analysis of the spectrum of this operator for discrete time steps, in accordance with the discrete adiabatic theorem, followed by eigenstate filtering. Their method prepares $\ket{x}$ using $O(\alpha_H\kappa\log(1/\varepsilon))$ queries to the block encoding of $H$. Thus, this algorithm uses more ancilla qubits and also needs several multi-qubit controlled gates. However, it achieves a $\log \kappa$ improvement in the error dependence over LCU and QSVT based approaches, and at the same time has a linear dependence on $\kappa$, without requiring the complicated VTAA procedure. Additionally, this technique requires four extra ancilla qubits (in addition to $a_H$ needed for implementing the block encoding of $H$). The query complexity per run is thus slightly lower than the maximal time evolution of $H$ needed by our method in terms of the precision (by a factor of $\log(\kappa)$), but the normalization of the block encoding $\alpha_H$ is multiplied. Moreover, this method would need $O(~||O||^2/\varepsilon^2)$ classical repetitions to estimate $\braket{x|O|x}$ incoherently. Using quantum amplitude estimation, this would require $\widetilde{O}(\alpha_H\alpha_O\kappa\norm{O}/\varepsilon)$ queries per coherent run, which is also the total query complexity. This is lower than the total time evolution of our algorithm, but the maximal time evolution of $H$ per coherent run is exponentially higher (in terms of $1/\varepsilon$). 

Huang et al. \cite{huang2021near} analyze the possibility of solving quantum linear systems using near-term quantum devices. In particular, their algorithms can be adapted to estimate expectation values of observables in the Hamiltonian evolution model. If $x$ is the (unnormalized) state that minimizes the loss function $\norm{x-H^{-1}\ket{b}}^2$, the authors devise a strategy to write down the  description of $x$ as a linear combination
$$
x=\sum_{j=1}^{m} \alpha_j \ket{\psi_{U_j}},
$$
where $\alpha_j\in \mathbb{C}$ and $\ket{\psi_{U_j}}$ is any \textit{good quantum state} such that $H^{-1}\ket{b}\in \mathrm{Span}\{\ket{\psi_{U_j}}\}$ for $j\in [1,m]$. If $H=\sum_{k=1}^{L}\lambda_k U_k$ is a linear combination of unitaries $U_k$, then given access to a set of such \textit{good quantum states}, the authors first estimate the optimal $\alpha_i$'s that minimize the aforementioned loss function (up to $\varepsilon$) accuracy, through a hybrid quantum classical algorithm: a combination of the Hadamard test and quadratic optimization. The quantum part of the algorithm requires $O(L^2m^3/\varepsilon)$ classical repetitions, and the cost of each such run depends on the cost of implementing the quantum circuit that prepares $\ket{\psi_{U_i}}$. Note that the Hadamard test requires a single ancilla qubit. Finally, the outcomes of the quantum procedure serve as inputs to the classical optimization problem that finds the optimal $\alpha_i$'s. Then having obtained such an $x$ they re-use the Hadamard test (a slightly modified version, to be precise) to estimate $\braket{\psi_{U_j}|O|\psi_{U_i}}$ for each pair $i,j \in [1, m]$, and then add up these quantities classically, weighted by the optimal $\alpha_i$'s to estimate $x^{\dag}O x$. The authors show that one way to obtain a \textit{good quantum state} is by applying the time evolution operator to $\ket{b}$, i.e.\ each $\ket{\psi_{U_j}}=e^{-it_jH}\ket{b}$. Here $t_j$ is chosen as per the quantum linear systems algorithm of Ref.~\cite{childs2017quantum}, and so any such $t_j\leq O(\kappa\log(\kappa/\varepsilon))$. Thus, the procedure requires using a single ancilla qubit and a maximal Hamiltonian evolution time $\tau_{\max}=\widetilde{O}(\kappa)$ in each coherent run, which matches our method. However, the number of classical repetitions is exponentially higher than our procedure. First, the estimation of the $\alpha_i$'s require $O(L^2m^3/\varepsilon)$ classical repetitions, such that the maximal time evolution of $H$ per run is $\widetilde{O}(\kappa)$. Note that $m$ can be $O(N)$ in general for a $N\times N$ Hamiltonian $H$. Furthermore, for each $i,j\in [1,m]$, the quantity $\braket{\psi_{U_j}|O|\psi_{U_k}}$ needs to be estimated to $O(\varepsilon/m^2)$ accuracy. So the total number of classical repetitions is $O(m^4/\varepsilon^2)$. Thus, overall, the total evolution time is exponentially higher than our method. There are other issues: $x$ is not normalized and so accounting for the normalization factor to ultimately estimate $\braket{x|O|x}$ adds to the overall cost. Moreover this estimation scheme does not work for general observables. The authors consider observables which can be diagonalized as $O=UDU^{\dag}$, such that $U$ is efficiently implementable and the entries of the diagonal matrix $D$ are also efficiently computable. Finally, there are other (variational) strategies in this work to obtain $x$, which are largely heuristic.

The quantum algorithm of Zhang et al. \cite{zhang2022computing} also estimates $\braket{x|O|x}$. It assumes the availability of a quantum state $\ket{\phi_0}$ with a constant overlap with $\ket{x}$, requiring a maximal evolution time of $\widetilde{O}(\kappa)$. This can be achieved by using the adiabatic-based framework of \cite{subacsi2019quantum}. Moreover, they also assume a $(\alpha_O, a_O, 0)$-block encoding of $O$ which increases the ancilla space. Having prepared such a $\ket{\phi}$, any ground state property estimation algorithm (including \textit{Single-Ancilla LCU}) can be employed to estimate the desired expectation value: the problem reduces to a particular case where both $\norm{c}_1$ and $\eta$ are constant. Their procedure requires a maximal time evolution of $\widetilde{O}(\kappa)$ per coherent run and $O(\alpha^2_O/\varepsilon^2)$ repetitions. Similar assumptions of access to such a state preparation procedure, would also result in our method estimating the desired expectation value with a matching maximal time evolution per coherent run, as well as the same total time evolution, without requiring a block encoding access to $O$, thereby requiring fewer ancilla qubits.

Finally, in the Appendix (Sec.~\ref{subsec-app:gsp-qls}), we have provided explicit gate depths of our algorithm while implementing the time evolution operator using (a) Hamiltonain simulation by \textit{Single-Ancilla LCU} and (b) $2k$-order Trotter (See Table \ref{table:circuit-depth-single-ancilla-gsp-qls}). Overall, our algorithm still requires only a single ancilla qubit and no multi-qubit controlled gates. Despite this, while comparing our method with other algorithms (See Table \ref{table:comparison-qls-depth}), we find that there are regimes where our method requires a shorter gate depth per coherent run as compared to even state-of-the-art methods.

\section{Applications to quantum walks}
\label{sec:quantum-random-walks}

So far, we have seen applications of the  \textit{Single-Ancilla LCU} and the \textit{Analog LCU} approaches. In this section, we will show that the \textit{Ancilla-free LCU} can be applied to the framework of quantum walks.  Recall from Sec.~\ref{subsec:ancilla-free-LCU}, that this approach is useful when we are interested in the average projection of the LCU state $f(H)\ket{\psi_0}$ in some subspace of interest. In such scenarios, it suffices to sample the unitaries $U_j$ according to the distribution of the LCU coefficients. This is because the projection of resulting density matrix on to this subspace is at least as large on average (See Theorem \ref{thm:ancilla-free-LCU-theorem}). We will show this is precisely the case for spatial search by quantum walks, where we are interested in finding out the expected number of steps after which the projection of the state of the quantum walk is high on some subset of the nodes (marked nodes) of the underlying Markov chain. Using \textit{Ancilla-free LCU} allows us to retain the same quadratic speedup as in prior works while requiring no ancilla qubits (other than the quantum walk space). 

We provide two quantum algorithms for spatial search by discrete-time quantum walks. The first one adapts the recent optimal algorithm of \cite{ambainis2019quadratic}, wherein the authors used quantum fast-forwarding via \textit{Standard LCU} \cite{apers2018FF} to obtain a generic quadratic speedup (up to a log factor). We show that by using \textit{Ancilla-free LCU} instead, we can obtain the same quadratic speedup while requiring  fewer ancilla qubits. This formalizes the observation of Ref.~\cite{apers2019unified} - where the authors stated that the LCU of \cite{ambainis2019quadratic} could indeed be bypassed. Our second quantum algorithm relies on fast-forwarding continuous-time random walks, which also fits nicely in the \textit{Ancilla-free LCU} framework. This algorithm too, achieves the same generic speedup, using fewer ancilla qubits as compared to \cite{ambainis2019quadratic}. For completeness, we would like to mention that the recent optimal spatial search algorithm by continuous-time quantum walk \cite{apers2022quadratic}, can also be seen as an exposition of a continuous-time variant of \textit{Ancilla-free LCU}. 

Similar to the previous sections, here too, we shall present our results based on generic Hamiltonians. We will refer to quantum walks only as a particular case of our general results, which we believe are more broadly applicable. We begin with a very brief review of random and quantum walks.
\\~\\
\subsection{Random and quantum walks: A very brief overview} 
\label{subsec:rw-qw}
Consider any ergodic, reversible Markov chain $P$ defined on a vertex space $X$ with $|X|=n$ nodes. One can think of such chains as a weighted graph of $n$ nodes (For detailed definitions of these terms, refer to the Appendix of \cite{krovi2016quantum}). Then $P$ is an $n\times n$ stochastic matrix. Let $p_{x,y}$ be the $(x,y)$ - th entry of $P$. We shall consider that the singular values of $P$ lie in $[0,1]$. This is without loss of generality: one can implement the transformation $P\mapsto (I+P)/2$ to always ensure this.

Then starting from any initial probability distribution over $X$, represented by the row vector $v_0$, $t$-steps of a classical random walk, results in distribution $v_t=v_0 P^t$ over $X$. For any such $P$ there exists a stationary distribution $\pi=(\pi_1, \pi_2, \cdots, \pi_n)$ such that $\pi=\pi P$. From any $P$ one obtains a continuous-time random walk by using $Q=I-P$ (under fairly general conditions). A continuous-time random walk, starting from $v_0$, evolves to $v_t=v_0 e^{Qt}$.

Since $P$ is not symmetric in general, it would be useful to work with the Discriminant matrix $D$ of $P$. $D$ is an $n\times n$ symmetric matrix such that its $(x,y)^{\mathrm{th}}$ - entry is $\sqrt{p_{xy}p_{yx}}$. The singular values of $P$ are the same as the eigenvalues of $D$. Moreover, the state $\ket{\sqrt{\pi}}=\sum_{x\in X}\sqrt{\pi_x}\ket{x}$ is the eigenstate of $D$ with eigenvalue $1$. 

In order to define a discrete-time quantum walk, define the unitary $U_P$ such that
$$
U_P\ket{\bar{0}}\ket{x}=\sum_{y=1}^{n}\sqrt{p_{xy}}\ket{y,x}.
$$  
where $\ket{\bar{0}}$ is some reference state. Let $S$ be the swap operation such that $S\ket{x,y}=\ket{y,x}$, and $\Pi_0=\ket{\bar{0}}\bra{\bar{0}}\otimes I$. Then the unitary defined by
\begin{equation}
\label{eq:dtqw-unitary}
V_P =[(2\Pi_0-I)\otimes I] U_P^{\dag}S U_P ,
\end{equation}
is a discrete-time quantum walk on the edges of $P$. For details on these discrete-time quantum walks, we refer the reader to Refs.~\cite{szegedy2004quantum,krovi2016quantum,ambainis2019quadratic, apers2019unified}. We now describe the spatial search problem, which we shall deal with in the subsequent sections.

Suppose a subset $M$ of the $n$ nodes of $P$ are marked. That is, its state space $X=U \cup M$, where $U$ is the set of unmarked nodes. Then, the spatial search problem can be defined as follows: suppose the random walk starts from the stationary distribution $\pi$ of $P$. What is the expected number of steps needed by the random walk to find some node $v\in M$? For random walks (both discrete and continuous-time), this is known as the hitting time ($HT$). Whether quantum walks can provide a quadratic advantage for the spatial search problem for any $P$ and any number of marked nodes, was open until recently. Ambainis et al. proved that discrete-time quantum walks solve the problem in $\tilde{O}(\sqrt{HT})$ steps, on average \cite{ambainis2019quadratic}. A similar result was shown for continuous-time quantum walks in \cite{apers2022quadratic}.

Both these quantum algorithms make use of the so-called interpolated Markov chains framework. Let $P'$ be the absorbing Markov chain, obtained from $P$ by replacing all outgoing edges from $M$ with self-loops. Then, the interpolated Markov chain is defined as $P(s)=(1-s)P +sP'$, where $s\in [0,1)$. One can define a Discriminant matrix $D(s)$ for $P(s)$, analogous to $P$. The relationship between $D(s)$ and $P(s)$ is also analogous to the non-interpolated case. In addition to interpolated Markov chains, the optimal quantum spatial search algorithms in \cite{ambainis2019quadratic} made use of  \textit{Standard LCU}-based techniques. Here, we will show that the framework of \textit{Ancilla-free LCU} quite naturally leads to new optimal quantum algorithms for spatial search while saving on the number of ancilla qubits needed. As mentioned previously, we will work with general Hermitian operators (Hamiltonians) and only invoke quantum (or random) walks as particular cases. 
\subsection{Applying Ancilla-free LCU: Optimal quantum spatial search by fast-forwarding discrete-time random walks}
\label{sec:ff-dtrw}

We begin by considering any Hamiltonian $H$ of unit spectral norm. 
Consider $U_H$, which is a $(1,a,0)$-block encoding of $H$, such that $U_H^2=I$. Then, it is well known that we can implement a block encoding of $H^t$ using LCU in cost scaling as $O(\sqrt{t})$. This has been implicit in Ref.~\cite{childs2017quantum}, and also appeared in Ref.~\cite{ambainis2019quadratic}. For this, generally one uses the fact that for any $x$ such that $|x|\leq 1$, $x^t$ can be expressed as a linear combination of Chebyshev polynomials. The following $d$-degree polynomial $p_{t,d}(x)$, defined as 
\begin{equation}
\label{eq:polynomial-approximating-power}
p_{t,d}(x)=
\begin{cases}
\dfrac{1}{2^t}\sum_{j=-d/2}^{d/2} \binom{t}{j+t/2} T_{2j}(x) & ~t,d \text{~are even}\\~\\~\\
\dfrac{2}{2^t}\sum_{j=0}^{(d-1)/2}\binom{t}{\frac{t+1}{2}+j}T_{2j+1}(x) & ~ t,d \text{~are odd.}
\end{cases},
\end{equation}
approximates $f(x)=x^t$, for any $t\in \mathbb{Z}$. This has been formally proven in Ref.~\cite{sachdeva2014faster} which we restate here:
~\\
\begin{lemma}{\cite{sachdeva2014faster}}
\label{lem:polynomial-approximating-power-monomial}
Suppose $\varepsilon>0$, $x\in [-1,1]$, $q\geq 1$ and $t\in\mathbb{R^+}$, then there exists a polynomial $p_{t,d}(x)$ of degree $d=\lceil\sqrt{2t\ln(2q/\varepsilon)}\rceil$ such that, 
$$
\sup_{x\in [-1,1]}\left|x^t-p_{t,d}(x)\right|\leq 2 e^{-d^2/2t}\leq \varepsilon/q.
$$
\end{lemma}

If $R=(2\ket{\bar{0}}\bra{\bar{0}}-I)\otimes I= 2\Pi_0-I\otimes I$, is the reflection operator. Then $V^t$ is a $(1,a,0)$-block encoding of $T_{t}(H)$, where $V=R.U_H$. That is, $\left(\bra{\bar{0}}\otimes I\right) V^t \left(\ket{\bar{0}}\otimes I\right)= T_{t}(H)$. Then $H^t$ can be approximated by a linear combination of powers of $V$. In fact,
\begin{align}
\label{eq:diff-power-of-H-LCU}
\norm{H^t-\sum_{\ell=0}^{d/2}\dfrac{c_{\ell}}{\norm{c}_1} V^{\ell}}\leq \varepsilon,
\end{align}
for $d=\lceil\sqrt{2t\ln(8/\varepsilon)}\rceil$. Here, for even $t$, the LCU coefficients are defined as
\begin{equation}
\label{eq:coefficients-lcu-even-t}
c_{\ell}=\begin{cases}
2^{1-t}\binom{t}{\frac{t}{2}+\ell}, & \ell>0\\
2^{-t}\binom{t}{t/2}, & \ell=0,
\end{cases},
\end{equation}
while for odd $t$,
\begin{equation}
\label{eq:coefficients-lcu-odd-t}
c_{\ell}=2^{1-t}\binom{t}{\frac{t+1}{2}+\ell},
\end{equation}
The $\ell_1$-norm of the LCU coefficients can be easily obtained by observing for $x\in [-1,1]$
\begin{align}
\norm{c}_1&=\left|x^t-\sum_{\ell=d/2+1}^t 2^{1-t}\binom{t}{\frac{t}{2}+\ell}\right|\\
		  &\geq 1-\left|\sum_{\ell=d/2+1}^t 2^{1-t}\binom{t}{\frac{t}{2}+\ell}\right|\\
		  &\geq 1-\varepsilon/4,
\end{align}
for even $t$, while an analogous bound can also be obtained for odd $t$. Now, implementing the linear combination of powers of $V$ in Eq.~\eqref{eq:diff-power-of-H-LCU} via \textit{Standard LCU} requires $O(a+\log t +\log\log(1/\varepsilon))$ ancilla qubits, and $O(\sqrt{t\ln(1/\varepsilon)})$ cost, which has been used in quantum fast-forwarding \cite{apers2018FF, ambainis2019quadratic}. 

Now suppose we are concerned about the average projection of $H^t$ in some subspace of interest. Then by \textit{Ancilla-free LCU}, we do not need the additional $O(\log t + \log\log(1/\varepsilon))$ ancilla qubits (See Theorem \ref{thm:ancilla-free-LCU-theorem}): we can simply sample some $\ell$ according to $c_{\ell}/\norm{c}_1$ and apply $V^\ell$ to the initial state. On average, the projection of the resulting density matrix would be at least as large as the projection of $H^t$. This is what we show next.
~\\~\\
\textbf{Fast-forwarding by \textit{Ancilla-free LCU}:~} Given access to $U_H$, which is a $(1,a,\delta)$-block encoding of some Hamiltonian $H$, we want to prepare a state whose projection on to the subspace of interest (spanned by $\Pi$) is at least $\norm{\Pi H^t\ket{\psi_0}}^2$, on average. In such a scenario,  
we can implement \textit{Ancilla-free LCU} by replacing \textit{Standard LCU} with importance sampling. Consider Algorithm \ref{algo:pow-ham}, where $V=R.U_H$. If the initial state is $\rho_0=\ket{\psi_0}\bra{\psi_0}$, Algorithm \ref{algo:pow-ham} outputs a density matrix, which, on average, is
\begin{equation}
\label{eq:avg-dm}
\bar{\rho}=\sum_{\ell=0}^{d/2} \dfrac{c_\ell}{\norm{c}_1} V^{2\ell} \rho_0 V^{-2\ell},
\end{equation}
if $t$ is even (an analogous expression is obtained when $t$ is odd).
\RestyleAlgo{boxruled}
\begin{algorithm}[ht]
  \caption{$\texttt{POW-HAM}(t,d, V,\ket{\psi_0})$}\label{algo:pow-ham}
   \textbf{Inputs:} The unitary $V$, an initial state $\ket{\psi_0}$ and parameters $t\in \mathbb{R}^{+}$ and $d\in \mathbb{N}$.
   \\~\\
   \begin{itemize}
  \item[1.~] If $t$ is even, 
  \subitem(a) Pick $\ell \in [0,d/2]$ according to $c_{\ell}/\norm{c}_1$, where $c_\ell=2^{1-t}\binom{t}{\frac{t}{2}+\ell}$.\\ 
  \subitem(b) Apply $V^{2\ell}$ to $\ket{\psi_0}$.\\~\\
  \item[2.~] If $t$ is odd, 
 \subitem(a) Pick $\ell \in \left[0,\frac{d-1}{2}\right]$ according to $c_{\ell}/\norm{c}_1$, where $c_\ell=2^{1-t}\binom{t}{\frac{t+1}{2}+\ell}$. 
 \subitem(b) Apply $V^{2\ell+1}$ to $\ket{\psi_0}$. \\~\\
  \end{itemize}
\end{algorithm}
Then, we can use Theorem \ref{thm:ancilla-free-LCU-theorem} to prove that $\Tr[\Pi\bar{\rho}]\geq \Tr[\Pi H^t \rho_0 H^t]-\varepsilon$. However, there are still issues to consider before we can do so. For instance, $U_H$ is not a perfect block encoding of $H$. How should the precision in block encoding, $\delta$, scale for this to hold? We formally state this, and prove the algorithmic correctness via the following Theorem:
\begin{lemma}
\label{lem:ancilla-free-dtqw-1}
Suppose $\varepsilon\in (0,1)$ and we have access to $U_H$, which is a $(1,a,\delta)$-block encoding of a Hamiltonian $H$ such that $\norm{H}=1$ and $U^2_H=I$. Then, provided $d=\lceil\sqrt{2t\ln(24/\varepsilon)}\rceil$ and,
$$
\delta\leq \dfrac{\varepsilon^2}{1152~t\ln(24/\varepsilon)},
$$
for any $t\in\mathbb{R}^+$, initial state $\rho_0=\ket{\psi_0}\bra{\psi_0}$, and projector $\Pi$ on to the space of $\rho_0$, then the sampling in Algorithm \ref{algo:pow-ham} prepares some density matrix $\rho$ such that $\mathbb{E}[\rho]=\bar{\rho}$, where $\bar{\rho}$ is defined in Eq.~\eqref{eq:avg-dm}, and satisfies
$$
\Tr[\Pi\bar{\rho}] \geq \Tr[\Pi H^t\rho_0 H^t]-\varepsilon,
$$
using $O\left(\sqrt{t\log(1/\varepsilon)}\right)$ queries to $V=R.U_H$.
\end{lemma}
\begin{proof}
Let $H'$ be a $(1,a,0)$-block encoding of $U_H$. Then, by definition $\norm{H-H'}\leq \delta$. Let us choose the degree of the polynomial $p_{t,d}(H')$ to be $d=\lceil\sqrt{2t\ln(24/\varepsilon)}\rceil$, which ensures that $\norm{x^t-p_{t,d}(x)}\leq \varepsilon/12$ (from Lemma \ref{lem:polynomial-approximating-power-monomial}).

Now, the \textit{Standard LCU} procedure would implement the state
$$
\ket{\psi_t}=\ket{\bar{0}}\dfrac{p_{t,d}(H')}{\norm{c}_1}\ket{\psi_0}+\ket{\Phi}^{\perp}. 
$$
From the choice of $d$, we ensure that $\alpha=\norm{c}_1 \geq 1-\varepsilon/12$.
Also,
\begin{align}
\norm{H^t-p_{t,d}(H')/\alpha}&\leq \norm{H^t-p_{t,d}(H')}+(1-\alpha)\norm{p_{t,d}(H')/\alpha}\\
                              &\leq \varepsilon/12+\norm{H^t-p_{t,d}(H)}+\norm{p_{t,d}(H)-p_{t,d}(H')}\\
                              &\leq \varepsilon/12+\varepsilon/12+ 4d\sqrt{\delta}~~~~~~~~~~~~~~~~~~~~~~~~~~~~~~~~~~~~~~~~~~~~~[\text{~From Lemma~\ref{lem:robustness_of_QSVT}~}]\\
                              &\leq \varepsilon/6+\varepsilon/6~~~~~~~~~~~~~~~~~~~~~~~~~~~~~~~~~~~~~~~~~~~~~~~~~~~~~~~\left[\text{~As $\delta\leq \dfrac{\varepsilon^2}{576 d^2} $~}\right]\\
                              &\leq \varepsilon/3.
\end{align}
We will now use Theorem \ref{thm:ancilla-free-LCU-theorem}, which ensures that the average density matrix $\bar{\rho}$ from Algorithm \ref{algo:pow-ham}, satisfies
\begin{align}
 \Tr[(I\otimes\Pi)\bar{\rho}]&=\Tr[(I\otimes\Pi)\ket{\psi_t}\bra{\psi_t}]\\
 &\geq \dfrac{1}{\norm{c}^2_1}\left[\Tr[\Pi H^t\rho_0 H^t]-\varepsilon\right]\\
			   &\geq \Tr[\Pi H^t\rho_0 H^t]-\varepsilon.
\end{align} 
\end{proof}

Thus, on average, the projection of the density matrix prepared by Algorithm \ref{algo:pow-ham} on to the space spanned by $\Pi$ is at least as large as $\Tr[\Pi H^t \rho_0 H^t]$. Furthermore, Lemma \ref{lem:ancilla-free-dtqw-1} shows that the cost is $O(\sqrt{t\ln(1/\varepsilon)})$. Interestingly this procedure does not need the extra $O(\log t + \log\log(1/\varepsilon))$ ancilla qubits that \textit{Standard LCU} does.

Let us now discuss why this is important in the context of spatial search by quantum walk. Notice that for any ergodic, reversible Markov chain, the discrete-time quantum walk unitary $U_D = U_P^{\dag}S U_P$ (defined in Sec.~\ref{subsec:rw-qw}) is a block encoding of the discriminant matrix $D$ of $P$, i.e.\
$$
\left(\bra{\bar{0}}\otimes I\right) U^{\dag}_P S U_P \left(\ket{\bar{0}}\otimes I\right) = D.
$$
Thus $U_D$ is a $(1,\lceil \log_2 n \rceil,0)$-block encoding of $D$, where $a=|X|=n$. Thus, \textit{Standard LCU} prepares a quantum state that is $O(\varepsilon)$ - close to
$$
\ket{\psi_t}= \ket{\bar{0}} D^t\ket{\psi_0}+\ket{\Phi}^{\perp},
$$
in cost $O\left(\sqrt{t\log(1/\varepsilon)}\right)$, using $O(\log n+\log t +\log\log (1/\varepsilon))$ ancilla qubits. This is the essence of quantum fast-forwarding: the fact that we can implement $t$-steps of a random walk on $D$ quadratically faster using quantum walks (up to normalization) \cite{apers2018FF}. 

On the other hand, suppose for a specific problem (such as quantum spatial search) it suffices to ensure that the projection of $D^{t}\ket{\psi_0}$ on some subspace of interest is high on average. Then from Algorithm \ref{algo:pow-ham} we can prepare a quantum state whose projection, on average, is at least as high as the projection of $D^{t}\ket{\psi_0}$. For this, we can simply apply $k$ steps of the quantum walk operator $V=R.U_D$ for some $k\in [0,\sqrt{2t\ln(24/\varepsilon)}]$ steps, sampled at random according to $\{c_k/\norm{c}_1\}$. Indeed Algorithm \ref{algo:pow-ham}, on average, prepares the density matrix $\bar{\rho}$ such that from Lemma \ref{lem:ancilla-free-dtqw-1},
$$
\Tr[\Pi\bar{\rho}]\geq \Tr[\Pi D^t\rho_0 D^t]-\varepsilon.
$$
Overall, quantum fast forwarding by \textit{Ancilla-free LCU} requires only $O(\log n)$ ancilla qubits, which is the quantum walk space (edges of $P$). On the other hand, quantum fast-forwarding by \textit{Standard LCU} requires an additional $O(\log t + \log\log(1/\varepsilon))$ ancilla qubits. 

Now we are in a position to discuss the spatial search algorithm by discrete-time quantum walk which makes use of Algorithm \ref{algo:pow-ham}. Now for the spatial search problem, we are interested in showing that on average, the projection of the state $D^t\ket{\sqrt{\pi}}$ in the marked subspace is high, for $t=O(HT)$. As a result, we can drop the ancilla register and simply apply the \textit{Ancilla-free LCU} technique to implement the unitary $V$ for a random number of steps, sampled according to the distribution of the LCU coefficients. 
~\\~\\
\textbf{Method 1 -- Spatial search by discrete-time quantum walk:~}  Suppose there exists an ergodic, reversible Markov chain $P$ with state space $X$ and $|X|=n$. Let $M\subset X$ be a set of marked nodes. Classically, the spatial search algorithm boils down to applying the so-called, absorbing Markov chain $P'$ (obtained from $P$ by replacing the outgoing edges from $M$ with self-loops) to the stationary distribution of $P$ (say $\pi$). At every step, one checks to see if the vertex obtained is marked. The expected number of steps needed to find some $x\in M$ is known as the hitting time, denoted by $HT$. Consider the interpolated Markov chain $P(s)=(1-s)P+sP'$, where $s\in [0,1]$ and $D(P(s))$ be the corresponding discriminant matrix. Then $U_{D(s)}$ is a $(1,a,0)$-block encoding of $D(s)$, where $a=\lceil\log_2(n)\rceil$. 

First, let us look at the quantum spatial search algorithm of Ref.~\cite{ambainis2019quadratic}. Starting from the state $\ket{\sqrt{\pi}}=\sum_{x}\sqrt{\pi_x}\ket{x}$ (which corresponds to a quantum encoding of the stationary distribution of $P$), the first step involves preparing $\ket{\bar{0}}\ket{\sqrt{\pi_U}}=\ket{\bar{0}}\sum_{y\in X\backslash M} \sqrt{\pi_y}\ket{y}$, which is simply the support of $\ket{\sqrt{\pi}}$ in the unmarked subspace. This is quite easy: if $\Pi_M$ is the projector on to the marked subspace, simply measure $\ket{\sqrt{\pi}}$ in the basis $\{\Pi_M, I-\Pi_M\}$. Then $\ket{\sqrt{\pi_U}}$ is prepared if the outcome of the measurement is an unmarked vertex. Following this, the algorithm of \cite{ambainis2019quadratic} prepares the state $D(s)^t\ket{\sqrt{\pi_U}}$ using quantum fast-forwarding \cite{apers2018FF} (via \textit{Standard LCU}) for some randomly chosen values of $s$, and $t\in \Theta(HT)$. Then the central result of Ref.~\cite{ambainis2019quadratic} was to prove, using an involved combinatorial lemma, that the projection of this state on to the marked subspace is $\widetilde{\Omega}(1)$ on average for these choices of $s$ and $t$. More precisely, the authors prove that if we choose parameters $s\in \{1-1/r: r=1,2,\cdots, 2^{\lceil\log T\rceil}\}$ and $T\in \Theta\left(HT\right)$ uniformly at random,
\begin{equation}
\label{eq:success-probability-dtqw}
\mathbb{E}\left[\norm{\Pi_M D(s)^T \ket{\sqrt{\pi_U}}}^2\right]\in \Omega\left(1/\log T\right). \nonumber
\end{equation}
This is optimal as it provides a generic quadratic speedup over classical random walks (up to a log factor) for any reversible $P$ and marked subspace $M$ of any cardinality, unlike prior works. However, the algorithm uses quantum fast forwarding by \textit{Standard LCU}, which requires $O(\log HT)$ ancilla qubits to implement.

Note that in this case, it suffices to ensure that the projection of $D(s)^t\ket{\sqrt{\pi_U}}$ in the marked subspace is large on average for the aforementioned choices of $s$ and $t$. Then from Lemma \ref{lem:ancilla-free-dtqw-1}, we can replace the quantum fast forwarding by \textit{Standard LCU} of Ref.~\cite{ambainis2019quadratic}, with \textit{Ancilla-free LCU}: Apply $k$ steps of the quantum walk operator $V(s)=R.U_{D(s)}$, where $k\in [0, O(\sqrt{HT})]$ is chosen according to $\{c_k/\norm{c}_1\}$, as stated in Algorithm \ref{algo:pow-ham}. This prepares, on average, the density matrix $\bar{\rho}$ whose projection on to the marked subspace is at least $\Omega\left(1/\log T\right)$. 

The quantum spatial search algorithm, after incorporating Algorithm \ref{algo:pow-ham} as a subroutine, is stated via Algorithm \ref{algo:search-dtqw-1}.
\RestyleAlgo{boxruled}
\begin{algorithm}[ht]
  \caption{$\texttt{QSpatial Search - 1}$~ Spatial search by DTQW}\label{algo:search-dtqw-1}
  ~\\
   \begin{itemize}
  \item[1.~] Pick $t$ uniformly at random from $[0,T]$, where $T=\Theta\left(HT\right)$ and set \\ $d=\lceil\sqrt{T\log(T)}\rceil$.

  \item[2.~] Set $s=1-1/r$, where $r$ is picked uniformly at random from\\ $R=\{2^0, 2^1,\cdots, 2^{\lceil\log T\rceil}\}$.
  
  \item[3.~] Construct $U_{D(s)}$, which is a $(1,\lceil\log_2(|X|)\rceil,0)$-block encoding of $D(s)$.
  
  \item[4.~] Prepare the quantum state $\ket{\bar{0}}\ket{\sqrt{\pi}}=\ket{\bar{0}}\sum_{y\in X}\sqrt{\pi_y}\ket{y}$.
  
  \item[5.~] Measure in the basis $\{\Pi_M,I-\Pi_M\}$ in the second register. If the output is marked, measure in the node \\ basis to output some $x\in M$. Otherwise, we are in the state $\ket{\bar{0}}\ket{\sqrt{\pi_U}}=\ket{\bar{0}}\sum_{y\in X\backslash M}\sqrt{\pi_y}\ket{y}$.
  
  \item[6.~] Call $\texttt{POW-HAM}(t,d,V=R.U_{D(s)},\ket{\sqrt{\pi_U}})$.\\
  
  \item[7.~] Measure the first register in the basis of the nodes of $P$.
  \end{itemize}
\end{algorithm}
The key difference as compared to the algorithm of \cite{ambainis2019quadratic} is that our procedure does not require quantum fast-forwarding by \textit{Standard LCU}: it suffices to call Algorithm \ref{algo:pow-ham} instead. This ensures that our method requires $O(\log HT)$ fewer ancilla qubits. 

From Lemma \ref{lem:ancilla-free-dtqw-1}, we obtain that, Algorithm \ref{algo:search-dtqw-1} outputs a density matrix, which on average, has a projection on to the marked subspace, lower bounded as 
\begin{equation}
\label{eq:spatial-search-prob-lower-bound-1}
\Tr[(I\otimes \Pi_M)\bar{\rho}]\geq \norm{\Pi_M D(s)^T \ket{\sqrt{\pi_U}}}^2 - \varepsilon,
\end{equation}
for a small enough $\varepsilon \in \Theta(1/\log(T))$. Then, the combinatorial lemma of Ref.~\cite{ambainis2019quadratic} ensures that expected value of the RHS of the aforementioned equation is $\widetilde{\Omega}(1)$. We refer the readers to \cite{ambainis2019quadratic} for the proof of this lemma.

What Equation \eqref{eq:spatial-search-prob-lower-bound-1} tells us is that if we run Algorithm \ref{algo:search-dtqw-1}, and measure the second register in the vertex basis, the probability of finding a marked element would be at least $\Omega(1/\log^2 T)$, on average, after $O(\sqrt{T\log T})$ DTQW steps (where $T=\Theta(HT)$). Thus, just as Ref.~\cite{ambainis2019quadratic}, this algorithm also yields a quadratic improvement over its classical counterpart (up to a log factor). However, it requires $O(\log HT))$ fewer ancilla qubits.

We shall use similar ideas to develop an alternative quantum algorithm for spatial search by discrete-time quantum walk.

\subsection{Applying Ancilla-free LCU: Optimal quantum spatial search by fast-forwarding continuous-time random walks}
\label{sec:ff-ctrw}
We develop a quantum algorithm for fast-forwarding the dynamics of a continuous-time random walk. Given the discriminant matrix $D$, a continuous-time random walk is defined by the operator $e^{D-I}$, where $Q=D-I$ is the continuous-time random walk kernel. We first show that a block encoding of $e^{t(D-I)}$ can be implemented in cost $O(\sqrt{t}\log 1/\varepsilon)$, using \textit{Standard LCU}. This requires $O(\log n +\log t +\log\log(1/\varepsilon))$ ancilla qubits. Next, we demonstrate that the fast forwarding can be leveraged to develop an optimal quantum spatial search algorithm: requiring $\widetilde{O}(\sqrt{HT})$ steps on average to find marked vertices on any reversible $P$ with $M$ being the set of marked nodes. As this relies on the \textit{Standard LCU} technique, the overall algorithm requires $O(\log n + \log HT)$ ancilla qubits, analogous to \cite{ambainis2019quadratic}. 

Finally, similar to the previous section, for optimal quantum spatial search, it suffices to ensure that the projection of $e^{t(D-I)}\ket{\sqrt{\pi_U}}$ in the marked subspace is high, on average after $t=O(HT)$ steps of the continuous-time random walk. Thus, we can bypass quantum fast forwarding by \textit{Standard LCU}, and use \textit{Ancilla-free LCU}, instead. This only requires $O(\log n)$ ancilla qubits required to implement the quantum walk on the edges of $P$, saving on $O(\log HT)$ ancilla qubits overall. As in the previous section, we will work with general Hamiltonians and discuss quantum walks as a particular case \\~\\
Now the polynomial approximation to $x^t$ can be used to obtain the following low degree polynomial that approximates $e^{-t(1-x)}$ \cite{gilyen2018quantum}:
$$
q_{t,d,d'}(x)=e^{-t}\sum_{j=0}^{d}\dfrac{t^j}{j!}p_{j,d'}(x).
$$ 
This has degree
$$
d'= \lceil\sqrt{2d\ln(4/\varepsilon)}\rceil \in O\left(\sqrt{t}\log(1/\varepsilon)\right),
$$ 
for $d=\lceil\max\{te^2,\ln(2/\varepsilon)\}\rceil$. Indeed, it can be shown that
\begin{equation}
\label{eq:exponential-poly-approx}
\sup_{x\in[-1,1]}\abs{e^{-t(1-x)}-q_{t,d,d'}(x)}\leq \varepsilon.
\end{equation}
Thus, given a block encoding of any Hermitian matrix $H$, with unit spectral norm, the operator $e^{-t(I-H)}$ can be implemented as a linear combination of unitaries. This is because the $d'$-degree polynomial $q_{t,d,d'}(x)$ approximates $e^{-t(I-H)}$ and is a linear combination of the $d'$-degree polynomial $p_{j,d'}(x)$. So overall, by LCU, we can implement the polynomial $q_{t,d,d'}(x)$, approximating $e^{-t(1-x)}$. We formally show this via the following lemma:
\begin{lemma}
\label{lem:block encoding-exp-of-Ham}
Suppose $\varepsilon\in (0,1)$ and we have access to $U_H$, which is a $(1,a,\delta)$-block encoding of a Hamiltonian $H$ such that $\norm{H}=1$ and $U^2_H=I$. Furthermore, let $d=\lceil\max\{te^2,\ln(8/\varepsilon)\}\rceil$ and $d'=\sqrt{2d\ln(16/\varepsilon)}$. Then, provided 
$$
\delta\leq \dfrac{\varepsilon^2}{128d~\ln(16/\varepsilon)},
$$
for any $t\in\mathbb{N}$, we can implement a $(1,O(a+\log t+\log\log(1/\varepsilon)),\varepsilon)$-block encoding of $e^{-t(I-H)}$ in cost $O\left(\sqrt{t}\log(1/\varepsilon)\right)$.
\end{lemma}
\begin{proof}
Let $U_H$ be a $(1,a,0)$-block encoding of $H'$. By definition, $\norm{H-H'}\leq \delta$. For the polynomial $q_{t,d,d'}(x)$, we choose $d=\lceil\max\{te^2,\ln(8/\varepsilon)\}\rceil$ and $d'=\sqrt{2d\ln(16/\varepsilon)}$. This ensures $\norm{e^{-t(1-x)}-q_{t,d,d'}(x)}\leq \varepsilon/4$.

As before, let $\widetilde{W}$ be the unitary that implements the LCU. Then,
\begin{equation}
\label{eq:lcu-qpoly}
\left(\bra{\bar{0}}\otimes I\right)\widetilde{W}\left(\ket{\bar{0}}\otimes I\right)=\dfrac{q_{t,d,d'}(H')}{\norm{c}_1}.
\end{equation}
For the choice of $d,d'$ we have
\begin{align}
\norm{c}_1&=e^{-t}\sum_{j=0}^d \dfrac{t^j}{j!} p_{j,d'}(H)\\
		  &\geq e^{-t}\sum_{j=0}^d \dfrac{t^j}{j!} \left(1-\varepsilon/8\right)~~~~~~~~~~~~~~[~\text{As $d'=\lceil\sqrt{2d\ln(16/\varepsilon)}\rceil$}~]\\
		  &\geq \left(1-e^{-t}\sum_{j=d+1}^{\infty}\dfrac{t^{j}}{j!}\right)\left(1-\varepsilon/8 \right)\\
		  &\geq \left(1-\varepsilon/8 \right)\left(1-\varepsilon/8 \right)\\
		  &\geq 1-\varepsilon/4.
\end{align}
In order to go from the third to the fourth line we have used the fact that by Stirling's approximation:
\begin{align}
e^{-t}\sum_{j=d+1}^{\infty}\dfrac{t^{j}}{j!}&\leq e^{-t}\sum_{j=d+1}^{\infty} \left(\dfrac{te}{j}\right)^j\\
										   &\leq e^{-t}\sum_{j=d+1}^{\infty}e^{-j}\left(\dfrac{te^2}{j}\right)^j\\
										   &\leq e^{-t}\sum_{j=d+1}^{\infty} e^{-j}~~~~~~~~~~~~~[\text{~As~} te^2/j\leq 1, ~\forall j\geq d+1~]\\
										   & \leq e^{-t-d}\leq \varepsilon/8,								
\end{align}
for $d=\lceil\ln(8/\varepsilon)\rceil$.

This implies $\norm{c}_1\in [1-\varepsilon/4,1]$. Now, we will show that $\widetilde{W}$ indeed implements a block encoding of $e^{-t(I-H)}$.
\begin{align}
\norm{e^{-t(I-H)}-\dfrac{q_{t,d,d'}(H')}{\norm{c}_1}}&\leq \norm{e^{-t(I-H)}-q_{t,d,d}(H)} + \norm{e^{-t(I-H)}-q_{t,d,d}(H)} +\left(1-\norm{c}_1\right)\\
													&\leq \varepsilon/4+\varepsilon/4+4d'\sqrt{\delta}~~~~~~~~~~~~~~~~[\text{From Lemma \ref{lem:robustness_of_QSVT}}]\\
													&\leq \varepsilon/2+\varepsilon/2~~~~~~~~~~~~~~~~~~~~~~~~~\left[\text{~As~}\delta\leq \frac{\varepsilon^2}{64d'^2}\right]
\end{align}
\end{proof}
It is easy to see that this leads to the fast-forwarding of continuous-time random walks. The unitary $V=U_P^\dag S U_P$ is a $(1,\lceil\log n\rceil,0)$-block encoding of the random walk discriminant matrix $D$. Then by using Lemma \ref{lem:block encoding-exp-of-Ham}, given an initial state $\ket{\psi_0}$, we can prepare a quantum state that is $O\left(\varepsilon\cdot\norm{e^{-(I-D)t}\ket{\psi_0}}\right)$-close to
\begin{equation}
\label{eq:final-state-ff-ctrw}
\ket{\psi_t}=\ket{0} \dfrac{e^{-t(I-D)}\ket{\psi_0}}{\norm{e^{-t(I-D)}\ket{\psi_0}}}+\ket{\psi^\perp},
\end{equation}
with success probability $\Theta\left(\norm{e^{-(I-D)t}\ket{\psi_0}}^2\right)$, in cost $O\left(\sqrt{t}\log\left(\varepsilon^{-1}\cdot\norm{e^{-(I-D)t}\ket{\psi_0}}^{-1}\right)\right)$.

Finally, by applying $O\left(\norm{e^{-(I-D)t}\ket{\psi_0}}^{-1}\right)$ rounds of amplitude amplification, we prepare, with $\Omega(1)$ probability, a quantum state that is $O(\varepsilon)$-close to $\ket{\psi_t}$ in cost
$$
T=O\left(\dfrac{\sqrt{t}}{\norm{e^{-(I-D)t}\ket{\psi_0}}}\log\left(\dfrac{1}{\varepsilon\norm{e^{-(I-D)t}\ket{\psi_0}}}\right)\right).
$$
Other than the walk space of $O(\log n)$ qubits (edges of the Markov chain), fast forwarding continuous-time random walks by \textit{Standard LCU} additionally requires $O(\log t +\log\log(1/\varepsilon))$ ancilla qubits.  As before, for the spatial search algorithm, we can drop these ancilla qubits completely and implement \textit{Ancilla-free LCU} instead.
~\\~\\
\textbf{Fast-forwarding by \textit{Ancilla-free LCU}:}~ For the spatial search problem, we are concerned about the average projection of $e^{-t(I-H)}\ket{\psi_0}$, on to the marked subspace. Thus, we can apply \textit{Ancilla-free LCU} instead. The overall procedure is outlined in Algorithm \ref{algo:implement-qpoly}. We will be implementing $q_{t,d,d'}(H)$ which is itself a linear combination of $p_{t,d'}(H)$. So, Algorithm \ref{algo:implement-qpoly} also calls Algorithm \ref{algo:pow-ham} as a subroutine. 
\\~\\
\RestyleAlgo{boxruled}
\begin{algorithm}[ht]
  \caption{$\texttt{EXP-HAM}(t,d',d,V,\ket{\psi_0})$}\label{algo:implement-qpoly}
  \textbf{Inputs:~} A unitary $V$, $t\in \mathbb{R}^+$, $d,d'\in\mathbb{N}$, and an initial state $\ket{\psi_0}$.
   \begin{itemize}
  \item[1.~] Pick some integer $\ell \in [0,d]$ according to $c_\ell/\norm{c}_1$, where $c_\ell = \dfrac{e^{-t} t^\ell}{\ell!}$

  \item[2.~] Call $\texttt{POW-HAM}(\ell,d',V,\ket{\psi_0})$.
  \end{itemize}
\end{algorithm}

If $\rho_0=\ket{\psi_0}\bra{\psi_0}$, for $d, d'\in \mathbb{N}$, Algorithm \ref{algo:implement-qpoly}, on average, prepares the  density matrix
\begin{equation}
\bar{\rho}=e^{-t}\sum_{j=0}^{d} \dfrac{t^j}{j!}\left[\sum_{j\in\mathrm{Even}, k=0}^{d'/2} 2^{1-j}\binom{j}{j+k/2} V^{2k} \rho_0 V^{-2k}+\sum_{j\in\mathrm{Odd}, k=0}^{(d'-1)/2} 2^{1-j}\binom{j}{(j+1)/2+k} V^{2k+1} \rho_0 V^{-(2k+1)} \right],
\end{equation} 
where $V=R\cdot U_H$ is the quantum walk operator. On average $O(d')$ queries are made to $V$. However, in order to ensure that $\bar{\rho}$ satisfies
$$
\Tr[\Pi\bar{\rho}]\geq \Tr[\Pi e^{-t(I-H)}\rho_0 e^{-t(I-H)}]-\varepsilon,
$$ 
the appropriate values of $\delta, d, d'$ need to be chosen. We determine these via the following lemma:
~\\~\\
\begin{lemma}
\label{lem:ancilla-free-dtqw-2}
Suppose $\varepsilon\in (0,1)$ and we have access to $U_H$, which is a $(1,a,\delta)$-block encoding of a Hamiltonian $H$ such that $\norm{H}=1$ and $U^2_H=I$. Then, provided $d=\lceil\max\{te^2, \ln (12/\varepsilon)\}\rceil$, $d'=\lceil\sqrt{2t\ln(48/\varepsilon)}\rceil$ and,
$$
\delta\leq \dfrac{\varepsilon^2}{1152~d\ln(48/\varepsilon)},
$$
for any $t\in\mathbb{R}^+$, projector $\Pi$ and initial state $\rho_0=\ket{\psi_0}\bra{\psi_0}$, then Algorithm \ref{algo:implement-qpoly} prepares, on average, the density matrix $\bar{\rho}$ such that
$$
\Tr[\Pi\bar{\rho}] \geq \Tr[\Pi e^{-t(I-H)}\rho_0 e^{-t(I-H)}]-\varepsilon,
$$
using $O\left(\sqrt{t}\log(1/\varepsilon)\right)$ queries to $V=R.U_H$.
\end{lemma}
\begin{proof}
Let $H'$ be a $(1,a,0)$ block encoding of $U_H$. Then, by definition $\norm{H-H'}\leq \delta$. By choosing the degree of the polynomial $q_{t,d,d'}(H')$ to be $d'=\lceil\sqrt{2d\ln(48/\varepsilon)}\rceil$, where $d=\lceil\max\{te^2, \ln (12/\varepsilon)\}\rceil$, ensures that $\norm{e^{-t(1-x)}-q_{t,d,d'}(x)}\leq \varepsilon/12$. 

Now, from Lemma \ref{lem:block encoding-exp-of-Ham}, the full LCU procedure would implement the state
$$
\ket{\psi_t}=\ket{\bar{0}}\dfrac{q_{t,d,d'}(H')}{\norm{c}_1}\ket{\psi_0}+\ket{\Phi}^{\perp}. 
$$
Now, from the choice of $d'$, we ensure that $\norm{c}_1 \geq 1-\varepsilon/12$.
Also,
\begin{align}
\norm{e^{-t(I-H)}-q_{t,d,d'}(H')/\norm{c}_1}&\leq \norm{e^{-t(I-H)}-q_{t,d,d'}(H')}+(1-\norm{c}_1)\norm{q_{t,d,d'}(H')/\norm{c}_1}\\
                              &\leq \varepsilon/12+\norm{e^{-t(I-H)}-q_{t,d,d'}(H)}+\norm{q_{t,d,d'}(H)-q_{t,d,d'}(H')}\\
                              &\leq \varepsilon/12+\varepsilon/12+ 4d'\sqrt{\delta}~~~~~~~~~~~~~~~~~~~~~[\text{~From Lemma~\ref{lem:robustness_of_QSVT}~}]\\
                              &\leq \varepsilon/6+\varepsilon/6~~~~~~~~~~~~~~~~~~~~~~~~~~~~~~~~~~~\left[\text{~As $\delta\leq \dfrac{\varepsilon^2}{576 d'^2} $~}\right]\\
                              &\leq \varepsilon/3.
\end{align}
We will now use Theorem \ref{thm:ancilla-free-LCU-theorem}, which ensures that Algorithm \ref{algo:implement-qpoly}, on average, prepares $\bar{\rho}$ such that
\begin{align}
 \Tr[\Pi\bar{\rho}]=\Tr[(I\otimes \Pi)\ket{\psi_t}\bra{\psi_t}]&\geq \dfrac{1}{\norm{c}^2_1}\left[\Tr[\Pi e^{-t(I-H)}\rho_0 e^{-t(I-H)}]-\varepsilon\right]\\
			   &\geq \Tr[\Pi e^{-t(I-H)}\rho_0 e^{-t(I-H)}]-\varepsilon.
\end{align} 
\end{proof}
Thus, if it suffices to ensure that the projection on to the subspace of interest is high on average, we can replace \textit{Standard LCU} with \textit{Ancilla-free LCU}, thereby saving on the $O(\log t+\log\log(1/\varepsilon))$ ancilla qubits. Finally, we apply this lemma to develop a new quantum algorithm for spatial search by discrete-time quantum walk, one that relies on quantum fast forwarding of continuous-time random walks.
~\\~\\
\textbf{Method 2 -- Spatial search by discrete-time quantum walk:~} The ability to fast-forward continuous-time random walks imply that we can design an alternate quantum spatial search algorithm by discrete-time quantum walks. 

\RestyleAlgo{boxruled}
\begin{algorithm}[ht]
  \caption{$\texttt{QSpatial Search - 2}$~ Alternative Spatial search by DTQW}\label{algo:search-dtqw-2}
  ~\\
   \begin{itemize}
  \item[1.~] Pick $t$ uniformly at random from $[0,T]$, where $T=\Theta\left(HT\right)$. Set $d=\lceil T e^2 \rceil$ and $d'=\Theta(\sqrt{T\log T})$.

  \item[2.~] Set $s=1-1/r$, where $r$ is picked uniformly at random from\\ $R=\{2^0, 2^1,\cdots, 2^{\lceil\log T\rceil}\}$.
  
  \item[3.~] Construct $U_{D(s)}$, which is a $(1,\lceil\log_2(|X|)\rceil,0)$-block encoding of $D(s)$.
  
  \item[4.~] Prepare the quantum state $\ket{\bar{0}}\ket{\sqrt{\pi}}=\ket{\bar{0}}\sum_{y\in X}\sqrt{\pi_y}\ket{y}$.
  
  \item[5.~] Measure in the basis $\{\Pi_M,I-\Pi_M\}$ in the second register. If the output is marked, measure in the node basis to output some $x\in M$.\\ Otherwise, we are in the state $\ket{\bar{0}}\ket{\sqrt{\pi_U}}=\ket{\bar{0}}\sum_{y\in X\backslash M}\sqrt{\pi_y}\ket{y}$.
  
  \item[6.~] Call $\texttt{EXP-HAM}(t,d',d,V=R.U_{D(s)},\ket{\sqrt{\pi_U}})$.
  
 \item[7.~] Measure the resulting state in the node basis in the first register. 
  \end{itemize}
\end{algorithm}

Suppose we consider an interpolated Markov chain $P(s)=(1-s)P+sP'$, where $s\in [0,1]$, and  $P'$ corresponds to the absorbing Markov chain. Suppose $D(s)$ denotes the Discriminant matrix of $P(s)$ and we have access to $U_{D(s)}$, which is a $(1,a,0)$-block encoding of $D(s)$, such that $U^2_{D(s)}=I$. Consider Algorithm \ref{algo:search-dtqw-2}, which is very similar to the first spatial search algorithm except that it calls Algorithm \ref{algo:implement-qpoly} as a subroutine.

The output of Algorithm \ref{algo:search-dtqw-2}, on average, is the density matrix $\bar{\rho}$ such that
$$
\Tr[(I\otimes \Pi_M)\bar{\rho}]\geq \norm{\Pi_M e^{T(D(s)-I)}\ket{\sqrt{\pi_U}}}^2-\varepsilon,
$$
for small enough $\varepsilon\in \Theta(1/\log^2(T))$. It remains to show that Algorithm \ref{algo:search-dtqw-2} succeeds on average with probability $\tilde{\Omega}(1)$, for appropriate choices of $s$ and $t$. Indeed, we show that for the same choices of $s$ and $t$ as in Algorithm \ref{algo:search-dtqw-1}, the average success probability is high. We demonstrate this via the following lemma.
\begin{lemma}
\label{lem:ctrw-success-prob}
Consider an ergodic, reversible Markov chain $P$ and a set of marked nodes $M$. If we choose parameters $s\in \{1-1/r: r=1,2,\cdots, 2^{\lceil\log T\rceil}\}$ and $T\in \Theta\left(HT\right)$ uniformly at random, then the the following holds
\begin{equation}
\mathbb{E}\left[\nrm{\Pi_Me^{(D(s)-I)T}\ket{\sqrt{\pi_U}}}^2\right]\in \Omega\left(1/\log^2 T\right). \nonumber
\end{equation} 
\end{lemma}
~\\~\\
\textbf{Proof sketch:~} We only provide a sketch of the proof here, while we refer the readers to the proof of Lemma S5 in the Supplemental Material of Ref.~\cite{apers2022quadratic}. The full derivation of this lemma can be obtained from the proof therein. The key idea is to show that the quantity we intend to estimate, i.e. $\nrm{\Pi_Me^{(D(s)-I)T}\ket{\pi_U}}^2$ is related to the behaviour of the original Markov chain $P(s)$ (which applies to any reversible Markov chain).
\begin{itemize}
\item[--] The first step is to show that $\nrm{\Pi_M e^{(D(s)-I)t}\ket{\pi_U}}^2$ is lower bounded by the probability of the following event occurring in
a continuous-time Markov chain, for any $t\geq 0$: starting from a distribution over the unmarked elements, a continuous-time random walk $X(s)$ is at some marked vertex after time $t$ and is at an unmarked vertex after time $t+t'$, where $t'>0$. Let us call this event $\mathcal{E}_{X(s)}$.

\item[--] The next step is to then show that the probability of this event occurring on a continuous-time Markov chain is lower bounded by the same event (say $\mathcal{E}_{Y(s)}$) happening in a discrete-time Markov chain $Y(s)$.

\item[--] So, by these two steps we have related the quantity $\nrm{\Pi_Me^{(D(s)-I)T}\ket{\pi_U}}^2$ to the probability of a specific event occurring on a discrete-time Markov chain. At this stage, we can make use of the combinatorial lemma of Ambainis et al. \cite{ambainis2019quadratic}, wherein the authors proved that for any reversible Markov chain $P$, the probability of the event $\mathcal{E}_{Y(s)}$ occurring is $\widetilde{\Omega}(1)$, on average which allows us to prove
$$
\mathbb{E}\left[\nrm{\Pi_Me^{(D(s)-I)T}\ket{\pi_U}}^2\right]=\tilde{\Omega}(1).
$$
\end{itemize}
\[
\pushQED{\qed} 
~\\\qedhere
\popQED
\]
Thus, we managed to develop an optimal quantum spatial search algorithm that relies on the fast forwarding of continuous-time random walks. Moreover, by taking advantage of \textit{Ancilla-free LCU} we require $O(\log HT)$ fewer ancilla qubits, as compared to \textit{Standard LCU}. The recently developed spatial search algorithm by continuous-time quantum walk also falls under the framework of \textit{Ancilla-free LCU}, as it bypasses the need to prepare the (continuous-variable) ancilla register in the Gaussian state. We refer the readers to \cite{apers2022quadratic} for the details. Overall, the \textit{Ancilla-free LCU} framework is applicable for quantum spatial search. It allows to retain a generic quadratic speedup over classical random walks, while using no ancilla qubits (other than the space of the quantum walk). In comparison, \textit{Standard LCU} requires $O(\log HT)$ ancilla qubits which are used to implement multi-qubit controlled operations. More broadly, \textit{Ancilla-free LCU} also helped establish a connection between discrete and continuous-time quantum walks with their classical counterparts. In addition, in the Appendix (Sec.~\ref{sec:relationship-ctqw-dtqw}), using the frameworks of block encoding and QSVT, we show that one can obtain a discrete-time quantum walk from a continuous-time quantum walk (and vice versa). Together, it connects both continuous-time quantum and random walks with discrete-time quantum and random walks, which has been shown in Fig.~\ref{fig:walk-rel}.

\section{Discussion}
\label{sec:discussion}

We considered the framework of Linear Combination of Unitaries, a quantum algorithmic paradigm that has been used to develop several quantum algorithms of practical interest. However, standard techniques to implement LCU require several ancilla qubits, a sequence of multi-qubit controlled operations, and hence are only implementable using fully fault-tolerant quantum computers, which are perhaps decades away. In this work, we significantly reduce the resources required to implement any LCU. Our motivation was to explore whether a broadly applicable framework such as LCU can be implemented on quantum devices that do not have the capabilities of a fully fault tolerant quantum computer. To this end, we provided three new approaches for LCU, considering the different intermediate-term hardware possibilities. 

The \textit{Single-Ancilla LCU} makes repeated use of a short-depth quantum circuit and only a single ancilla qubit, to estimate the expectation value $$\Tr[O\rho]=\Tr[O~f(H)\rho_0 f(H)^{\dag}]/\Tr[f(H)\rho_0f(H)^{\dag}],$$ where $f(H)$ can be well-approximated by a linear combination of unitaries, $O$ is any observable and $\rho_0$ is the initial state. The cost of each coherent run of the generic procedure is always lower than \textit{Standard LCU}. More precisely, we show that the average cost of each run of our procedure is $\langle\tau\rangle$ which is upper bounded by the cost of implementing the most expensive unitary $U_j$, $O(\tau_{\max})$. For \textit{Standard LCU}, the cost of each coherent run is $\tau_Q+2\tau_R$, with $\tau_R$ and $\tau_Q$ being the cost of implementing the \textit{prepare} and \textit{select} unitaries, respectively. We showed that $\tau_Q\geq \tau_{\max}$, and the cost of each run of our method is lower, i.e.\ $\langle\tau\rangle<O(\tau_R+\tau_Q)$. Interestingly, in the worst case, $\tau_Q=O(M\tau_{\max})$. This occurs when each of the $M$ unitaries in the LCU description costs $\tau_{\max}$, indicating the possibility of a significant separation between the cost of each run of \textit{Standard LCU} with our method.

However, for our applications to ground state property estimation and quantum linear systems, $\tau_Q=O(\tau_{\max})$, whenever \textit{Standard LCU} estimates $\Tr[O\rho]$ without using amplitude amplification and estimation. In this regard, it would be interesting to find an application where $\tau_Q$ is significantly larger than $\tau_{\max}$. The cost of each coherent run of \textit{Single-Ancilla LCU} would be significantly lower than \textit{Standard LCU} in such a scenario. Another direction of future research would be to apply our algorithms for Hamiltonian simulation and ground state property estimation to specific Hamiltonians in quantum chemistry \cite{mcardle2020quantum, cao2019quantum,su2021fault} and condensed matter physics \cite{babbush2018low}, and benchmark their performance against other near/intermediate-term quantum algorithms such as those making use of the Hadamard test \cite{zhang2022computing, lin2022heisenberg, wang2022quantum}.   

\textit{Analog LCU} is a physically motivated, continuous-time analogue of the LCU framework. It requires coupling a primary system to a continuous-variable ancilla. We apply this framework to develop continuous-time quantum algorithms for ground state preparation and also for solving quantum linear systems. This framework can be seen as a way to exploit qubit-qumode interactions to perform meaningful computational tasks. Such hybrid systems are currently being engineered in a number of quantum technological platforms such as photonics, trapped-ions, Circuit (or Cavity) QED and superconducting systems \cite{wallraff2004strong,pirkkalainen2013hybrid,kurizki2015quantum, andersen2015hybrid,campagne2020quantum,gan2020hybrid,sawaya2020resource}. In order to experimentally implement the quantum algorithms we discuss, it is crucial to undertake a detailed comparative analysis of the resource requirements for each of these platforms. In future, we plan to develop an experimental proposal in this regard. Our work could lead to further research into whether other, simpler interactions can be engineered on hybrid platforms \cite{arguello2019analogue}. This would help bring generic quantum algorithmic frameworks closer to realization. 

The \textit{Ancilla-free LCU} approach is useful when we are interested in the projection of $f(H)\ket{\psi_0}$ in some subspace of interest, and it suffices if the measurement is successful on average. We have shown that it is applicable to the framework of quantum walks, in particular, to quantum spatial search algorithms. This technique has been useful to connect discrete and continuous-time quantum walks, with their classical counterparts. Using this framework, we have also developed other results. We believe that this method is more widely applicable to quantum optimization and sampling algorithms such as quantum simulated annealing \cite{somma2008quantum} and quantum Metropolis sampling \cite{temme2011quantum,yung2012quantum}.

Overall, the new LCU techniques we develop are quite generic. Owing to the wide applicability of the LCU framework, our work can lead to the development of several new quantum algorithms beyond those considered here. One immediate direction of future research would be to investigate whether variants of other generic quantum algorithmic paradigms which require access to the block encoding of an operator (such as QSVT \cite{gilyen2019quantum,martyn2021grand}), can be modified so that they are implementable on intermediate-term quantum computers. 
\begin{acknowledgments}
I thank Andr\'{a}s Gily\'{e}n, J\'{e}r\'{e}mie Roland, Simon Apers, and Leonardo Novo for helpful discussions. I also thank Samson Wang and Mario Berta for providing feedback on an early version of the manuscript. I am grateful to Leonardo Novo and Hamed Mohammady for proofreading this manuscript. I acknowledge funding from the Science and Engineering Research Board, Department of Science and Technology (SERB-DST), Government of India under Grant No. SRG/2022/000354, and from the Ministry of Electronics and Information Technology (MeitY), Government of India, under Grant No. 4(3)/2024-ITEA. I also acknowledge support from Fujitsu Ltd, Japan and from IIIT Hyderabad via the Faculty Seed Grant. Finally, I would like to thank Tanima Karmakar for acting as a sounding board during the writing of this manuscript.
~\\~\\
\textit{Note added:~} \textit{We refer the readers to the concurrent independent work of Wang, McArdle, and Berta~\cite{wang2023qubit}, where a technique similar to the \textit{Single-Ancilla LCU} method was applied to ground state preparation and quantum linear systems.} 
\end{acknowledgments}

\bibliography{bibliography}
\bibliographystyle{unsrturl}

\newpage
\setcounter{equation}{0}
\setcounter{figure}{0}
\setcounter{table}{0}
\setcounter{algocf}{0}
\setcounter{section}{0}
\setcounter{theorem}{0}
\renewcommand{\theequation}{A\arabic{equation}}
\renewcommand{\thetable}{A\arabic{table}}
\renewcommand{\thefigure}{A\arabic{figure}}
\renewcommand{\thetheorem}{A\arabic{theorem}}
\renewcommand{\thelemma}{A\arabic{lemma}}
\renewcommand{\thealgocf}{A\arabic{algocf}}
\renewcommand{\thecorollary}{A\arabic{corollary}}
\renewcommand{\thesection}{A - \Roman{section}}

\appendix
\begin{center}
\textbf{\huge Appendix}
\end{center}
Here, we provide detailed derivations of the unproven theorems/lemmas in the main manuscript, as well as develop some additional general results from the LCU techniques that have been introduced. In Sec.~\ref{sec:appendix-robustness-norm}, we provide a proof of a slightly more general version of Theorem \ref{thm:norm-approx}. We provide detailed derivations of the unproven results of our Hamiltonian Simulation procedure (Sec.~\ref{sec:ham-sim} of the main manuscript) in Sec.~\ref{sec-app:ham-sim}. Recall that in the main manuscript, we provided randomized quantum algorithms for ground state property estimation, as well as quantum linear systems. However, we assumed access to a Hamiltonian evolution oracle. In Sec.~\ref{sec-appendix:single-ancilla-LCU}, we analyze the performance of these algorithms while considering particular Hamiltonian simulation procedures. The goal is to provide end-to-end complexities (in terms of the gate depth required) while still requiring only a single ancilla qubit (and no multi-qubit controlled operation). We provide an optimal circuit model quantum algorithm for ground state preparation in Sec.~\ref{subsec:gsp-qsvt} which makes use of QSVT to implement the polynomial $e^{-tx^2}$. Finally, in Sec.~\ref{sec:relationship-ctqw-dtqw}, we use block encoding and QSVT to obtain a relationship between discrete-time and continuous-time quantum walks.
\section{Proof of Theorem \ref{thm:norm-approx}}
\label{sec:appendix-robustness-norm}
Here we prove a general version of the statement of Theorem \ref{thm:norm-approx}.
~\\
\begin{theorem}[Robustness of normalization factors]
\label{thm:norm-approx-general}
Let $\varepsilon\in (0,1)$, $\rho_0$ be some initial state and $P$ be an operator. Furthermore, let $\ell_*\in \mathbb{R}^+$ satisfies $\ell^2=\Tr[P\rho_0 P^{\dag}]\geq \ell_*$, and $O$ be some observable with $\norm{O}\geq 1$. Suppose for positive integers $a,b>1$, we obtain an estimate $\tilde{\ell}$ such that
\begin{equation}
\label{eq:normalization-dist-general}
\left|\tilde{\ell}-\ell^2\right|\leq \dfrac{\varepsilon \ell_*}{a\norm{O}},
\end{equation}
and some parameter $\mu$ such that,
\begin{equation}
\label{eq:estimation-dist-1-general}
\left|\mu-\Tr[O~P \rho_0 P^{\dag}]\right| \leq  \dfrac{\varepsilon\ell_*}{b},
\end{equation}
then,
$$
\left|\dfrac{\mu}{\tilde{\ell}}-\dfrac{\Tr[O~P \rho_0 P^{\dag}]}{\ell^2}\right| \leq \dfrac{\varepsilon}{a-1}+\dfrac{\varepsilon}{b(a-1)}+\dfrac{\varepsilon}{b}.
$$
\end{theorem}
\begin{proof}
By the triangle inequality, we have
~\\
\begin{align}
\left|\dfrac{\mu}{\tilde{\ell}}-\dfrac{\Tr[O P \rho_0 P^{\dag}]}{\ell^2}\right|&\leq \left|\dfrac{\mu}{\ell^2}-\dfrac{\mu}{\tilde{\ell}}\right|+\left|\dfrac{\mu}{\ell^2}-\dfrac{\Tr[O P \rho_0 P^{\dag}]}{\ell^2}\right|\\
      & \leq \left|\dfrac{1}{\ell^2}-\dfrac{1}{\tilde{\ell}}\right|\left|\mu\right|+\dfrac{\varepsilon\ell_*}{b\ell^2}~~~~~~~~~~~~~~~~~~~~[\text{~Using Eq.~\eqref{eq:estimation-dist-1-general}~}]\\
      & \leq \left|\dfrac{1}{\ell^2}-\dfrac{1}{\tilde{\ell}}\right|\left|\mu\right|+\dfrac{\varepsilon}{b}~~~~~~~~~~~~~~~~~~~~~~~~[\text{~As~}\ell^2\geq \ell_*~]\\
      &\leq \left|\dfrac{1}{\ell^2}-\dfrac{1}{\tilde{\ell}}\right|\left(\dfrac{\varepsilon \ell_*}{b}+\Tr[O P \rho_0 P^{\dag}]\right)+\dfrac{\varepsilon}{b}~~~~~~~~~~[\text{~Using Eq.~\eqref{eq:estimation-dist-1-general}~}]\\
      & \leq \dfrac{1}{\ell^2} \left|\dfrac{\tilde{\ell}-\ell^2}{\tilde{\ell}}\right|\left(\dfrac{\varepsilon \ell_*}{b}+\norm{O}\norm{P\rho_0 P^{\dag}}_1\right)+\varepsilon/b~~~~~\left[\text{~Using Lemma \ref{thm:holder} with $p=\infty, q=1$~}\right]\\
      & \leq \dfrac{1}{\ell^2} \left|\dfrac{\tilde{\ell}-\ell^2}{\tilde{\ell}}\right|\left(\dfrac{\varepsilon \ell_*}{b}+\norm{O}\ell^2\right)+\varepsilon/b~~~~~~~~~~\left[\text{~As $\norm{P\rho_0 P^{\dag}}_1=\Tr[P\rho_0 P^{\dag}]=\ell^2$~}\right]\\
      & \leq \left|\dfrac{\tilde{\ell}-\ell^2}{\tilde{\ell}}\right|\left(\dfrac{\varepsilon}{b}+\norm{O}\right)+\varepsilon/b~~~~~~~~~~~~~~~~~~~~~[\text{~As~}\ell^2\geq \ell_*]\\
      &\leq \left|\dfrac{1}{\tilde{\ell}}\right| \dfrac{\varepsilon\ell_*}{a\norm{O}}\left(\norm{O}+\dfrac{\varepsilon}{b}\right) +\varepsilon/b ~~~~~~~~~~~~~~~~~~[\text{~Using Eq.~\eqref{eq:normalization-dist-general}~}]\\
      &\leq \left|\dfrac{1}{\ell^2-\frac{\varepsilon\ell_*}{a~\norm{O}}}\right| \dfrac{\varepsilon\ell_*}{a~\norm{O}}\left(\norm{O}+\dfrac{\varepsilon}{b}\right) +\varepsilon/b\\
       &\leq  \left(\sum_{k=0}^{\infty}\left(\dfrac{\varepsilon\ell_*}{a\ell^2\norm{O}}\right)^k\right)\dfrac{\varepsilon \ell_*}{a~\ell^2\norm{O}}\left(\norm{O}+\dfrac{\varepsilon}{b}\right) +\varepsilon/b~~~~[\text{Taylor series expansion of $1/\tilde{\ell}$.}]\\
      &\leq  \dfrac{a}{a-1}\left(\dfrac{\varepsilon}{a}+\dfrac{\varepsilon^2}{ab~\norm{O}}\right) +\varepsilon/b~~~~~~~\left[~\mathrm{As~} \sum_{k=0}^{\infty}\left(\dfrac{\varepsilon\ell_*}{a\ell^2\norm{O}}\right)^k\leq \sum_{k=0}^{\infty}\left(\dfrac{1}{a}\right)^k = a/(a-1)~\right]\\
      & \leq \dfrac{\varepsilon}{a-1}+\dfrac{\varepsilon}{(a-1)b}+\dfrac{\varepsilon}{b}~~~~~~~~~~~~\left[\mathrm{~As~}\dfrac{\varepsilon^2}{(a-1)b~\norm{O}}\leq \frac{\varepsilon}{(a-1)b}~\right].
\end{align}
\end{proof}
~\\
Theorem \ref{thm:norm-approx} is a particular case of this theorem, where we substitute $a=b=3$.

\section{Hamiltonian simulation: Detailed proofs}
\label{sec-app:ham-sim}
In this section, we will proof the results of Sec.~\ref{sec:ham-sim} in detail. Recall that, we considered a Hamiltonian
$$
H=\sum_{k=1}^{L}\lambda_k P_k,
$$
where $P_k$ are strings of Pauli operators, such that $\beta=\sum_{k}|\lambda_k|$. First, set $\tilde{H} = H/\beta$ and $\tilde{t}=\beta t$. This give us,
\begin{equation}
\label{eq:ham-sim-hamiltonian}
\tilde{H} = H/\beta =\sum_{k=1}^{L} p_k P_k,
\end{equation}
where $\sum_{k} |p_k| =1$. Also, 
$$
e^{-iHt}=\left(e^{-iHt/r}\right)^r=\left(e^{-i\tilde{H}\tilde{t}/r}\right)^r,
$$
where $r$ (to be selected later) is a parameter such that $r>t$. 

First note that by truncating $S_r=e^{-i\tilde{H}\tilde{t}/r}$ to $K$ terms, we obtain
$$
\tilde{S}_r=\sum_{k=0}^{K} \dfrac{(-i\tilde{t}\tilde{H}/r)^k}{k!}.
$$
Then by choosing some 
$$
K=O\left(\dfrac{\log(r/\gamma)}{\log\log(r/\gamma)}\right),
$$
we ensure that $\norm{S_r - \tilde{S}_r}\leq \gamma/r$.

We obtain the LCU decomposition of $\tilde{S}_r$, similar in spirit to Ref.~\cite{wan2022randomized}. This gives us, 
\begin{align}
\tilde{S}_r&=\sum_{k=0}^{K} \dfrac{(-i\tilde{t}\tilde{H}/r)^k}{k!}\\
		   &=\sum_{k=0,~k\in \mathrm{even}}^{K} \dfrac{1}{k!}(-i\tilde{t}\tilde{H}/r)^k \left(I- \dfrac{i\tilde{t}\tilde{H}/r}{k+1}\right)\\
		   &=\sum_{k=0,~k\in \mathrm{even}}^{K} \dfrac{1}{k!}\left(-i\tilde{t}/r \sum_{\ell=1}^{L} p_{\ell} P_{\ell}\right)^k \left(I- \dfrac{i\tilde{t}/r}{k+1}\left(\sum_{m=1}^{L} p_{m} P_{m}\right)\right)\\
		   &= \sum_{k=0,~k\in \mathrm{even}}^{K} \dfrac{(-i\tilde{t}/r)^k}{k!}\sum_{\ell_1,\ell_2,\cdots \ell_k = 1}^{L} p_{\ell_1}p_{\ell_2}\cdots p_{\ell_k} P_{\ell_1}P_{\ell_2}\cdots P_{\ell_k} \sum_{m=1}^{L} p_m \left(I- \dfrac{i\tilde{t}P_m /r}{k+1}\right)\\
		   &=\sum_{k=0,~k\in \mathrm{even}}^{K} \dfrac{(-i\tilde{t}/r)^k}{k!} \sqrt{1+\left(\dfrac{\tilde{t}/r}{k+1}\right)^2}\sum_{\ell_1,\ell_2,\cdots \ell_k, m = 1}^{L} p_{\ell_1}p_{\ell_2}\cdots p_{\ell_k} p_m P_{\ell_1}P_{\ell_2}\cdots P_{\ell_k} e^{-i\theta_{m} P_m},
\end{align}
where $e^{-i\theta_{m}P_m}$ is a Pauli rotation operator, defined as follows:
\begin{align}
e^{-i\theta_{m}P_m}=\dfrac{1}{\sqrt{1+\left(\dfrac{\tilde{t}/r}{k+1}\right)^2}}\left(I- \dfrac{i\tilde{t}P_m /r}{k+1}\right),
\end{align}
such that
\begin{equation}
\theta_m = \arccos\left(\left[1+\left(\dfrac{\tilde{t}/r}{k+1}\right)^2\right]^{-1/2}\right).
\end{equation}
Thus, $\tilde{S}_r=\sum_{j\in M} \alpha_j U_j$, where the index set $M$ can be defined as 
$$
M =\left\{(k, \ell_1, \ell_2,\cdots \ell_k, m): 0\leq k \leq K;  \ell_1,\ell_2, \cdots \ell_k, m \in \{1,2,\cdots, L\}\right\}. 
$$
Also,
\begin{equation}
\label{eq:simulation-segment-lcu-coefficient}
\alpha_j = \dfrac{(\tilde{t}/r)^k}{k!} \sqrt{1+\left(\dfrac{\tilde{t}/r}{k+1}\right)^2}p_{\ell_1}p_{\ell_2}\cdots p_{\ell_k} p_m,
\end{equation}
while
$$
U_j=(-i)^k P_{\ell_1}P_{\ell_2}\cdots P_{\ell_k} e^{-i\theta_{m} P_m}.
$$
Now, the sum of coefficients
\begin{align}
\sum_{j\in M} |\alpha_j| &= \sum_{k=0,~k\in \mathrm{even}}^{K}  \dfrac{(\tilde{t}/r)^k}{k!} \sqrt{1+\left(\dfrac{\tilde{t}/r}{k+1}\right)^2}\sum_{\ell_1,\ell_2,\cdots \ell_k, m = 1}^{L} p_{\ell_1}p_{\ell_2}\cdots p_{\ell_k} p_m\\
						&=\sum_{k=0,~k\in \mathrm{even}}^{K}  \dfrac{(\tilde{t}/r)^k}{k!} \sqrt{1+\left(\dfrac{\tilde{t}/r}{k+1}\right)^2}\\
						& \leq \sum_{k=0,~k\in \mathrm{even}}^{\infty}  \dfrac{(\tilde{t}/r)^k}{k!} \sqrt{1+\left(\dfrac{\tilde{t}/r}{k+1}\right)^2}= \sum_{k=0}^{\infty}  \dfrac{(\tilde{t}/r)^{2k}}{(2k)!} \sqrt{1+\left(\dfrac{\tilde{t}/r}{2k+1}\right)^2}\\
						& \leq \sum_{k=0}^{\infty}  \dfrac{(\tilde{t}/r)^{2k}}{k!} = e^{\tilde{t}^2/r^2}.
\end{align}
Finally, in order to write down $S$ as an LCU, we write $S=\tilde{S}_r^r$. That is,
\begin{align}
S&=\left(\sum_{j\in M} \alpha_j U_j\right)^r=\sum_{j_1, j_2,\cdots j_r \in M} \alpha_1 \alpha_2 \cdots \alpha_r~U_{j_1} U_{j_2}\cdots U_{j_r}=\sum_{m}c_m W_m,
\end{align}
where $\norm{c}_1=\sum_{m}|c_m|=(\sum_{j\in M} |\alpha_j|)^r\leq e^{\tilde{t}^2/r}$. We choose $r=\tilde{t}^2=\beta^2t^2$, which ensures $\norm{c}_1=O(1)$. Moreover, for this choice of $r$ and  
$$\gamma\leq \dfrac{\varepsilon}{6\norm{O}},$$
by truncating the Taylor series of $e^{itH/r}$ at some
$$
K=O\left(\dfrac{\log(\beta t\norm{O}/\varepsilon)}{\log\log(\beta t\norm{O}/\varepsilon)}\right),
$$
we have 
$$
\norm{e^{-itH}-S}\leq \dfrac{\varepsilon}{6\norm{O}}.
$$
So from Theorem \ref{thm:randomized-time-evolution}, if we can sample $V_1, V_2$ from the LCU decomposition of $S$, we will be able to output an $\varepsilon$-accurate estimate of $\Tr[O~e^{-itH}\rho_0 e^{itH}]$, using Algorithm \ref{algo:randomized-time-evolution}. We discuss this sampling strategy in a bit more detail as compared to the main manuscript here.
~\\~\\
\textbf{Sampling $\mathbf{V_1}$ and $\mathbf{V_2}$:~} We first pick an even integer $k\in [0,K]$ according to $\alpha_j/\sum_j \alpha_j$ and select $k+1$ unitaries, $P_{\ell_1}, P_{\ell_2},\cdots P_{\ell_k}$, and $P_m$ (as in Eq.~\eqref{eq:simulation-segment-unitary}), where each $P_{\ell_i}$ is sampled according to the distribution $\{p_{\ell_i}\}^{L}_{\ell_i=1}$  and $P_m$ is sampled from $\{p_m\}^{L}_{m=1}$. From this sampling procedure, we can obtain a product of the unitaries $W_1=(-i)^k P_{\ell_1}P_{\ell_2}\cdots P_{\ell_k} P_m$, of $k$ Pauli operators and a Pauli rotation. Finally, we repeat this procedure $r$ times, which essentially results in the final unitary 
$$
W=W_r W_{r-1}\cdots W_1.
$$
Thus, this sampling procedure outputs some unitary $W$ such that $\mathbb{E}[W]=S/\norm{c}_1$. 

This allows us to use Algorithm \ref{algo:randomized-time-evolution} and Theorem \ref{thm:randomized-time-evolution}. Clearly, each run of our procedure has gate depth at most $O(Kr)$, which leads to the gate depth per coherent run, and the overall gate depth as stated in Theorem \ref{thm:ham-sim}.
\section{Single Ancilla LCU: from the Hamiltonian Evolution model to gate depth}
\label{sec-appendix:single-ancilla-LCU}
In the main manuscript, we analyzed the complexity of our Hamiltonian simulation algorithm by \textit{Single-Ancilla LCU} (Sec.~\ref{sec:ham-sim}) in terms of the gate depth per coherent run, as well as the overall gate depth. On the other hand, for both ground state property estimation (Sec.~\ref{subsec:gsp-single-ancilla}) and estimating expectation values of observables with respect to the solution of quantum linear systems (Sec.~\ref{subsec:single-ancilla-qls}), we assumed that the Hamiltonian $H$ can be accessed through the Hamiltonian evolution oracle $U_{\tau}=\exp[-iH\tau]$. We measured the performance of our algorithm in terms of (a) the maximal time evolution of $H$ in one coherent run (denoted as $\tau_{\max}$), and (b) the number of classical repetitions $T$, where $O(\tau_{\max}\cdot T)$ is the total evolution time. As argued in the main article, both (a) and (b) are different from the actual circuit depth required to implement these algorithms. 

In this section, we obtain the gate depth required to run both these algorithms using the \textit{Single-Ancilla LCU} method. To this end, we consider specific Hamiltonian simulation algorithms to implement $U_t$ to some desired precision. Our goal is to keep the number of ancilla qubits to just one, and avoid the use of multi-qubit controlled operations. As we shall see next, this limits the simulation algorithms we can use. For both algorithms, we assume the following:
\begin{itemize}
\item[(i)]~ The Hamiltonian is a linear combination of unitaries (e.g.~strings of Paulis), i.e.\ $H=\sum_{k=1}^{L}\lambda_k P_k$. The total weight of the coefficients $\beta=\sum_{k}|\lambda_k|$.~\\
\item[(ii)]~The observable $O$ we intend to measure is itself a linear combination of easy to implement unitary observables, i.e.\ $O=\sum_{j=1}^{L_O} h_{j} O_{j}$, such that $\norm{h}_1=\sum_{j}|h_j|$ and for each $j$, $\norm{O_{j}}=1$.
\end{itemize}
~\\
Note that any block encoding of $H$ requires $O(\log L)$ ancilla qubits, and has a sub-normalization factor of $\beta$, while a block encoding of $O$ requires $O(\log L_O)$ ancilla qubits, with $\norm{h}_1$ being the sub-normalization factor. Additionally this requires a gate depth of $O(L)$ and $O(L_O)$, respectively. Thus, constructing the block encoding of both $H$ and $O$ needs multi-qubit controlled operations and adds to the overall gate depth, which are undesirable for early fault-tolerant quantum computers. Moreover, our goal of using a solitary ancilla qubit (and hence, no multi-qubit control) for our algorithms implies that we cannot use any Hamiltonian simulation algorithm that uses the block encoding of $H$. This rules out Hamiltonian simulation by qubitization, the state-of-the-art method \cite{low2019hamiltonian}. In fact, given this restriction we can only use Trotter-based methods and the Hamiltonian simulation algorithm based on \textit{Single-Ancilla LCU} (Sec.~\ref{sec:ham-sim}). 

We indeed show that whenever $H$ and $O$ are linear combinations of unitaries, we can implement \textit{Single-Ancilla LCU} using only a single ancilla qubit and no multi-qubit controlled unitaries. In this regard, we first adapt the generic \textit{Single-Ancilla LCU} scheme to allow for (a) the measurement of any $O=\sum_{j} h_j O_j$, and (b) the application of imprecise unitaries. The reason for analyzing (b) is that whenever $f(H)\approx \sum_{j} c_j e^{-ijH}$, the Hamiltonian simulation algorithm needs to be implemented to some desired precision, which is what we shall estimate. Equipped with (a) and (b), we can directly calculate the gate depth of both ground state property estimation, and quantum linear systems, by invoking particular Hamiltonian simulation algorithms that fit our goals. We begin by incorporating (a) and (b) into Theorem \ref{thm:single-ancilla-overall}. 

\subsection{Single Ancilla LCU: general observables and imperfect unitaries}

In order to measure any $O=\sum_{j=1}^{L_O} h_{j} O_{j}$ within the framework of \textit{Single-Ancilla LCU}, we do the following: instead of measuring $O$ directly, we simply sample an $O_j$ 
according to $\{h_j/\norm{h}_1\}_{j}$, and implement a POVM measurement of $X\otimes O_j$ in Step 3 of Algorithm \ref{algo:randomized-time-evolution}. Since~$\norm{O_j}=1$, the POVM measurement yields an outcome in $[-1,+1]$. Note that this strategy ensures that the expected outcome of the $j^{\mathrm{th}}$ iteration of Algorithm \ref{algo:randomized-time-evolution} is
$$
\mathbb{E}[\mu_j]=\dfrac{1}{\norm{c}^2_1\norm{h}_1} \Tr[O~g(H)\rho_0 g(H)^{\dag}],
$$
where $g(H)=\sum_{j}c_j U_j$ is the LCU that approximates the function $f(H)$ we wish to apply. So, $\mu$ and $\tilde{\ell}$ can be obtained as in Algorithm \ref{algo:single-ancilla-overall}, except now $\mu=\norm{c}^2_1 \norm{h}_1\sum_{j}\mu_j/T$, and $\norm{O} \leq \norm{h}_1$. If we consider that the cost of implementing any $O_j=\Theta(1)$, then the cost of each coherent run is still upper bounded by $O(2\tau_{\max}+\tau_{\rho_0})$. Furthermore, following the arguments of the proofs of Theorem \ref{thm:randomized-time-evolution} and Theorem \ref{thm:single-ancilla-overall}, it is easy to show that $\mu/\tilde{\ell}$ is an $\varepsilon$-accurate estimate of $\Tr[O\rho]$, with a constant success probability, for
$$
T=O\left(\dfrac{\norm{c}^4_1 \norm{h}^2_1}{\varepsilon^2\ell^2_*}\right).
$$

In order to take into account the implementation of imperfect unitaries, as in Theorem \ref{thm:randomized-time-evolution}, consider that $f(H)$ is the function we wish to apply, and $g(H)=\sum_{j}c_jU_j$, is the LCU it approximates. However now, instead of $g(H)$, we implement some $h(H)$ such that $h(H)=\sum_{j}c_j \tilde{U}_j$. Then if
\begin{align}
\label{eqsup:function-lcu-upper-bound}
\norm{f(H)-g(H)}\leq \dfrac{\varepsilon}{9\norm{O}\norm{f(H)}}, 
\end{align}
and,
\begin{align}
\label{eqsup:exact-lcu-imprecise-lcu-upper-bound}
\norm{g(H)-h(H)}\leq \dfrac{\varepsilon}{9\norm{O}\norm{f(H)}}, 
\end{align}
it suffices if in Theorem \ref{thm:randomized-time-evolution}, $\mu$ outputs an $\varepsilon/3$-accurate estimate of $\Tr[O~h(H)\rho_0h(H)^{\dag}]$. This ensures, 
\begin{align}
\left|\mu-\Tr[O~f(H)\rho_0 f(H)^{\dag}]\right|&\leq \left|\mu-\Tr[O~h(H)\rho_0 h(H)^{\dag}]\right|+ \left|\Tr[O~f(H)\rho_0 f(H)^{\dag}]-\Tr[O~h(H)\rho_0 h(H)^{\dag}]\right|\\
&\leq \varepsilon/3+ 3\norm{f(H)}\norm{O}\norm{f(H)-h(H)}~~~~~~[~\mathrm{Using~Theorem~\ref{thm:distance-expectation}}~]\\
          & \leq \varepsilon/3+ 3\norm{f(H)}\norm{O}\left[\norm{f(H)-g(H)}+\norm{g(H)-h(H)}\right]\\
          & \leq \varepsilon.
 \end{align}
The equivalent statement of Theorem \ref{thm:single-ancilla-overall} remains the same with the upper bound on the precision adjusted appropriately. Here, we combine both these results on performing \textit{Single-Ancilla LCU} with (a) Imperfect unitaries, and (b) observables that are LCU, and state our findings formally via the following theorem:   
~\\~\\
\begin{theorem}
\label{thmsup:single-ancilla-lcu-local-observable}
Let $\varepsilon,\delta\in (0,1)$, $O$ be some observable such that $\sum_{j=1}^{L_O} h_{j} O_{j}$, with $\norm{h}_1=\sum_{\alpha}|h_{\alpha}|$ , and for any $j$, $\norm{O_{j}}=1$. Also, let $\rho_0$ be some initial state, prepared in cost $\tau_{\rho_0}$. Suppose $H\in\mathbb{C}^{N\times N}$ be a Hermitian matrix such that for some function $f:[-1,1]\mapsto \mathbb{R}$ and unitaries $\{U_j\}_{j}$, 
$$\norm{f(H)-\sum_{j}c_j U_j}\leq \dfrac{\varepsilon\ell_*}{27\norm{h}_1\norm{f(H)}},
$$ 
and $\ell^2=\Tr[f(H)\rho_0 f(H)^{\dag}]\geq \ell_*$. Moreover, suppose that each $U_j$ can only be imperfectly implemented: $\tilde{U}_j$ approximates $U_j$ such that
$$
\max_j \norm{U_j-\tilde{U}_j}\leq \dfrac{\varepsilon\ell_*}{27\norm{h}_1\norm{f(H)}\norm{c}_1}.
$$ 
Furthermore, suppose that the maximum cost of implementing any $\tilde{U}_j$ is at most $\tau_{\max}$. Then there exists an algorithm that outputs $\mu$ and $\tilde{\ell}$ such that
$$
\left|\mu/\tilde{\ell} - \Tr[O\rho]\right| \leq \varepsilon,
$$
with probability $(1-\delta)^2$, using 
$$
T=O\left(\dfrac{\norm{c}^4_1 \norm{h}^2_1}{ \varepsilon^2\ell^2_*}\ln(1/\delta)\right) 
$$  
classical repetitions, where the cost of each such run is at most $O(2\tau_{\max}+\tau_{\rho_0})$. 
\end{theorem}

\begin{proof}
Let $g(H)=\sum_{j}c_j U_j$ and $h(H)=\sum_{j}c_j \tilde{U}_j$. Then from the upper bound of $\max_j \norm{U_j-\tilde{U}_{j}}$ mentioned in the statement of the theorem, we have 
\begin{align*}
\norm{g(H)-h(H)}&= \norm{\sum_{j}c_j \left(U_j-\tilde{U}_j\right)}\\
                &\leq \dfrac{\varepsilon\ell_*}{27\norm{h}_1\norm{f(H)}}.
\end{align*}
Now, from Theorem \ref{thm:norm-approx}, we need to obtain an estimate $\tilde{\ell}$ of $\ell^2$ such that
$$
\left|\tilde{\ell}-\ell^2\right|\leq \dfrac{\varepsilon\ell_*}{3\norm{h}_1},
$$
and also output a $\mu$ such that
$$
\left|\mu-\Tr[O~f(H)\rho_0 f(H)^{\dag}]\right|\leq \varepsilon \ell_*/3.
$$
The upper bounds on $\norm{f(H)-g(H)}$ and $\norm{g(H)- h(H)}$ can be obtained by substituting $\varepsilon$ with $\varepsilon\ell_*/3$ in Eq.~\eqref{eqsup:function-lcu-upper-bound} and Eq.~\eqref{eqsup:exact-lcu-imprecise-lcu-upper-bound}, respectively (additionally substituting $\norm{O}$ with $\norm{h}_1$). Both $\mu$ and $\tilde{\ell}$ can be obtained by running Algorithm \ref{algo:single-ancilla-overall}, except in each iteration the observable measured is some $O_j$ sampled according to $h_j/\norm{h}_1$. 
\end{proof}

If $g(H)$ is a linear combination of time evolution operators, i.e.\ $g(H)=\sum_{j}c_j e^{-iHt_j}$, Theorem Theorem \ref{thmsup:single-ancilla-lcu-local-observable} gives us the precision with which Hamiltonian simulation needs to be performed. This allows us to analyze the complexity of our algorithms for ground state property estimation (Sec.~\ref{subsec:gsp-single-ancilla}) and quantum linear systems (Sec.~\ref{subsec:single-ancilla-qls}) in terms of their gate depth. Clearly, if the maximal time evolution of the underlying algorithm was $\tau_{\max}$, now we can run a Hamiltonian simulation algorithm to implement $e^{-itH}$ for $t=\tau_{\max}$ and precision $O(\varepsilon \ell_* \norm{c}^{-1}_1\norm{f(H)}^{-1}\norm{h}_1^{-1})$.

\subsection{Ground state property estimation and quantum linear systems}
\label{subsec-app:gsp-qls}

In this section, we assume $H=\sum_{k=1}^{L} c_k P_k$ with $\beta=\sum_{k}|c_k|$, and the observable $O=\sum_{j=1}^{L_O}h_jO_j$, with $\norm{O_j}=1$. The assumptions on $H$ for ground state property estimation are quite natural, as this is precisely the form of most physical Hamiltonians. For quantum linear systems too, we assume that the matrix to be inverted can be written down as linear combination of Paulis, which is non-standard. Indeed generally, it is assumed that $H$ can be accessed via a block encoding (implicitly implying that $H$ is sparse or it is stored in some quantum accessible data structure \cite{chakraborty2019power}). However, (i) this requires an additional overhead which is undesirable in the intermediate-term and (ii) assumes access to a quantum RAM. It is then reasonable to assume that in the early fault-tolerant era, quantum linear systems algorithm will be employed to solve physically relevant problems, where the underlying data matrix directly corresponds to some physical Hamiltonian, which would avoid the need to handle classical data. This motivates using $H$ which can be expressed as a linear combination of Paulis, also for solving quantum linear systems.

For both of our quantum algorithms, we need to implement a linear combination of Hamiltonian evolution operators, i.e.\
$$
f(H)\approx \sum_{j} c_j e^{-it_jH}.
$$
In order to estimate the gate depth, we need to choose a specific Hamiltonian simulation procedure to implement $e^{-iHt}$ to the desired precision. The goal of running our algorithms with a solitary ancilla qubit and no multi-qubit controlled gates imply that we can make use of only Trotter-based methods and the \textit{Single-Ancilla LCU} Hamiltonian simulation algorithm. Both these methods require have a super linear dependence on $t$ in terms of the gate depth, which is suboptimal (See Table \ref{table:comparison-ham-sim}). However, state-of-the-art methods which make use of a block encoding of $H$, have a gate depth per coherent run which depends on $L$. Thus, there exist regimes where our method requires a shorter gate depth per iteration. 

\begin{table}[ht!!]
\begin{center}
    \resizebox{\columnwidth}{!}{
    \renewcommand{\arraystretch}{3} 
    \begin{tabular}{|c|c|c|c|c|c|}
    \hline
    Problem & \vtop{\hbox{\strut Hamiltonian Simulation} \hbox{\strut ~~~~~~~procedure used}} & Ancilla & \vtop{\hbox{\strut ~~Gate depth} \hbox{\strut per coherent run}} & Classical repetitions \\ \hline\hline
  \multirow{2}{*}{\vtop{\hbox{\strut ~~~~Ground state} \hbox{\strut property estimation}}} & \vtop{\hbox{\strut Single-Ancilla LCU} \hbox{\strut ~~~~~~~(Sec.~\ref{sec:ham-sim})}} & $1$ & $\widetilde{O}\left(\dfrac{\beta^2}{\Delta^2}\right)$ & $O\left(\dfrac{\norm{h}_1^2}{\eta^4\varepsilon^2}\right)$ \\ \cline{2-5}
    
   & $2k$-order Trotter & $1$ & $\widetilde{O}\left(L\left(\dfrac{\beta}{\Delta}\right)^{1+\frac{1}{2k}}\left(\dfrac{\norm{h}_1}{\varepsilon\eta^2}\right)^{\frac{1}{2k}}\right)$ & $O\left(\dfrac{\norm{h}_1^2}{\eta^4\varepsilon^2}\right)$ \\ \hline

 \multirow{2}{*}{Quantum Linear Systems} & \vtop{\hbox{\strut Single-Ancilla LCU} \hbox{\strut ~~~~~~~(Sec.~\ref{sec:ham-sim})}} & $1$ & $\widetilde{O}\left(\beta^2\kappa^2\right)$ & $\widetilde{O}\left(\dfrac{\kappa^4\norm{h}_1^2}{\varepsilon^2}\right)$ \\ \cline{2-5}
 
  & $2k$-order Trotter & $1$ & $\widetilde{O}\left(L\kappa^{1+\frac{3}{2k}}\left(\beta \right)^{1+\frac{1}{2k}}\left(\dfrac{\norm{h}_1}{\varepsilon}\right)^{\frac{1}{2k}}\right)$ & $\widetilde{O}\left(\dfrac{\kappa^4\norm{h}_1^2}{\varepsilon^2}\right)$ \\ 
 \hline
\end{tabular}}
  \caption{Summary of the complexity of quantum algorithms by \textit{Single-Ancilla LCU} for ground state property estimation and quantum linear systems. For both the algorithms we assume that $H=\sum_{k=1}^{L}\lambda_k P_k$, where $P_k$ is a string of Pauli operators. We define $\beta=\sum_{k}|\lambda_k|$. Also, for both cases we assume that the observable $O=\sum_{j=1}^{L_O} h_j O_j$, with total weight $\norm{h}_1=\sum_j |h_j|$ and each $\norm{O_j}=1$. The complexity is measured in terms of the gate depth per coherent run of the algorithm, with the overall gate depth being the product of this quantity with the total number of classical repetitions (product of the complexity of the last two columns). For ground state property estimation, we assume that $H$ has a spectral gap $\Delta$, and we have knowledge of its ground energy to some precision. Furthermore, we also assume that we have an initial state $\ket{\psi_0}$ with an overlap of at least $\eta$ with $\ket{v_0}$, the unknown ground state of $H$. Our algorithm estimates $\braket{v_0|O|v_0}$ to additive accuracy $\varepsilon$. On the other hand, for quantum linear systems we assume that $H$ has eigenvalues in the range $[-1,-1/\kappa]\cup [1/\kappa,1]$ and access to an initial state $\ket{b}$. The algorithm outputs an $\varepsilon$-accurate estimate of $\braket{x|O|x}$, where $\ket{x}=H^{-1}\ket{b}/\norm{H^{-1}\ket{b}}$.} 
\label{table:circuit-depth-single-ancilla-gsp-qls}
    \end{center}
\end{table}
\renewcommand{\arraystretch}{1}

We consider two Hamiltonian simulation procedures: (a) Hamiltonian simulation by \textit{Single-Ancilla LCU}, and (b) $2k$-order Trotter. For these methods, given any $H$ and $O$ defined as above, in order to implement the LCU $g(H)=\sum_j c_j e^{-it_jH}$ requires invoking Hamiltonian simulation for a maximal time $\max_j t_j$ and precision of at most $\varepsilon_H=O(\varepsilon\ell_*\norm{c}^{-1}_1 \norm{h}^{-1}_1 \norm{f(H)}^{-1})$. This can be integrated into the framework of \textit{Single-Ancilla LCU} in the following way:~\\
\begin{itemize}
\item[(i)]~\textbf{Single-Ancilla LCU:~} For each coherent run, the crucial task is to sample $V_1$ and $V_2$. We first sample $j$ according to $\{c_j/\norm{c}_1\}_j$, fix $r=\beta^2t^2_j$ and some 
$$
K=O((\log(\beta t_j\norm{h}_1/\varepsilon_H)/\log\log(\beta t_j\norm{h}_1/\varepsilon_H)).
$$ 
Now we can sample a sequence of Pauli matrices and a single Pauli rotation, repeating the sampling procedure $r$ times (as described in Sec.~\ref{sec:ham-sim}). As the $\ell_1$ -norm of the LCU coefficients due to Hamiltonian simulation is $O(1)$, the $\ell_1$-norm of the overall LCU is still $O(~\norm{c}_1)$. Moreover, this requires only a single ancilla qubit.

Then if $\tau_{\max}$ is the maximal evolution time, the gate depth of each coherent run is at most
\begin{equation}
\label{eqsup:circuit-depth-single-ancilla}
O\left(\beta^2\tau_{\max}^2\dfrac{\log\left(\beta \tau_{\max} \norm{h}_1\norm{c}_1 \norm{f(H)}\ell^{-1}_*/\varepsilon\right)}{\log\log\left(\beta \tau_{\max} \norm{h}_1\norm{c}_1 \norm{f(H)}\ell^{-1}_*/\varepsilon\right)}\right)=\widetilde{O}(\beta^2\tau_{\max}^2).
\end{equation}
The number of classical repetitions remain the same as in Theorem \ref{thmsup:single-ancilla-lcu-local-observable}.
~\\

\item[(ii)]~\textbf{$2k$-order Trotter}: In this case, we simply sample $j$ according to $\{c_j/\norm{c}_1\}_j$ and implement $e^{-it_jH}$ to precision $\varepsilon_H$, using $2k$-order Trotter. This boils down to implementing a product of Pauli rotations controlled by the single ancilla qubit. The gate depth for each coherent run is upper bounded by
\begin{equation}
\label{eqsup:circuit-depth-2k-trotter}
O\left(L(\beta \tau_{\max})^{1+\frac{1}{2k}}\left(\dfrac{\norm{h}_1\norm{c}_1\norm{f(H)}}{\varepsilon\ell_*}\right)^{\frac{1}{2k}}\right),
\end{equation}
where $L$ is the total number of terms in the Hamiltonian $H$. The number of classical repetitions remain the same as in Theorem \ref{thmsup:single-ancilla-lcu-local-observable}.
\end{itemize}
~\\
Finally, we can directly use Eq.~\eqref{eqsup:circuit-depth-single-ancilla} and Eq.~\eqref{eqsup:circuit-depth-2k-trotter} into our algorithms after substituting the appropriate values of $\tau_{\max}$, $\norm{f(H)}$, $\ell_*$ and $\norm{c}_1$. From Sec.~\ref{subsec:gsp-single-ancilla}, we have that for ground state property estimation (Sec.~\ref{subsec:gsp-single-ancilla}), $\tau_{\max}=O(\Delta^{-1}\log(~||h||_1\eta^{-1}\varepsilon^{-1}))$, $\norm{f(H)}=1$, $\ell_*\geq \eta^2$ and $\norm{c}_1=O(1)$. On the other hand, from Sec.~\ref{subsec:single-ancilla-qls}, for quantum linear systems we have, $\tau_{\max}=O(\kappa\log(\kappa\norm{h}_1/\varepsilon))$, $\norm{f(H)}\leq \kappa$, $\ell_*=1$ and $\norm{c}_1=\widetilde{O}(\kappa)$. The gate depth for each coherent run and the total number of classical repetitions are summarized in Table \ref{table:circuit-depth-single-ancilla-gsp-qls}. The overall gate depth is the product of the gate depth per coherent run and the number of classical repetitions.

\subsubsection{Comparison with other methods}
\label{subsubsec-appendix:comparison-depth}

Let us now compare our algorithms with other methods. As before (Table \ref{table:comparison-gsp} and Table \ref{table:comparison-qls}), we compare with Standard LCU, QSVT, as well as early fault-tolerant quantum algorithms. For Standard LCU, we need to implement multi-qubit controlled Hamiltonian simulation. For this, we assume that \textit{Standard LCU} uses the state-of-the-art algorithm (Hamiltonian simulation by qubitization \cite{low2019hamiltonian}). This requires a block encoding access to $H$. When $H$ is a linear combination of strings of Pauli operators as defined previously, it is possible to construct a $(\beta,O(\log L),0)$-block encoding of $H$, in a gate depth of $O(L)$. The QSVT-based quantum algorithms that we consider also require access to this block encoding. So, this cost will be incorporated into the complexity of the algorithms we analyze next.

\begin{table}[ht!!]
\begin{center}
    \resizebox{\columnwidth}{!}{
    \renewcommand{\arraystretch}{3} 
    \begin{tabular}{|c|c|c|c|c|c|}
    \hline
    Algorithm & Variant & Ancilla & \vtop{\hbox{\strut ~~Gate depth} \hbox{\strut per coherent run}} & Classical repetitions \\ \hline\hline
  \multirow{3}{*}{\vtop{\hbox{\strut ~~~~~~\large{Standard LCU} \cite{ge2019faster}} \hbox{\strut ~} \hbox{\strut (with Hamiltonian simulation} \hbox{\strut  ~~~~~~by qubitization \cite{low2019hamiltonian})}}} & QAA & $O\left(\log L + \log\left(\log\left(\frac{\norm{O}}{\eta\varepsilon}\right)/\Delta\right)\right)$ & $\widetilde{O}\left(\dfrac{\beta L}{\Delta \eta}\right)$ & $O\left(\dfrac{\norm{h}_1^2}{\varepsilon^2}\right)$ \\ \cline{2-5}
    
   &  QAE & $O\left(\log L_O +\log L+\log\left(\log\left(\frac{\norm{O}}{\eta\varepsilon}\right)/\Delta\right)\right)$ & $\widetilde{O}\left(\dfrac{\beta L \norm{h}_1}{\Delta\eta\varepsilon}+\dfrac{\norm{h}_1 L_O}{\varepsilon}\right)$ & $O(1)$ \\ \cline{2-5}

   & Without QAA or QAE & $O\left(\log L+\log\left(\log\left(\frac{\norm{O}}{\eta\varepsilon}\right)/\Delta\right)\right)$ & $\widetilde{O}\left(\dfrac{\beta L}{\Delta}\right)$ & $O\left(\dfrac{\norm{h}_1^2}{\varepsilon^2\eta^2}\right)$ \\ \hline 
   
   \multirow{3}{*}{\large{QSVT \cite{lin2020near}}} & QAA & $O(\log L)$ & $\widetilde{O}\left(\dfrac{\beta L}{\Delta\eta}\right)$ & $O\left(\dfrac{\norm{h}_1^2}{\varepsilon^2}\right)$ \\ \cline{2-5}
    
   & QAE & $O\left(\log L_O +\log L\right)$ & $\widetilde{O}\left(\dfrac{\beta L\norm{h}_1}{\Delta\eta\varepsilon}+\dfrac{\norm{h}_1 L_O}{\varepsilon}\right)$ & $O(1)$ \\ \cline{2-5}
   
   & Without QAA or QAE & $O\left(\log L\right)$ & $\widetilde{O}\left(\dfrac{\beta L}{\Delta}\right)$ & $O\left(\dfrac{\norm{h}_1^2}{\varepsilon^2\eta^2}\right)$ \\ \hline 
   
 \multirow{2}{*}{\large{This work}} & \vtop{\hbox{\strut ~~~Ham. Sim. by} \hbox{~\strut Single-Ancilla LCU}} & $1$ & $\widetilde{O}\left(\dfrac{\beta^2}{\Delta^2}\right)$ & $O\left(\dfrac{\norm{h}_1^2}{\varepsilon^2\eta^4}\right)$ \\ \cline{2-5}
 
 & \vtop{\hbox{ \strut ~~Ham. Sim. by} \hbox{~~\strut $2k$-order Trotter}} & $1$ & $\widetilde{O}\left(L\left(\dfrac{\beta}{\Delta}\right)^{1+\frac{1}{2k}}\left(\dfrac{\norm{h}_1}{\varepsilon\eta^2}\right)^{\frac{1}{2k}}\right)$ & $O\left(\dfrac{\norm{h}_1^2}{\eta^4\varepsilon^2}\right)$\\ 
 \hline
    \end{tabular}}
  \caption{Comparison of the ground state property estimation algorithm by Single-Ancilla LCU with other methods in terms of gate depth. We assume that $H$ is a linear combination of $L$ terms (strings of Pauli operators), i.e\ $H=\sum_{k=1}^{L}\lambda_k P_k$ such that $\beta=\sum_{k}|\lambda_k|$. The algorithms in the table estimate $\braket{v_0|O|v_0}$ to $\varepsilon$-additive accuracy, where $\ket{v_0}$ is the unknown ground state of $H$, and $O$ is an observable which is also a linear combination of (easy to implement) unitaries, i.e.\ $\sum_{j=0}^{L_O}h_j O_j$, with $\norm{O_j}=1$ and $\norm{h}_1=\sum_{j=1}^{L_O} |h_j|$. The overall gate depth is the product of the complexities in the last two columns. We assume (i) access to a quantum state with an overlap at least $\eta$ with $\ket{v_0}$, and (ii) that the ground energy of $H$ is known to precision $\varepsilon_g$ (See Sec.~\ref{sec:gsp}). We assume that a $(\beta, \lceil \log L \rceil,0)$-block encoding of $H$ is implementable in gate depth $O(L)$ (Standard LCU and QSVT-based algorithms require access to such a block encoding). Furthermore,  coherent procedures to estimate the desired expectation value require access to a block encoding of $O$. For any $O$ which is an LCU, we can implement a $(~\norm{h}_1, \lceil \log L_O \rceil,0)$-block encoding. The gate depth of this construction is $O(L_O)$.
\label{table:comparison-gsp-depth}}
    \end{center}
\end{table}
\renewcommand{\arraystretch}{1}
~\\
\textbf{Ground state property estimation:~} We compare the complexities of our method with other approaches in Table \ref{table:comparison-gsp-depth}. Let us begin by looking at the complexity of the three ways in which the \textit{Standard LCU} procedure can estimate $\braket{v_0|O|v_0}$:

\begin{itemize}
\item By using Standard LCU to implement controlled Hamiltonian simulation followed by amplitude amplification, the state $\ket{v_0}$ can be prepared. If the Hamiltonian simulation makes use of qubitization as in \cite{low2019hamiltonian}, the overall method requires ancilla qubits to implement (a) the block encoding of $H$, and (b) the linear combination of Hamiltonian simulation. Overall, $O(\log L + \log(\log(\norm{h}_1\eta^{-1}\varepsilon^{-1})/\Delta)$ ancilla qubits are needed, along with multi-qubit controlled operations. The gate depth of the circuit that prepares $\ket{v_0}$ is $\widetilde{O}(\beta L \Delta^{-1}\eta^{-1})$. Following this, $\braket{v_0|O|v_0}$ can be measured by simply sampling $O_j$ according to $h_j/\norm{h}_1$ in each run, and measuring $O_j$. This requires $O(~\norm{h}^2_1/\varepsilon^2)$ classical repetitions. Thus, as compared to our method (Table \ref{table:circuit-depth-single-ancilla-gsp-qls}), the gate depth per coherent run has a better dependence on $\beta$ and $\Delta$, but also depends on $L$ and $1/\eta$. Clearly, there are regimes where our method (using Hamiltonian simulation by \textit{Single-Ancilla LCU}) has a shorter gate depth per coherent run as compared to this approach. For instance, for any Hamiltonian satisfying $L\geq \beta/\Delta$, our method has a shorter gate depth by a factor of at least $O(1/\eta)$. On the other hand, the overall gate depth of this approach has a better dependence on $1/\eta$ and $1/\Delta$. However, our method requires a solitary ancilla qubit and no multi-qubit controlled operations.
~\\
\item It is possible to directly estimate $\braket{v_0|O|v_0}$ by using standard LCU with quantum amplitude estimation (QAE). For any observable which is a linear combination of $L_O$ terms, we can obtain a $(~\norm{h}_1, \log L_O, 0)$-block encoding of $O$ in circuit depth $O(L_O)$. So, overall the total number of ancilla qubits increases to $O(\log L+\log L_O + \log(\log(\norm{h}_1\eta^{-1}\varepsilon^{-1})/\Delta))$. On the other hand, the gate depth of the procedure is
$$
\widetilde{O}\left(\dfrac{\beta L\norm{h}_1}{\Delta\eta\varepsilon}+\dfrac{\norm{h}_1 L_O}{\varepsilon}\right).
$$ 
Only $O(1)$ classical repetitions are needed. As compared to our algorithm, the overall complexity is lower, at the cost of exponentially increasing the gate depth (in terms of $1/\varepsilon$). Moreover, the total number of ancilla qubits required also increases substantially for this strategy.
~\\
\item Without using quantum amplitude amplification or estimation, $\braket{v_0|O|v_0}$ can be measured by simply using standard LCU  followed by repeatedly measuring some $O_j$ sampled according to $h_j/\norm{h}_1$. For each coherent run, the gate depth is $\widetilde{O}\left(\beta L/\Delta\right)$, which has a better dependence on both $\beta$ and $\Delta$, as compared to our method primarily due to the advantage of using the state-of-the-art Hamiltonian simulation procedure. However, the dependence on $L$ ensures that our method (using Hamiltonian simulation by \textit{Single-Ancilla LCU}) has a shorter gate depth per coherent run, for any Hamiltonian where $L> \beta/\Delta$. This advantage can be observed in several quantum chemistry Hamiltonians where typically $\beta\ll L$ \cite{campbell2019random, su2021fault}, as well as Hamiltonians in condensed matter physics such as quantum Ising Hamiltonians with long-range interactions \cite{sachdev2011quantum}. However, the number of classical repetitions needed is quadratically better dependence on $1/\eta$ as compared to our method, which also means that the overall gate depth of this approach is generally lower. As before, our method has a better scaling in terms of the number of ancilla qubits, and the fact that our method requires no multi-qubit controlled operation.
\end{itemize}

The complexity of ground state preparation by QSVT \cite{lin2020near} also compares similarly to our method. All the three variants require more ancilla qubits and multi-qubit controlled operations as compared to our method (using Hamiltonian simulation by \textit{Single-Ancilla LCU}), while the overall complexity is lower. Moreover, as compared to each of the three variants, there are regimes where the gate depth per coherent run of our method is lower, despite requiring fewer ancilla qubits and no multi-qubit controlled gates. The details can be found in Table \ref{table:comparison-gsp-depth}.

The early fault-tolerant quantum algorithm of Dong et al.~\cite{dong2022ground} for ground state preparation can be leveraged to estimate $\braket{v_0|O|v_0}$. However, there are issues if one wants to incorporate the \textit{Single-Ancilla LCU} Hamiltonian simulation technique into their method. This is because, the algorithm requires querying $U=e^{-iH}$, which needs to be implemented without any subnormalization factor. However, our Hamiltonian simulation algorithm implements an LCU $S$ such that $S/\norm{c}_1 \approx e^{-iH}$, where $\norm{c}_1=O(1)$. Thus, several queries to $U$ and $U^{\dag}$ would exponentially blow up the simulation cost ($d$ queries would lead to an effective overhead of $\norm{c}^{d}_1$). Consequently, only Trotter based methods can be suitably incorporated into this algorithm, which has been already analyzed in the article. This also means that the gate depth per coherent run of our approach (using where we use Hamiltonian simulation by \textit{Single-Ancilla LCU}) is sub-exponentially better (in terms of $1/\varepsilon$). Additionally the gate depth also depends on $L$, which can be significantly larger than $\beta$ in many cases. However, the overall gate depth has a better dependence on our method in terms of $1/\eta$.
~\\~\\
\textbf{Quantum Linear Systems:~} The detailed complexities have been summarized in Table \ref{table:comparison-qls-depth}. As with ground state property estimation, we assume that the quantum algorithm by Childs et al. \cite{childs2017quantum}, relying on \textit{Standard LCU}, makes use of the state-of-the-art Hamiltonian simulation procedure by Low and Chuang \cite{low2019hamiltonian}. This reduces the overall gate depth as compared to our method but increases the number of ancilla qubits and multi-qubit controlled operations.

\begin{table}[ht!!]
\begin{center}
    \resizebox{\columnwidth}{!}{
    \renewcommand{\arraystretch}{3} 
    \begin{tabular}{|c|c|c|c|c|c|}
    \hline
    Algorithm & Variant & Ancilla & Gate depth per coherent run & Classical repetitions \\ \hline\hline
  \multirow{3}{*}{\vtop{\hbox{\strut ~~~~~~\large{Standard LCU} \cite{childs2017quantum}} \hbox{\strut ~} \hbox{\strut (with Hamiltonian simulation} \hbox{\strut  ~~~~~~by qubitization \cite{low2019hamiltonian})}}} & QAA & $O\left(\log L + \log\left(\frac{\kappa~\norm{O}}{\varepsilon}\right)\right)$ & $\widetilde{O}\left(\beta L \kappa^2  \right)$ & $O\left(\dfrac{\norm{h}_1^2}{\varepsilon^2}\right)$ \\ \cline{2-5}
    
   &  QAE & $O\left(\log L_O +\log L+\log\left(\frac{\kappa~\norm{O}}{\varepsilon}\right)\right)$ & $\widetilde{O}\left(\dfrac{\beta L \norm{h}_1 \kappa^2}{\varepsilon}+\dfrac{L_O \norm{h}_1}{\varepsilon}\right)$ & $O(1)$ \\ \cline{2-5}

   & Without QAA or QAE & $O\left(\log L+\log\left(\frac{\kappa~\norm{O}}{\varepsilon}\right)\right)$ & $\widetilde{O}\left(\beta L \kappa\right)$ & $\widetilde{O}\left(\dfrac{\kappa^2\norm{h}_1^2}{\varepsilon^2}\right)$ \\ \hline 
   
   \multirow{3}{*}{\large{QSVT} \cite{gilyen2019quantum}} & QAA & $O(\log L)$ & $\widetilde{O}\left(\beta L \kappa^2 \right)$ & $O\left(\dfrac{\norm{h}_1^2}{\varepsilon^2}\right)$ \\ \cline{2-5}
    
   & QAE & $O\left(\log L_O +\log L\right)$ & $\widetilde{O}\left(\dfrac{\beta L \norm{h}_1 \kappa^2 }{\varepsilon}+\dfrac{\norm{h}_1 L_O}{\varepsilon}\right)$ & $O(1)$ \\ \cline{2-5}
   
   & Without QAA or QAE & $O\left(\log L\right)$ & $\widetilde{O}\left(\beta L \kappa\right)$ & $\widetilde{O}\left(\dfrac{\kappa^2\norm{h}_1^2}{\varepsilon^2}\right)$ \\ \hline 
   
   \multirow{2}{*}{\vtop{\hbox{\strut \large{Discrete adiabatic}} \hbox{\strut \large{~~~~theorem \cite{costa2022optimal}}}}} & Classical repetitions & $O(\log L)$ & $\widetilde{O}\left(\beta L \kappa\right)$ & $O\left(\dfrac{\norm{h}_1^2}{\varepsilon^2}\right)$ \\ \cline{2-5}
 
  & QAE & $O\left(\log L + \log L_O \right)$ & $\widetilde {O}\left(\dfrac{\beta L \norm{h}_1 \kappa }{\varepsilon}+\dfrac{\norm{h}_1 L_O}{\varepsilon}\right)$ & $O(1)$ \\ \hline
   
 \multirow{2}{*}{\large{This work}} & \vtop{\hbox{\strut ~~~Ham. Sim. by} \hbox{~\strut Single-Ancilla LCU}} & $1$ & $\widetilde{O}\left(\beta^2\kappa^2\right)$ & $\widetilde{O}\left(\dfrac{\norm{h}_1^2\kappa^4}{\varepsilon^2}\right)$ \\ \cline{2-5}
 
 & \vtop{\hbox{ \strut ~~Ham. Sim. by} \hbox{~~\strut $2k$-order Trotter}} & $1$ & $\widetilde{O}\left(L\kappa^{1+\frac{3}{2k}}\left(\beta \right)^{1+\frac{1}{2k}}\left(\dfrac{\norm{h}_1}{\varepsilon}\right)^{\frac{1}{2k}}\right)$ & $\widetilde{O}\left(\dfrac{\kappa^4\norm{h}_1^2}{\varepsilon^2}\right)$ \\  
 \hline
    \end{tabular}}
  \caption{Comparison of quantum algorithms to estimate expectation values of observables with respect to the solution of quantum linear systems. We compare the complexity of the quantum algorithm by \textit{Single-Ancilla LCU} with other methods. We assume that $H$ is a linear combination of $L$ terms (strings of Pauli operators), i.e\ $H=\sum_{k=1}^{L}\lambda_k P_k$ such that $\beta=\sum_{k}|\lambda_k|$ and the eigenvalues of $H$ lie in $[-1,-1/\kappa]\cup [1/\kappa, 1]$. We also assume that an initial state $\ket{b}$ can be prepared efficiently. The algorithms in the table estimate $\braket{x|O|x}$ to $\varepsilon$-additive accuracy, where $\ket{x}=H^{-1}\ket{b}/\norm{H^{-1}\ket{b}}$. $O$ is an observable which is also a linear combination of (easy to implement) unitaries, i.e.\ $\sum_{j=0}^{L_O}h_j O_j$, with $\norm{O_j}=1$ and $\norm{h}_1=\sum_{j=1}^{L_O} |h_j|$. The overall circuit depth is the product of the complexities in the last two columns. We assume that a $(\beta, \lceil \log L \rceil,0)$-block encoding of $H$ is implementable in gate depth $O(L)$ (Standard LCU and QSVT-based algorithms require access to such a block encoding). Furthermore,  coherent procedures to estimate the desired expectation value require access to a block encoding of $O$. For any $O$ which is an LCU, we can implement a $(~\norm{h}_1, \lceil \log L_O \rceil,0)$-block encoding. The circuit depth of this construction is $L_O$.
\label{table:comparison-qls-depth}}
    \end{center}
\end{table}
\renewcommand{\arraystretch}{1}

\begin{itemize}
\item Standard LCU, along with controlled Hamiltonian simulation followed by amplitude amplification, prepares the state $\ket{x}$. This requires ancilla qubits to implement (a) the block encoding of $H$, and (b) the linear combination of Hamiltonian simulation. Overall, $O(\log L + \log\log(\kappa\norm{h}_1/\varepsilon))$ ancilla qubits are needed, along with multi-qubit controlled operations. Using Hamiltonian simulation by qubitization results in the gate depth per coherent scaling as $\widetilde{O}(\beta L \kappa^2)$. Following this, $\braket{x|O|x}$ can be measured by sampling $O_j$ according to $h_j/\norm{h}_1$ in each run, and measuring $O_j$. This requires $O(~\norm{h}^2_1/\varepsilon^2)$ classical repetitions. As compared to our method (Table \ref{table:circuit-depth-single-ancilla-gsp-qls}), the gate depth per coherent run has a better dependence on $\beta$, primarily due to the use of a more advanced simulation algorithm. However, despite this, the dependence on $L$ ensures that for any $H$ with $L>\beta\kappa$, our method (using Hamiltonian simulation by \textit{Single-Ancilla LCU}) has a shorter gate depth per coherent run. The overall gate depth of this approach has a better dependence on $\kappa$ (by a factor of $\kappa^4$), but the dependence on $L$ means that there are regimes of $L$, $\beta$ and $\kappa$ for which our method also has a shorter overall gate depth. Note that our method requires only a single ancilla qubit and no multi-qubit controlled operations.
~\\
\item The coherent estimation of $\braket{x|O|x}$ by using quantum amplitude estimation requires accessing a block encoding of $O$, as defined previously. So, the total number of ancilla qubits increases to $O(\log L+ \log(\kappa \norm{h}_1/\varepsilon)+\log L_O)$. On the other hand, the gate depth of the procedure is
$$
\widetilde{O}\left(\dfrac{\beta\norm{h}_1 \kappa^2 L}{\varepsilon}+\dfrac{\norm{h}_1 L_O}{\varepsilon}\right).
$$ 
Only $O(1)$ classical repetitions are needed. So this is also the overall gate depth. As compared to our algorithm, the overall complexity is lower, at the cost of exponentially increasing the gate depth per coherent run (in terms of $1/\varepsilon$). Moreover, the total number of ancilla qubits required is quite high.
~\\
\item Estimating $\braket{x|O|x}$ without using quantum amplitude amplification or estimation a requires a gate depth per coherent run of $\widetilde{O}\left(\kappa \beta  L \right)$, which has a better dependence on both $\beta$ and $\kappa$, as compared to our method primarily due to the advantage of using the state-of-the-art Hamiltonian simulation procedure. However, our method requires a shorter gate depth per coherent run for $H$ satisfying $L>\beta\kappa$. The number of classical repetitions needed, given by $O(~\norm{h}^2_1\kappa^2/\varepsilon^2)$ has a quadratically better dependence on $\kappa$ as compared to our method (which also means that the overall gate depth is lower in general). However, this procedure requires $O(\log L)$ ancilla qubits and multi-qubit controlled gates.
\end{itemize}

The quantum linear systems algorithm by QSVT \cite{gilyen2019quantum} can also estimate $\braket{x|O|x}$ in three ways. As before, all three variants always require more ancilla qubits and multi-qubit controlled operations, as compared to our method. Despite this, there are regimes where our method (using Hamiltonian simulation by \textit{Single-Ancilla LCU}) has a shorter gate depth per coherent run. The overall complexity of QSVT-based approaches is however, lower. The details can be found in Table \ref{table:comparison-qls-depth}.

The state-of-the-art quantum linear systems algorithm \cite{costa2022optimal} requires a circuit depth per coherent run that is linear in $\kappa$ while $O(~||h||^2_1/\varepsilon^2)$ classical repetitions are needed. However, this approach still requires access to a block encoding of $H$, which requires $O(\log L)$ ancilla qubits. This adds an overhead of $O(L)$ to the gate depth per coherent run. Consequently, there are regimes where our method has an advantage. By using QAE, the ancilla qubit overhead is higher. The gate depth per run is exponentially worse than our method in terms of $1/\varepsilon$, but is quadratically better in terms of $\beta$ and $\kappa$. For both these approaches to estimate the desired expectation value, the overall gate depth has a better dependence on $\kappa$ as compared to our method.

Overall, our algorithm provides a generic exponential speedup over the best known classical algorithms. Given
recent dequantization algorithms \cite{chia2020quantum, tang2021quantum, gilyen2022improved, shao2022faster}, the speedup is polynomial for any $H$ that is low-rank.
\section{Ground state preparation using QSVT on fully fault-tolerant quantum computers} 
\label{subsec:gsp-qsvt}
In this section, we provide a quantum algorithm for the GSP problem for fully fault-tolerant quantum computers. The key idea is to implement the function $f(H)=e^{-tH^2}$ in the circuit model. A straightforward approach would be to use the decomposition of $f(H)$ in Lemma \ref{lem:lcu-decomp-exp} and implement a standard LCU procedure. However, a more efficient approach would be to implement a polynomial approximation of $f(H)$ using QSVT. We begin by providing a polynomial approximation of the Gaussian operator $e^{-tx^2}$. 
~\\
\begin{restatable}[Polynomial approximation to $e^{-tx^2}$]{lemma}{lempolyapproxgaussian}
\label{lem:polynomial-gaussian-operator}
Suppose $x\in [-1,1]$, $\varepsilon\in [0,1/2)$ and $t\in \mathbb{R}^+$. Furthermore, suppose $d=\lceil\max\{te^2/2,\ln(2/\varepsilon)\}\rceil$. Then, there exists a polynomial $\tilde{q}_{t,d,d'}(x)$ of degree 
$$
d'=\lceil\sqrt{2d\ln(4/\varepsilon)}\rceil\in O\left(\sqrt{t}\log(1/\varepsilon)\right),
$$ 
for which the following holds
$$
\sup_{x\in [-1,1]}\abs{e^{-tx^2}-\tilde{q}_{t,d,d'}(x)}\leq \varepsilon.
$$ 
\end{restatable}
\begin{proof}
For $x\in [-1,1]$, we assign $z=1 - 2 x^2$. Note that as $x^2\in [0,1]$, we have $z\in [-1,1]$. Moreover, following this substitution, we need a polynomial approximation to $e^{-\frac{t}{2}(1-z)}$. Now, we define the polynomial $\tilde{q}_{t,d,d'}(x)= q_{\frac{t}{2},d,d'}(1-2x^2)$. The degree of $\tilde{q}_{t,d,d'}(x)$ is the same as that of $q_{\frac{t}{2},d,d'}(1 - 2x^2)$, which is $d'$. So, we have to bound
$$
\sup_{x\in [-1,1]}\abs{e^{-tx^2}-\tilde{q}_{t,d,d'}(x)}= \sup_{z\in [-1,1]}\abs{e^{-\frac{t}{2}(1-z)}-q_{\frac{t}{2},d,d'}(z)}.
$$
Now from Eq.~\eqref{eq:exponential-poly-approx},
$$
\sup_{z\in [-1,1]}\abs{e^{-\frac{t}{2}(1-z)}-q_{\frac{t}{2},d,d'}(z)}\leq \varepsilon,
$$
for $d=\lceil\max\{te^2/2,\ln(2/\varepsilon)\}\rceil$ and $d'=\lceil\sqrt{2d\ln(4/\varepsilon)}\rceil$.
\end{proof}

From Lemma \ref{lem:polynomial-gaussian-operator}, we can use QSVT to obtain a block encoding of $e^{-tH^2}$, given an approximate block encoding of $H$. Subsequently, we shall show that this results in a robust quantum algorithm for preparing the ground state of $H$ under the assumptions we have considered.
~\\
\begin{lemma}
\label{lem:implementing-gaussian-operator}
Let $H$ be a Hermitian matrix with eigenvalues in $[-1,1]$ and $\varepsilon\in (0,1/2)$. Furthermore, suppose $t\in \mathbb{R}^{+}$ and $U_H$ is an $(1,a,\delta)$-block encoding of $H$, implementable in time $T_H$. Also, let $d=\lceil\max\{te^2/2,\ln(4/\varepsilon)\}\rceil$ and $d'=\lceil\sqrt{2d\ln(8/\varepsilon)}\rceil$. Then, provided
$$
\delta\leq \dfrac{\varepsilon^2}{128 d~\ln(8/\varepsilon)},
$$
we can  implement an $(1,a+1,\varepsilon)$-block encoding of $e^{-t H^2}$ in cost
$$
T=O\left(T_H \sqrt{t} \log(1/\varepsilon)\right).
$$ 
\end{lemma}
\begin{proof}
Now suppose $\tilde{H}=\left(\bra{0}^{\otimes a}\otimes I\right)U_H \left(\ket{0}^{\otimes a}\otimes I\right)$. Then, from the definition of block encoding of operators, $\norm{H-\tilde{H}}\leq \delta$. Also, from Lemma \ref{lem:polynomial-gaussian-operator}, we can use the polynomial of degree $d'=\lceil\sqrt{2d\log(8/\varepsilon)}\rceil$ to implement $(1,a+1,\varepsilon/2)$-block encoding of $\tilde{q}_{t,d,d'}(\tilde{H})$ in cost 
$$
T=d'\in O\left( T_H \sqrt{t} \log(1/\varepsilon)\right).
$$
The number of ancilla qubits increased by one because of the QSVT procedure. So, we have
\begin{align}
\norm{e^{-tH^2}-\tilde{q}_{d,d'}\left(\tilde{H}\right)}
&\leq \norm{e^{-tH^2}-\tilde{q}_{d,d'}\left(\tilde{H}\right)}+\norm{\tilde{q}_{d,d'}\left(\tilde{H}\right)-\tilde{q}_{d,d'}\left(H\right)}\\
&\leq \varepsilon/2+4d'\sqrt{\delta}\\
&\leq \varepsilon/2+\varepsilon/2=\varepsilon.~~~~~~~~~~[\text{~Substituting the value of $\delta$ and $d'$}].
\end{align}
\end{proof}
Now that we have a procedure to implement a block encoding of $e^{-tH^2}$, given an approximate block encoding of $H$, we can use this to obtain a circuit model quantum algorithm for preparing the $0$-eigenstate of $H$. As before, let us make some assumptions on the spectrum of $H$. We assume that we are given a Hamiltonian $H$ of unit norm with ground energy, $\lambda_0$ and we intend to prepare a state that is close to its ground state, $\ket{v_0}$. We assume that the gap between the ground state and the rest of the spectrum is lower bounded by $\Delta$. We also assume that we have knowledge of $E_0$ such that
$$
|\lambda_0-E_0|\leq O\left(\Delta/\sqrt{\log\frac{1}{\eta\varepsilon}}\right).
$$
\begin{lemma}
\label{lem:gsp-discrete}
Let $\varepsilon\in (0,1/2)$ and $H$ be a Hamiltonian with unit spectral norm. Furthermore, suppose we are given $U_H$, which is a $(1,a,\delta)$-block encoding of $H$, implemented in time $T_H$. Let $\ket{v_0}$ be the ground state of $H$ with eigenvalue $\lambda_0$ such that the value of $\lambda_0$ is known up to precision
$\eps_g\in \Oo\left(\Delta/\sqrt{\log\frac{1}{\eta\eps}}\right)$, where $\Delta$ is a lower bound on the spectral gap of $H$. 

Additionally, let us assume access to a state preparation procedure $B$ which prepares a state $\ket{\psi_0}$ in time $T_{\psi_0}$ such that $|\braket{\psi_0|v_0}|\geq \eta$. Also, let
$$
\delta\leq \dfrac{\varepsilon^2\eta^2}{512 d~\ln\left(\dfrac{16}{\eta\varepsilon}\right)},
$$
where, $d=\lceil\max\{te^2/2, \ln(8/\varepsilon)\}\rceil$, and
$$
t>\dfrac{1}{2\Delta^2}\log\left(\dfrac{4(1-\eta^2)}{\eta^4\varepsilon^2}\right).$$ 
Then there exists a quantum algorithm that prepares a quantum state that is $O(\varepsilon)$-close to $\ket{v_0}$ with $\Omega(1)$ probability in cost
\begin{equation}
\label{eq:complexity-gsp}
T=O\left(\dfrac{T_H}{\eta\Delta}\log\left(\dfrac{1}{\eta\varepsilon}\right)+\dfrac{T_{\psi_0}}{\eta}\right).
\end{equation}
\end{lemma}
\begin{proof}
In lemma \ref{lem:implementing-gaussian-operator}, we replace $\varepsilon$ with $\varepsilon\eta/2$ to prepare an $(1,a+1,\varepsilon\eta/2)$-block encoding of $e^{-tH^2}$. Furthermore, we choose
\begin{equation}
\label{eq:value-of-t}
t\geq \dfrac{1}{2\Delta^2}\log\left(\dfrac{4(1-\eta^2)}{\eta^4\varepsilon^2}\right)=O\left(\dfrac{1}{\Delta^2}\log\left(\dfrac{1}{\eta\varepsilon}\right)\right).
\end{equation}
To get an $\varepsilon\eta/2$-precision in the block encoding the degree of the polynomial $\tilde{q}_{t,d,d'}(H')$ is 
$$
d'=\left\lceil\sqrt{2d\ln\left(\frac{16}{\eta\varepsilon}\right)}\right\rceil,
$$ 
where $d=\left\lceil\max\{te^2/2, \ln\left(\frac{8}{\varepsilon\eta}\right)\}\right\rceil$.
This yields that the precision of block encoding of $H$ needs to be at least $\delta$-precise where,
$$
\delta\leq \dfrac{\varepsilon^2\eta^2}{512d\ln\left(\dfrac{16}{\eta\varepsilon}\right)}.
$$
Thus, with cost,
$$
O\left(\dfrac{T_H}{\Delta}\log\left(\dfrac{1}{\eta\varepsilon}\right)+T_{\psi_0}\right),
$$
we prepare a quantum state that is $O\left(\varepsilon\eta/2\right)$-close to
$$
\ket{\eta_t}=\ket{\bar{0}}\dfrac{e^{-tH^2}}{\sqrt{\braket{\psi_0|e^{-2tH^2}|\psi_0}}}\ket{\psi_0}+\ket{\Phi^\perp}.
$$
Post-selected on obtaining $\ket{\bar{0}}$ in the first register, we obtain a quantum state that is $O(\varepsilon\eta/2)$-close to 
\begin{equation}
\label{eq:state-phi}
\ket{\phi}=\dfrac{e^{-tH^2}}{\sqrt{\braket{\psi_0|e^{-2tH^2}|\psi_0}}}\ket{\psi_0},
\end{equation}
with amplitude $\sqrt{\braket{\psi_0|e^{-2tH^2}|\psi_0}}=\Omega(\eta)$. Now by choosing $t$ as in Eq.~\eqref{eq:value-of-t}, we have
$$
\norm{\ket{v_0}-\ket{\phi}}\leq O(\varepsilon\eta/2).
$$
By the triangle inequality, this implies that the quantum state prepared is $O(\varepsilon\eta)$-close to $\ket{v_0}$ with probability $\Omega(\eta)$. So by using $O(1/\eta)$-rounds of amplitude amplification, we obtain a quantum state that is $O(\varepsilon)$-close to $\ket{v_0}$ with probability $\Omega(1)$. The overall cost will be
$$
T=O\left(\dfrac{T_H}{\eta\Delta}\log\left(\dfrac{1}{\eta\varepsilon}\right)+\dfrac{T_{\psi_0}}{\eta}\right).
$$
\end{proof}
\section{Relationship between discrete-time and continuous-time quantum walks}
\label{sec:relationship-ctqw-dtqw}

The \textit{Ancilla-free LCU} framework helped us relate between discrete and continuous-time quantum walks and their classical counterparts. Here, using block encoding and QSVT, we establish a relationship between discrete-time quantum walks and continuous-time quantum walks, from both directions. In a seminal work, Childs \cite{childs2010relationship} showed that, given any Hamiltonian $H$, one can implement $e^{-iHt}$ using a discrete-time quantum walk. This generates a continuous-time quantum walk on the vertices of the underlying Markov chain $P$. However, such a continuous-time quantum walk time evolution operator cannot be leveraged to fast-forward continuous-time random walk dynamics. Here we show that there exists a Hamiltonian $H_P$ (defined on the edges of $P$) is able to achieve this. For this, we show that the block encoding of $H_P$ can be efficiently constructed from a block encoding of $H$.  

The problem of obtaining a discrete-time quantum walk, given access to a continuous-time quantum walk has not been addressed in the literature. In fact, there has been very little progress towards answering this question. In this section, we show that one can establish a relationship in this direction by minor modifications to existing theorems in Ref.~\cite{gilyen2018quantum}. 

In order to obtain a continuous-time quantum walk from a discrete-time quantum walk, we show that given $U_H$, a block encoding of any $H$ such that $U^2_H=I$, we can obtain a Hamiltonian $H_P$, such that $H^2_P$ is a block encoding of $I-H^2$. Thus, when $H=D$, simulating $H_P$ allows us to obtain a continuous-time quantum walk, on the edges of $P$, from a block encoding of $H$.

For the other direction, that is, to obtain a discrete-time quantum walk from a continuous-time quantum walk, we will assume that we have access to some $U=e^{iH}$. This corresponds to a continuous-time quantum walk with respect to the Hamiltonian $H$. From this, using Corollary 71 of Ref.~\cite{gilyen2018quantum}, we will show that we can obtain a block encoding of $H$. Finally, we show that any block encoding of $H$ can lead to a discrete-time quantum walk on the edges of $H$. We discuss each of these approaches next.
~\\~\\
\textbf{From discrete-time quantum walks to continuous-time quantum walks:~} We first show that given any block encoding of $H$, we can obtain a Hamiltonian that block-encodes $I-H^2$.
~\\
\begin{lemma}
\label{lem:dtqw-to-ctqw}
Suppose $\varepsilon\in (0,1)$, $\Pi_0=\ket{\bar{0}}\bra{\bar{0}}\otimes I$ and $R=2\Pi_0-I\otimes I$. Let $U_H$ be any $(1,a,0)$-block encoding of $H$ such that $U^2_H=I$. Then the Hamiltonian,
\begin{equation}
\label{eq:somma-ortiz-general}
H_P=i[U_H,\Pi_0]
\end{equation}
can be constructed from one query to the (controlled) discrete-time quantum walk unitary $V=R\cdot U_H$ and its conjugate transpose. Furthermore, $H^2_P$ is a $(1,a+1,\varepsilon)$ block encoding of $I-H^2$.
\end{lemma}
\begin{proof}
It is easy to see that $H_P=i(V-V^\dag)/2$. So if $W_V=\ket{0}\bra{0}\otimes e^{i\pi/2}V+\ket{1}\bra{1} \otimes e^{-i\pi/2}V^\dag$, then  
$Q=\left(H\otimes I\right)W_V\left(H\otimes I\right)(\sigma_x\otimes I)$ is a $(1,a+1, 0)$-block encoding of $H_P$. $Q$ is implemented by versions of $V$ and $V^\dag$ and also single Hadamard gates.

It is easy to verify that $H_P$ is a Hamiltonian (Hermitian operator) of unit norm. To prove that $H^2_P$ is a $(1,a,0)$ block encoding of $I-H^2$, observe
\begin{align}
\left(\bra{\bar{0}}\otimes I\right)H_P^2\left(\ket{\bar{0}}\otimes I\right)&=\left(\bra{\bar{0}}\otimes I\right)\left[\Pi_0+U\Pi_0 U-U\Pi_0 U\Pi_0-\Pi_0 U\Pi_0 U\right]\left(\ket{\bar{0}}\otimes I\right)\\
       &=I+H^2-2H^2=I-H^2.
\end{align}
\end{proof}

From a $(1,a+1,\varepsilon)$ block encoding of $H_P$, using QSVT, we can implement a $(1,a+3,\varepsilon)$ block encoding of $e^{-itH_P}$ using $\Theta(t+\log(1/\varepsilon))$ queries to the controlled versions of the DTQW unitary $V$ and its conjugate transpose \cite{low2019hamiltonian,gilyen2019quantum}.
~\\~\\
This implies, from a block encoding of $H$, we can simulate a continuous-time quantum walk on the vertices of $H$ (by implementing $e^{-iHt}$) as well as on the edges of $H$ (by implementing $e^{-iH_Pt}$), requiring in both cases, $\Theta\left(t+\log(1/\varepsilon\right))$ queries to the corresponding discrete-time quantum walk unitary. 
~\\~\\
\textbf{From continuous-time quantum walks to discrete-time quantum walks:~} In this section, we begin by assuming that we have access to a continuous-time quantum walk evolution operator $U=e^{iH}$. The goal would be to construct a discrete-time quantum walk, given access to $U$. For this, first we shall show that from $U$, we can obtain a block encoding of $H$ with unit sub-normalization. A good starting point is Corollary 71 of Ref.~\cite{gilyen2018quantum}, which shows that it is possible to obtain a block encoding of $H$ given access to $U$, provided $\norm{H}\leq 1/2$. We restate the same here.  
~\\
\begin{lemma}[Corollary 71 of \cite{gilyen2018quantum}]
\label{lem:ctqw-to-dtqw}
Given any $U=e^{iH}$, such that $H$ is some Hamiltonian with $\norm{H}\leq 1/2$. Let $\varepsilon\in (0,1/2]$. Then we can implement a $\left(1,2,\varepsilon\right)$-block encoding of $H$ with cost $O(\log 1/\varepsilon)$.
\end{lemma}
~\\
This Lemma already gives a block encoding of $H$, starting from $U$. However, one issue here is that Lemma \ref{lem:ctqw-to-dtqw} does not work when $\norm{H}=1$. This is because, the polynomial that implements this transformation , only approximates $\arcsin(x)$ in the domain $[-1+\delta, 1-\delta]$, for some $\delta >0$. For discrete-time quantum walks it is important that the sub-normalization factor of the block-encoded matrix is one. To see this, observe that Lemma \ref{lem:ctqw-to-dtqw} is effectively a $(1/2, 2, \varepsilon)$-block encoding of $H/\norm{H}$. Implementing $t$ quantum walk steps would shrink this factor to $2^{-t}$ which is undesirable. Moreover, in the context of quantum fast forwarding, the polynomials $p_{t,d}(x)$ and $q_{t,d,d'}(x)$ approximate $x^t$ and $e^{t(x-1)}$ (respectively) on the entire domain $\mathcal{I}\in[-1,1]$. However, for block-encoded matrices with normalization $\alpha>1$, we would need to approximate these functions in $[-1/\alpha, 1/\alpha]$. Using $p_{t,d}(x/\alpha)$ or $q_{t,d,d'}(x/\alpha)$ would lead to an exponential overhead of $\alpha^t$ in the cost.  

One way to circumvent this problem is to instead consider access to the continuous-time evolution operator $U=e^{iH/2}$, where now $\norm{H}=1$. Using Lemma \ref{lem:ctqw-to-dtqw}, we obtain a $(2,2,\varepsilon/2)$-block encoding of $H$ in cost $O(\log(1/\varepsilon))$. At this stage, we can make use of the procedure of uniform singular value amplification [Theorem 17 of Ref.~\cite{gilyen2019quantum}], which amplifies all the singular values (in our case the eigenvalues) of a block-encoded matrix. This allows us to obtain a $(1,3,\varepsilon)$ block encoding of $H$ as we prove next.
~\\
\begin{theorem}[From continuous-time quantum walks to discrete-time quantum walks]
\label{thm:ctqw-dtqw-main}
Suppose $\varepsilon\in (0,1)$ and $H$ is a Hermitian operator. Suppose we have access to $U=e^{iH/2}$. Then there exists a procedure that implements a $(1,3,\varepsilon)$ - block encoding of $H$ in cost $O\left(\frac{1}{\varepsilon}\log(1/\varepsilon)\right)$.
\end{theorem}
\begin{proof}
From $U$, we obtain $U_H$, which is a $(2,2,\delta)$ - block encoding of $H$ in cost $O(\log(1/\delta))$, using Lemma \ref{lem:ctqw-to-dtqw}, for any $\delta\leq \varepsilon/2$. Then, we use the uniform spectral amplification theorem [Theorem 17 of \cite{gilyen2019quantum}]. In Theorem 17 of \cite{gilyen2019quantum}, set $\gamma=2(1-\varepsilon)$. This gives us a $(1,3,\varepsilon)$ - block encoding of $H$ in cost $O(\frac{1}{\varepsilon}\log(1/\varepsilon))$.
\end{proof}

Thus given access to a continuous-time quantum walk $U=e^{iH/2}$, we can obtain $U_H$, which is a block encoding of $H$. Then, if $U^2_H=I$, following \cite{ambainis2019quadratic}, it is possible to show that if $R=(2\ket{\bar{0}}\bra{\bar{0}}-I)\otimes I$ is a reflection operator, then $V=R.U_H$ can generate a discrete-time quantum walk on the edges of $H$, as $V^{t}$ is a block encoding of $T_{t}(H)$. Thus a discrete-time quantum walk can be implemented given access to the continuous-time quantum walk evolution operator $U=e^{-iH/2}$. However, the condition $U^2_H=I$ need not be satisfied in general. However, it is easy to show that if $U_H$ is any block encoding of $H$, the unitary $V=R.U^{\dag}_H.R.U_H$ is a block encoding of the Chebyshev polynomial $T_{2}(H)$. Then $V^t$ is a block encoding of $T_{2t}(H)$, similar to the standard discrete-time quantum walk operator.

Let us now discuss the issue of fast-forwarding this block encoding of $H$. For this we need to show that given any block encoding of $H$, we can have a quantum walk that fast forwards $H^t$. We shall prove this next. Ours is a slightly more general result as compared to Ref.~\cite{apers2019unified} in that (a) it works for both even and odd $t$, and (b) it is robust: provides the precision with which $U_H$ approximates the block encoding of $H$. The later estimate is crucial for highlighting the limitations of fast-forwarding discrete-time quantum walks when given access to the time evolution ocontinuous-time quantum walk .
~\\ 
\begin{lemma}
\label{lem:block encoding-power-of-Ham}
Suppose $\varepsilon\in (0,1)$ and we have access to $U_H$, which is a $(1,a,\delta)$-block encoding of a Hamiltonian $H$ such that $\norm{H}=1$. Then, provided 
$$
\delta\leq \dfrac{\varepsilon^2}{128 t~\ln(8/\varepsilon)},
$$
for any $t\in\mathbb{N}$, we can implement a $(1,O(a+\log t+\log\log(1/\varepsilon)),O(\varepsilon))$-block encoding of $H^t$ in cost $O\left(\sqrt{t\log(1/\varepsilon)}\right)$.
\end{lemma}
\begin{proof}
We will implement a $(1,a+1,\varepsilon)$ block encoding of $H^t$ by separating out the cases where $t$ is even or odd. When $t$ is even, we implement $H^t$ by approximating it with the polynomial defined in Eq.~\eqref{eq:polynomial-approximating-power}. They are guaranteed to be $\varepsilon$-close following Lemma \ref{lem:polynomial-approximating-power-monomial}. The odd case also follows through via similar arguments.

Let $U_H$ be a $(1,a,0)$-block encoding of $H'$. Then $\norm{H-H'}\leq \delta$. Now the unitary $V=RU_H^\dag R U_H$ is a $(1,a,0)$-block encoding of $T_2(H')$. We will use LCU to implement the polynomial $p_{t,d}(H')$ defined in Eq.~\eqref{eq:polynomial-approximating-power}. The degree of the polynomial is chosen to be $d=\lceil\sqrt{2t\ln(8/\varepsilon)}\rceil$, which ensures (from Lemma \ref{lem:polynomial-approximating-power-monomial}) that $\norm{x^t-p_{t,d}(x)}\leq \varepsilon/4$. Consider the unitary $Q$ such that
\begin{equation}
\label{eq:state-prep-lcu-even-t}
Q\ket{\bar{0}}=\dfrac{1}{\sqrt{\alpha}}\sum_{l=0}^{d/2} \sqrt{c_l}\ket{l},
\end{equation}
where,
\begin{equation}
\label{eq:coefficients-lcu-even-t}
c_l=\begin{cases}
2^{1-t}\binom{t}{\frac{t}{2}+l}, & l>0\\
2^{-t}\binom{t}{t/2}, & l=0,
\end{cases}
\end{equation}
and $\alpha=\norm{c}_1$, where $c=(c_0,\cdots, c_{d/2})$. Also, define the controlled unitary
$$
W=\sum_{j=0}^{d/2}\ket{j}\bra{j}\otimes V^j,
$$
where $V=R U_H^\dag R U_H$. Then, it is easy to see, using LCU that the unitary $\widetilde{W}=(Q^\dag \otimes I)W(Q\otimes I)$ is a $(\alpha, a+\lceil\log_2 d\rceil-1, 0)$-block encoding of $p_{t,d}(H')$. That is, 
\begin{equation}
\label{eq:block encoding-lcu}
\left(\bra{\bar{0}}\otimes I\right)\widetilde{W}\left(\ket{\bar{0}}\otimes I\right)= \dfrac{p_{t,d}(H')}{\alpha},
\end{equation}
where $\alpha$ is obtained by observing that for any $x\in [-1,1]$
\begin{align}
\alpha &=\abs{x^t-\sum_{l=d/2+1}^{t}2^{1-t}\binom{t}{t/2+l}}\\
       &\geq 1-\abs{\sum_{l=d/2+1}^{t} 2^{1-t}\binom{t}{t/2+l}}\\
       &\geq 1-\varepsilon/4~~~~~~~~~~~~~~~~~~~~~~~~~~~~~~~~~~~~~~~~~~~~~~~[~\text{From Lemma \ref{lem:polynomial-approximating-power-monomial}}~].
\end{align}
Now by using triangle inequality we obtain,
\begin{align}
\norm{H^t-p_{t,d}(H')/\alpha}&\leq \norm{H^t-p_{t,d}(H')}+(1-\alpha)\norm{p_{t,d}(H')/\alpha}\\
                              &\leq \varepsilon/4+\norm{H^t-p_{t,d}(H)}+\norm{p_{t,d}(H)-p_{t,d}(H')}\\
                              &\leq \varepsilon/4+\varepsilon/4+ 4d\sqrt{\delta}~~~~~~~~~~~~~~~~~~~~~~~~~~~~~~~~~~~~~~~~~~~~~[\text{~From Lemma~\ref{lem:robustness_of_QSVT}~}]\\
                              &\leq \varepsilon/2+\varepsilon/2~~~~~~~~~~~~~~~~~~~~~~~~~~~~~~~~~~~~~~~~~~~~~~~~~~~~~~~\left[\text{~As $\delta\leq \dfrac{\varepsilon^2}{64 d^2} $~}\right]\\
                              &\leq \varepsilon.
\end{align}
Now for odd $t$, we will modify the quantum walk unitary slightly. Let $t=2k+1$ for some $k\in\mathbb{N}$. Note that 
\begin{align}
\abs{x^{2k+1}-x.p_{2k,d}(x)}\leq \abs{x}.\abs{x^{2k}-p_{2k,d}(x)}\leq \varepsilon.
\end{align}
We already know that there exists a unitary $\widetilde{W}$ which is a $(\alpha, a+\lceil \log_2 d \rceil-1,\varepsilon)$-block encoding of $H^{2k}$, which can be implemented with cost $O(\sqrt{2k\log(1/\varepsilon)})$. Note that $U_H$ is a $(1,a,\delta)$-block encoding of $H=T_1(H)$.  Furthermore, let us define $b=a+\lceil\log_2 d\rceil - 1$ to be the number of ancillary qubits required to implement the unitary $\tilde{W}$. We will show that the unitary $(I_b\otimes U_H)(I_a\otimes \widetilde{W})$ is a $(1,a+b,O(\varepsilon))$ block encoding of $H^{2k+1}$.
\begin{align}
&\norm{H^{2k+1}-(\bra{\bar{0}}^{\otimes a+b}\otimes I)(I_b\otimes U_H)(I_a \otimes \widetilde{W})(\ket{\bar{0}}^{\otimes a+b}\otimes I)}\\
&=\norm{H^{2k+1}-\underbrace{(\bra{\bar{0}}^{\otimes a}\otimes I)U_H(\ket{\bar{0}}^{\otimes a}\otimes I)}_{=H'}\underbrace{(\bra{\bar{0}}^{\otimes b}\otimes I)\widetilde{W}(\ket{\bar{0}}^{\otimes b}\otimes I)}_{=p_{2k,d}(H')/\alpha}}\\
&\leq \norm{\left(H-H'\right)H^{2k}}+\norm{H'}\norm{H^{2k}-p_{2k,d}(H')/\alpha}\\
&\leq \delta+\varepsilon=O(\varepsilon).
\end{align}
\end{proof}

Having proven this, the next question we ask is: Can the block encoding of $H$ obtained from the time evolution operator $U=e^{-iH/2}$ (Theorem \ref{thm:ctqw-dtqw-main}) be used to fast-forward discrete-time random walks? We argue here that some issues still remain. For instance, from Lemma \ref{lem:block encoding-power-of-Ham}, we can see that the precision $\delta$ required in the block encoding of $H$ is $\widetilde{O}(\varepsilon^2/t)$. Theorem \ref{thm:ctqw-dtqw-main} implies that to implement a block encoding of $H$, from $U$ would require $\widetilde{O}(t/\varepsilon^2)$ cost. Thus, the advantage of quantum fast-forwarding would be lost.

In order to avoid this, we need a polynomial of degree $t$ that approximates the monomial $(2x)^t$ in the domain $\mathcal{D}:=\left[-\frac{1}{2}(1-1/t), \frac{1}{2}(1-1/t)\right]$. The existence of such a polynomial $P(x)$ of degree $O(t\log (1/\varepsilon))$ can indeed be guaranteed by Corollary 66 of \cite{gilyen2018quantum}. For this set $f(x)=(2x)^t, x_0=0$, $r=1/2$ and $\delta=1/t$ in the corollary. We have not been able to find an explicit construction of this polynomial and leave it open for future work. Expressing $P(x)$ as a linear combination of Chebyshev polynomials, we would obtain an $\varepsilon$-approximation of $(2x)^t$ in $\mathcal{D}$ (having $\sqrt{t}$ terms). Given access to $U=e^{iH/2}$, we obtain a $(2,2,\varepsilon)$ - block encoding of $H$ using Lemma \ref{lem:ctqw-to-dtqw}. We can then directly apply QSVT directly to implement the polynomial $P(x)$ on this block-encoded Hamiltonian. This would allow us to fast-forward discrete-time quantum walks, starting from continuous-time quantum walks. We leave open the explicit construction of such a polynomial, for future work.
\end{document}